\documentclass[11pt]{article}
\usepackage[colorlinks=true,linkcolor=blue,citecolor=blue]{hyperref}
\usepackage{amsmath,amsthm,amssymb}

\usepackage{paralist}
\usepackage{enumerate}
\usepackage{framed}

\usepackage{tikz}
\usetikzlibrary{arrows,automata}
\usepackage{comment, times}
\usepackage[margin=1in]{geometry}
\usepackage{graphicx,float,wrapfig}
\usepackage{mathrsfs}
\usepackage[usenames,dvipsnames]{pstricks}
\usepackage{epsfig}
\usepackage{pst-grad} % Fdi gradients
\usepackage{pst-plot} % Fdi axes
\usepackage{mathdots}
\usepackage[linesnumbered,boxed,ruled,vlined]{algorithm2e}
\usepackage{mathpazo}
\usepackage{framed}

\everymath{\displaystyle}

\usepackage{bbm}

\usepackage{url}

\usepackage{tabu}

\usepackage{xfrac}

\usepackage[normalem]{ulem}

\newtheorem{theorem}{Theorem}[section]

\newtheorem{lemma}[theorem]{Lemma}

\newtheorem{claim}{Claim}
\newtheorem{prop}[theorem]{Proposition}
\newtheorem{cor}[theorem]{Corollary}
\theoremstyle{definition}
\newtheorem{definition}[theorem]{Definition}
\theoremstyle{plain}
\newtheorem{remark}[theorem]{Remark}

\newcommand{\DP}{DP\xspace}
\newcommand{\DPBase}{\mathrm{DP}}
\newcommand{\typeOfDP}[1]{\DPBase_{\mathrm{ #1}}}

\newcommand{\localDP}{\typeOfDP{local}}
\newcommand{\shuffledDP}{\typeOfDP{shuffle}}
\newcommand{\twopartyDP}{\typeOfDP{two\text{-}party}}

\newenvironment{proofof}[1]{\begin{proof}[Proof of #1]}{\end{proof}}
\newenvironment{proofsketch}{\begin{proof}[Proof Sketch]}{\end{proof}}

\newenvironment{reminder}[1]{\bigskip
	\noindent {\bf #1 (restated) }\em}{\smallskip}

\renewcommand{\epsilon}{\varepsilon}
\newcommand{\Ex}{\operatornamewithlimits{\mathbb{E}}}
\newcommand{\F}{\mathbb{F}}
\newcommand{\poly}{\mathop{\mathrm{poly}}}
\newcommand{\polylog}{\mathop{\mathrm{polylog}}}
\newcommand{\eps}{\epsilon}

\newcommand{\R}{\mathbb{R}}
\newcommand{\N}{\mathbb{N}}
\newcommand{\Z}{\mathbb{Z}}

\renewcommand{\subset}{\subseteq}

\def \ro {\text{r.o.-}}

\newcommand{\cro}[1]{$c$\text{-}\ro}

\newcommand{\Hist}{\mathsf{Hist}}

\newcommand{\selection}{\textsf{\small Selection}\xspace}

\newcommand{\paritylearning}{\textsf{\small ParityLearning}\xspace}

\def\ShowAuthNotes{1}
\ifnum\ShowAuthNotes=1
\newcommand{\authnote}[2]{\ \\ \textcolor{red}{\parbox{0.9\linewidth}{[{\footnotesize {\bf #1:} { {#2}}}]}}\newline}
\else
\newcommand{\authnote}[2]{}
\fi

\def\ShowAuthNotes{1}
\ifnum\ShowAuthNotes=1
\newcommand{\lnote}[1]{\textcolor{red}{[Lijie: #1]}}
\newcommand{\badih}[1]{\textcolor{red}{[Badih: #1]}}
\newcommand{\pasin}[1]{\textcolor{red}{[Pasin: #1]}}
\newcommand{\ravi}[1]{\textcolor{red}{[Ravi: #1]}}
\else
\newcommand{\lnote}[1]{}
\newcommand{\badih}[1]{}
\newcommand{\pasin}[1]{}
\newcommand{\ravi}[1]{}
\fi

\newcommand{\VIEW}{\textsf{VIEW}}

\newcommand{\highlight}[1]{\medskip \noindent \textbf{#1}}

\title{On Distributed Differential Privacy\\ and Counting Distinct Elements}
\author{}

\newcommand{\calT}{\mathcal{T}}
\newcommand{\calX}{\mathcal{X}}
\newcommand{\calM}{\mathcal{M}}
\newcommand{\calU}{\mathcal{U}}
\newcommand{\calD}{\mathcal{D}}
\newcommand{\calB}{\mathcal{B}}
\newcommand{\calF}{\mathcal{F}}
\newcommand{\calE}{\mathcal{E}}
\newcommand{\calS}{\mathcal{S}}
\newcommand{\calZ}{\mathcal{Z}}
\newcommand{\calY}{\mathcal{Y}}
\newcommand{\calC}{\mathcal{C}}
\newcommand{\Rrand}{R^{\sf rand}}
\newcommand{\qD}{q^{\calD}}
\newcommand{\calI}{\mathcal{I}}
\newcommand{\bfW}{\mathbf{W}}

\newcommand{\supp}{\mathrm{supp}}

\newcommand{\DE}{\textsf{\small CountDistinct}\xspace}

\newcommand{\Ber}{\mathsf{Ber}}
\newcommand{\Bin}{\mathsf{Bin}}
\newcommand{\Poi}{\mathsf{Poi}}
\newcommand{\NB}{\mathsf{NB}}

\newcommand{\vPoi}{\vec{\mathsf{Poi}}}

\newcommand{\WT}{\widetilde}
\newcommand{\argmin}{\mathop{\textrm{argmin}}}

\newcommand{\vlambda}{\vec{\lambda}}
\newcommand{\vtheta}{\vec{\theta}}
\newcommand{\vmu}{\vec{\mu}}
\newcommand{\vnu}{\vec{\nu}}
\newcommand{\vm}{\vec{m}}

\newcommand{\vU}{\vec{U}}
\newcommand{\vV}{\vec{V}}
\newcommand{\vz}{\vec{z}}
\newcommand{\vu}{\vec{u}}

\newcommand{\valpha}{\vec{\alpha}}
\newcommand{\vbeta}{\vec{\beta}}
\newcommand{\vmx}{\vec{m}^{*}}
\newcommand{\hist}{\mathrm{hist}}
\newcommand{\calDpb}{\calD_{\sf pub}}

\author{
	Lijie Chen\thanks{Email: {\tt lijieche@mit.edu}.  Most of this work was done at Google Research, Mountain View, CA.}\\
	Massachusetts Institute of Technology \\
	Cambridge, MA.
	\and
	Badih Ghazi\thanks{Email: {\tt badihghazi@gmail.com}}\\
	Google Research \\
	Mountain View, CA.
	\and
	Ravi Kumar\thanks{Email: {\tt ravi.k53@gmail.com}}\\
	Google Research \\
	Mountain View, CA.
	\and
	Pasin Manurangsi\thanks{Email: {\tt pasin@google.com}}  \\
	Google Research \\
	Mountain View, CA.
}

\date{}

\begin{document}
	\maketitle
	\setcounter{page}{1} % Sets counter of page to 1
	\begin{abstract}

	We study the setup where each of $n$ users holds an element from a discrete set, and the goal is to count the number of distinct elements across all users, under the constraint of $(\epsilon,\delta)$-differentially privacy:
	\begin{itemize}
	    \item In the non-interactive \emph{local} setting, we prove that the additive error of any  protocol is $\Omega(n)$ for any constant $\epsilon$ and for any $\delta$ inverse polynomial in $n$. 
	    \item In the \emph{single-message shuffle} setting, we prove a lower bound of $\tilde{\Omega}(n)$ on the error for any constant $\epsilon$ and for some $\delta$ inverse quasi-polynomial in $n$. We do so by building on the moment-matching method from the literature on distribution estimation.
	    \item In the \emph{multi-message shuffle} setting, we give a protocol with at most one message per user in expectation and with an error of $\tilde{O}(\sqrt{n})$ for any constant $\epsilon$ and for any $\delta$ inverse polynomial in $n$. Our protocol is also robustly shuffle private, and our error of $\sqrt{n}$ matches a known lower bound for such protocols.
	\end{itemize}
	Our proof technique relies on a new notion, that we call \emph{dominated protocols}, and which can also be used to obtain the first non-trivial lower bounds against multi-message shuffle protocols for the well-studied problems of selection and learning parity.
	
	Our first lower bound for estimating the number of distinct elements provides the first $\omega(\sqrt{n})$ separation between global sensitivity and error in local differential privacy, thus answering an open question of Vadhan (2017). We also provide a simple construction that gives $\tilde{\Omega}(n)$ separation between global sensitivity and error in \emph{two-party} differential privacy, thereby answering an open question of McGregor et al. (2011). 
	\end{abstract}
	\thispagestyle{empty}
	\addtocounter{page}{-1}
	\newpage    
	
	\newpage
	\tableofcontents
	\newpage
	
\section{Introduction}\label{sec:intro}

Differential privacy (DP) \cite{DworkMNS06,DworkKMMN06} has become a leading framework for private-data analysis, with several recent practical deployments \cite{erlingsson2014rappor,CNET2014Google, greenberg2016apple, dp2017learning, ding2017collecting, abowd2018us}. The most commonly studied DP setting is the so-called central (aka curator) model whereby a single authority (sometimes referred to as the analyst) is trusted with running an algorithm on the raw data of the users and the privacy guarantee applies to the algorithm's output.

The absence, in many scenarios, of a clear trusted authority has motivated the study of \emph{distributed} DP models. The most well-studied such setting is the \emph{local} model \cite{kasiviswanathan2011can} (also \cite{warner1965randomized}), denoted henceforth by $\localDP$, where the privacy guarantee is enforced at each user's output (i.e., the protocol transcript). While an advantage of the local model is its very strong privacy guarantees and minimal trust assumptions, the noise that has to be added can sometimes be quite large. This has stimulated the study of ``intermediate'' models that seek to achieve accuracy close to the central model while relying on more distributed trust assumptions. One such middle-ground is the so-called \emph{shuffle} (aka \emph{anonymous}) model \cite{ishai2006cryptography,bittau2017prochlo,Cheu18,erlingsson2019amplification}, where the users send messages to a \emph{shuffler} who randomly shuffles these messages before sending them to the analyzer; the privacy guarantee is enforced on the shuffled messages (i.e., the input to the analyzer). We study both the local and the shuffle models in this work.

\subsection{Counting Distinct Elements}

A basic function in data analytics is estimating the number of distinct elements in a domain of size $D$ held by a collection of $n$ users, which we denote by $\DE_{n,D}$ (and simply by $\DE_n$ if there is no restriction on the universe size).  Beside its use in
%being one of the most basic functions supported by 
database management systems, it is a well-studied problem in sketching, streaming, and communication complexity (e.g., %\cite{flajolet1985probabilistic,alon1999space,cohen1997size,bar2002counting,durand2003loglog, estan2003bitmap, beyer2007synopses, 
\cite{kane2010optimal,
%\cite{cormode2011algorithms,woodruff2014optimal,
brody2014beyond} and the references therein). In central DP, it can be easily solved with constant error using the Laplace mechanism~\cite{DworkMNS06}; see 
also~\cite{mir2011pan, desfontaines2019cardinality, pagh2020efficient, choi2020differentially}. 

We obtain new results on $(\epsilon, \delta)$-DP protocols for $\DE$ in the local and shuffle settings\footnote{For formal definitions, please refer to Section~\ref{sec:preliminaries}. We remark that, throughout this work, we consider the \emph{non-interactive} local model where all users apply the \emph{same} randomizer (see Definition~\ref{def:dp_local}). We briefly discuss in Section~\ref{sec:conc} possible extensions to interactive local models.}.

\subsubsection{Lower Bounds for Local DP Protocols}

Our first result is a lower bound on the additive error of $\localDP$ protocols\footnote{See Section~\ref{sec:preliminaries} for the the formal (standard) definition of public-coin DP protocols. Note that private-coin protocols are a sub-class of public-coin protocols, so all of our lower bounds apply to private-coin protocols as well.} for counting distinct elements.

\begin{theorem}\label{theo:lowb-distinct-elements}
	For any $\eps = O(1) $, %if $P$ is a public-coin $(\eps,o(1/n))$-$\localDP$ protocol, then it cannot solve $\DE_{n,n}$ with error $o(n)$.
	no public-coin $(\eps,o(1/n))$-$\localDP$ protocol can solve\footnote{Throughout this work, we say that a randomized algorithm solves a problem with error $e$ if with probability 0.99 it incurs error at most $e$.} $\DE_{n,n}$ with error $o(n)$.
\end{theorem}

The lower bound in Theorem~\ref{theo:lowb-distinct-elements} is asymptotically tight\footnote{The trivial algorithm that always outputs $0$ incurs an error $n$.}.
Furthermore, it answers a question of Vadhan \cite[Open Problem 9.6]{vadhan2017complexity}, who asked if there is a function with a gap of $\omega(\sqrt{n})$ between its (global) sensitivity and the smallest achievable error by any $\localDP$ protocol.%
\footnote{
To the best of our knowledge, the largest previously known gap between global sensitivity and error was $O(\sqrt{n})$, which is achieved, e.g., by binary summation~\cite{chan2012optimal}.} 
As the global sensitivity of the number of distinct elements is $1$, Theorem~\ref{theo:lowb-distinct-elements} exhibits a (natural) function for which this gap is as large as $\Omega(n)$. While Theorem~\ref{theo:lowb-distinct-elements} applies to the constant $\epsilon$ regime, it turns out we can prove a lower bound for much less private protocols (i.e., having a much larger $\epsilon$ value) at the cost of polylogarithmic factors in the error:

\begin{theorem}\label{theo:LDP-LB-strong}
%	For some $\eps = \ln(n / \polylog(n))$ 
	For some $\eps = \ln(n) - O(\ln \ln n)$ 
	and $D = \Theta(n/\polylog(n))$, %if $P$ is a public-coin $(\eps, n^{-\omega(1)})$-$\localDP$ protocol, then it cannot solve $\DE_{n,D}$ with error $o(D)$.
	no public-coin $(\eps, n^{-\omega(1)})$-$\localDP$ protocol can solve $\DE_{n,D}$ with error $o(D)$.
\end{theorem}
To prove Theorem~\ref{theo:LDP-LB-strong}, we build on the \emph{moment matching} method from the literature on (non-private) distribution estimation, namely \cite{ValiantV17,WY2019chebyshev},
% \cite{WY16polyapprox,JiaoHW18,WY2019chebyshev,HanLec19},
and tailor it to \DE in the $\localDP$ setting (see Section~\ref{sec:overview-moment-matching} for more details on this connection). The bound on the privacy parameter $\epsilon$ in Theorem~\ref{theo:LDP-LB-strong} turns out to be very close to tight: the error drops quadratically when $\epsilon$ exceeds $\ln{n}$. %(namely, within an additive $O(\ln{\ln{n}})$).
This is shown in the next theorem:
\begin{theorem}\label{th:weak_privacy_local_prot_DE}
	There is a $(\ln(n) + O(1))$-$\localDP$ protocol solving $\DE_{n,n}$ with error $O(\sqrt{n})$.
\end{theorem}

\subsubsection{Lower Bounds for Single-Message Shuffle DP Protocols}

In light of the negative result in Theorem~\ref{theo:LDP-LB-strong}, a natural question is whether $\DE$ can be solved in a weaker distributed DP setting such as the shuffle model. It turns out that this is not possible using any shuffle protocol where each user sends no more than $1$ message (for brevity, we will henceforth denote this class by $\shuffledDP^1$, and more generally denote by $\shuffledDP^k$ the variant where each user can send up to $k$ messages). Note that the class $\shuffledDP^1$ includes any method obtained by taking a $\localDP$ protocol and applying the so-called \emph{amplification by shuffling} results of \cite{erlingsson2019amplification,BalleBGN19}.

% In light of the lower bound in Theorem~\ref{theo:lowb-distinct-elements}, it is natural to ask whether one can achieve an asymptotically non-trivial accuracy for the problem of counting distinct elements in weaker distributed DP models than local. For the basic version of the shuffle model where each user can send at most one message, it turns out that,
In the case where $\epsilon$ is any constant and $\delta$ is inverse quasi-polynomial in $n$, the improvement in the error for $\shuffledDP^1$ protocols compared to $\localDP$ is at most polylogarithmic factors:
\begin{theorem} \label{theo:single-message-LB}
	For all $\eps = O(1)$, there are $\delta = 2^{-\polylog(n)}$ and $D = n/\polylog(n)$ such that no public-coin $(\eps,\delta)$-$\shuffledDP^1$ protocol can solve $\DE_{n,D}$ with %probability $0.99$ and
	error $o(D)$.
\end{theorem}

We note that Theorem~\ref{theo:single-message-LB} essentially answers a more general variant of Vadhan's question: it shows that even for $\shuffledDP^1$ protocols (which include $\localDP$ protocols as a sub-class) the gap between sensitivity and the error can be as large as $\tilde{\Omega}(n)$ .

The proof of Theorem~\ref{theo:single-message-LB} follows by combining Theorem~\ref{theo:LDP-LB-strong} with the following connection between $\localDP$ and $\shuffledDP^1$:

\begin{lemma}\label{lm:SDP-to-LDP}
	For any $\eps = O(1)$ and $\delta \le \delta_0 \le 1/n$, if the randomizer $R$ is $(\eps,\delta)$-$\shuffledDP^1$ on $n$ users, then $R$ is %$\left(\ln(\Theta_{\eps}(n\log \delta_0^{-1}/\log \delta^{-1})) ,\delta_0\right)$-$\localDP$.
	$\left(\ln n - \ln(\Theta_{\eps}(\log \delta_0^{-1}/\log \delta^{-1})) ,\delta_0\right)$-$\localDP$.
\end{lemma}
We remark that Lemma~\ref{lm:SDP-to-LDP} provides a stronger quantitative bound than the qualitatively similar connections in \cite{Cheu18,GhaziGKPV19}; specifically, we obtain the term $\ln(\Theta_{\eps}(\log \delta_0^{-1}/\log \delta^{-1}))$, which was not present in the aforementioned works. This turns out to be crucial for our purposes, as this term gives the $O(\ln \ln n)$ term necessary to apply Theorem~\ref{theo:LDP-LB-strong}. 

\subsubsection{A Communication-Efficient Shuffle DP Protocol}

In contrast with Theorem~\ref{theo:single-message-LB}, Balcer et al. \cite{BCJM20} recently gave a $\shuffledDP$ protocol for $\DE_{n,D}$ with error $O(\sqrt{D})$. Their protocol sends $\Omega(D)$ messages per user. We instead show that an error of $\tilde{O}(\sqrt{D})$ can still be guaranteed with each user sending \emph{in expectation} at most one message each of length $O(\log D)$ bits.

\begin{theorem}\label{theo:upper-bound-DE}
	For all $\eps  \le O(1)$ and $\delta \le 1/n$, there is a  public-coin $(\eps,\delta)$-$\shuffledDP$ protocol that solves $\DE_{n}$ with error $\sqrt{\min(n,D)}\cdot \poly(\log(1/\delta) / \eps)$ %and with probability at least $0.99$, and 
	where the expected number of messages sent by each user is at most one.
\end{theorem}
In the special case where $D = o(n/\poly(\epsilon^{-1}\log(\delta^{-1})))$, we moreover obtain a \emph{private-coin} $\shuffledDP$ protocol achieving the same guarantees as in Theorem~\ref{theo:upper-bound-DE} (see Theorem~\ref{theo:upper-bound-DE-base} for a formal statement).
% We point out that all our lower bounds hold against public-coin protocols, and unless otherwise specified, the protocols in our upper bounds only require private randomness.
Note that Theorem~\ref{theo:upper-bound-DE} is in sharp contrast with the lower bound shown in Theorem~\ref{theo:single-message-LB} for $\shuffledDP^1$ protocols. Indeed, for $\delta$ inverse quasi-polynomial in $n$, the former implies a public-coin protocol with less than one message per-user \emph{in expectation} having error $\tilde{O}(\sqrt{n})$ whereas the latter proves that no such protocol exists, even with error as large as $\tilde{\Omega}(n)$, if we restrict each user to send one message \emph{in the worst case}.

A strengthening of $\shuffledDP$ protocols called \emph{robust $\shuffledDP$ protocols}\footnote{Roughly speaking, they are $\shuffledDP$ protocols whose transcript remains private even if a constant fraction of users drop out from the protocol.} was studied by~\cite{BCJM20}, who proved an $\Omega\left(\sqrt{\min(D,n)}\right)$ lower bound on the error of any protocol solving $\DE_{n,D}$. Our protocols are robust $\shuffledDP$ and, therefore, achieve the optimal error (up to polylogarithmic factors) among all robust $\shuffledDP$ protocols, while only sending at most one message per user in expectation.

\subsection{Dominated Protocols and  Multi-Message Shuffle DP Protocols}

The technique underlying the proof of Theorem~\ref{theo:lowb-distinct-elements} can be extended beyond $\localDP$ protocols for $\DE$. It applies to a broader category of protocols that we call \emph{dominated}, defined as follows: 

\begin{definition}\label{defi:dominated-algo}
	We say that a randomizer $R\colon \calX \to \calM$ is \emph{$(\eps,\delta)$-dominated}, if there exists a distribution $\calD$ on $\calM$ 
	such that for all $x \in \calX$ and all $E \subseteq \calM$, 
	\[
	\Pr[R(x) \in E] \le e^{\eps} \cdot \Pr_{\calD}[E] + \delta.
	\]
	In this case, we also say $R$ is $(\eps,\delta)$-dominated by $\calD$. We define $(\eps,\delta)$-dominated protocols in the same way as $(\eps,\delta)$-$\localDP$, except that we require the randomizer to be $(\eps,\delta)$-dominated instead of being  $(\eps,\delta)$-DP.
\end{definition}
Note that an $(\eps,\delta)$-$\localDP$ randomizer $R$ is $(\eps,\delta)$-dominated: we can fix a $y^* \in \calX$ and take $\calD = R(y^*)$. Therefore, our new definition is a relaxation of $\localDP$.
%$(\eps,\delta)$-dominated protocols is a relaxation of $(\eps,\delta)$ locally DP in the sense that we now only requires an upper bound

We show that \emph{multi-message} $\shuffledDP$ protocols are dominated, which allows us to prove the first non-trivial lower bounds against $\shuffledDP^{O(1)}$ protocols.

Before formally stating this connection, we recall why known lower bounds against $\shuffledDP^1$ protocols \cite{Cheu18,GhaziGKPV19,balcer2019separating} do not extend to $\shuffledDP^{O(1)}$ protocols.\footnote{We remark that~\cite{pure-dp-shuffled} developed a technique for proving lower bounds on the \emph{communication complexity} (i.e., the number of bits sent per user) for multi-message protocols. Their techniques do not apply to our setting as our lower bounds are in terms of the number of messages, and do not put any restriction on the message length. Furthermore, their technique only applies to \emph{pure}-DP where $\delta = 0$, whereas ours applies also to \emph{approximate}-DP where $\delta > 0$.} These prior works use the connection stating that any $(\eps,\delta)$-$\shuffledDP^1$ protocol is also $(\eps + \ln n, \delta)$-$\localDP$~\cite[Theorem~6.2]{Cheu18}. It thus suffices for them to prove lower bounds for $\localDP$ protocols with low privacy requirement (i.e., $(\eps + \ln n,\delta)$-$\localDP$), for which lower bound techniques are known or developed. For $\eps$-$\shuffledDP^1$ protocols,~\cite{balcer2019separating} showed that they are also $\eps$-$\localDP$; therefore, lower bounds on $\localDP$ protocols automatically translate to lower bounds on pure-$\shuffledDP^1$ protocols.
%\cite{Cheu18} combines the reduction from approximate locally DP to pure locally DP~\cite{xxx} and the lower bound for selection in the pure locally DP model~\cite{xxx}. Due to the polynomially loss in the reduction from approximate locally DP to pure locally DP, the only obtained an $n^{1/17}$ lower bound for single-message shuffle protocols.
To apply this proof framework to $\shuffledDP^{O(1)}$ protocols, a natural first step would be to connect $\shuffledDP^{O(1)}$ protocols to $\localDP$ protocols. However, as observed by~\cite[Section~4.1]{balcer2019separating}, there exists an $\eps$-$\shuffledDP^{O(1)}$ protocol that is not $\localDP$ for any privacy parameter. That is, there is no analogous connection between $\localDP$ protocols and multi-message $\shuffledDP$ protocols, even if the latter can only send $O(1)$ messages per user.

% A key notion we introduce is that of \emph{dominated protocols}.

In contrast, the next lemma captures the connection between multi-message $\shuffledDP$ and dominated protocols.

\begin{lemma}\label{lm:SDP-to-dominated-Toy}
	If $R$ is $(\eps,\delta)$-$\shuffledDP^{k}$ on $n$ users, then it is $(\eps + k(1 + \ln n),\delta)$-dominated.
\end{lemma}

By considering dominated protocols and using Lemma~\ref{lm:SDP-to-dominated-Toy}, we obtain the first lower bounds for \emph{multi-message} $\shuffledDP$ protocols for two well-studied problems: \selection and \paritylearning. %(see Section~\ref{subsec:related_work} for more details and related work).

\subsubsection{Lower Bounds for Selection}

The \selection problem on $n$ users is defined as follows. The $i$th user has an input $x_i \in \{0,1\}^D$ and the goal is to output an index $j \in [D]$ such that $\sum_{i=1}^{n} x_{i,j} \ge \left(\max_{j^*} \sum_{i=1}^{n} x_{i,j^*} \right) - n/10$.
%
% \begin{itemize}
% % 	\item \textbf{Distinct Elements ($\DE$).} Each user gets an input from a known universe $\calX$, the goal is to estimate how many distinct inputs are hold among $n$ users.
% % 	We also use $\DE_{n,D}$ to denote the restriction of Distinct Elements in which there are $n$ users each gets an input from $[D]$. If there is no restriction on the size of universe $\calX$, we denote it simply by $\DE_n$.
% \end{itemize}
% \subsection{Lower Bounds and Upper Bounds for Distinct Elements}
% \subsection{Lower Bounds for Selection and Learning Parity}
% \subsubsection*{Approximate-DP Protocols}
% The main result of this note is showing both problems are hard even in $k$-messages shuffle model.
%
\selection is well-studied in DP (e.g., \cite{duchi2013local, steinke2017tight,ullman2018tight}) and its variants are useful primitives for several statistical and algorithmic problems including feature selection, hypothesis testing and clustering.  In central DP, the exponential mechanism of \cite{mcsherry2007mechanism} yields an $\eps$-DP algorithm for \selection when $n = O_\eps(\log{D})$. On the other hand, it is known that any $(\epsilon, \delta)$-$\localDP$ protocol for \selection with $\epsilon = O(1)$ and $\delta = O(1/n^{1.01})$ requires $n = \Omega(D \log{D})$ users \cite{ullman2018tight}.
Moreover, \cite{Cheu18} obtained a $(\eps, 1/n^{O(1)})$-$\shuffledDP^D$ protocol for $n = \tilde{O}_{\eps}(\sqrt{D})$. By contrast, for $\shuffledDP^1$ protocols, a lower bound of $\Omega(D^{1/17})$ was obtained in \cite{Cheu18} and improved to $\Omega(D)$ in \cite{GhaziGKPV19}. 
%To the best of our knowledge, prior to the present work no $\omega(\log{D})$ lower bound was known for $\shuffledDP^k$ protocols for any $k \geq 2$.

The next theorem give a lower bounds for \selection that holds against approximate-$\shuffledDP^{k}$ protocols. To the best of our knowledge, this is the first lower bound even for $k=2$ (and even for the special case of pure protocols, where $\delta = 0$).

\begin{theorem}\label{theo:lowb-selection}
	For any $\eps = O(1)$, %if $P$ is a public-coin $(\eps,o(1/D))$-$\shuffledDP^{k}$ protocol that solves \selection with probability at least $0.99$, then $n \ge \Omega\left(\frac{D}{k}\right)$.
	any public-coin $(\eps,o(1/D))$-$\shuffledDP^{k}$ protocol that solves \selection requires $n \ge \Omega\left(\frac{D}{k}\right)$.
\end{theorem}

We remark that combining the advanced composition theorem for DP and known $\shuffledDP$ aggregation algorithms, one can obtain a $(\eps,1/\poly(n))$-$\shuffledDP^{k}$ protocol for \selection with $\tilde{O}(D / \sqrt{k})$ samples for any $k \leq D$ (see Appendix~\ref{app:up-selection} for details).

\begin{comment}
Moreover, we obtain a $\shuffledDP^{k}$ protocol for \selection with smaller error than previously known (see Theorem~\ref{theo:uppb-selection} for more details).
\begin{theorem}\label{theo:uppb-selection}
	For any $k \le D$ and $\eps = O(1)$, there is an $(\eps,1/\poly(n))$-$\shuffledDP^{k}$ protocol solving \selection with probability at least $0.99$ and $n = O(D \log D / \sqrt{k})$.
\end{theorem}

\end{comment}

\subsubsection{Lower Bounds for Parity Learning}

In \paritylearning, there is a hidden random vector $s \in \{0,1\}^D$, each user gets a random vector $x \in \{0,1\}^D$ together with the inner product $\langle s,x\rangle$ over $\mathbb{F}_2$, and the goal is to recover $s$. 
This problem is well-known for separating PAC learning from the Statistical Query (SQ) learning model \cite{kearns1998efficient}. In DP, it was studied by \cite{kasiviswanathan2011can} who gave a central DP protocol (also based on the exponential mechanism) computing it for $n = O(D)$, and moreover proved a lower bound of $n = 2^{\Omega(D)}$ for any $\localDP$ protocol, thus obtaining the first exponential separation between the central and local settings.

We give a lower bound for \paritylearning that hold against approximate-$\shuffledDP^{k}$ protocols:

\begin{theorem}\label{theo:lowb-learning-parity}
	For any $\eps = O(1)$, if $P$ is a public-coin $(\eps,o(1/n))$-$\shuffledDP^{k}$ protocol that solves \paritylearning with probability at least $0.99$, then $n \ge \Omega(2^{D/(k+1)})$.
\end{theorem}

Our lower bounds for \paritylearning can be generalized to the Statistical Query (SQ) learning framework of \cite{kearns1998efficient} (see Section~\ref{sec:SQ_connection} for more details).

\paragraph{Independent Work.} In a recent concurrent work, Cheu and Ullman~\cite{cheu2020limits} proved that robust $\shuffledDP$ protocols solving \selection and \paritylearning require $\Omega(\sqrt{D})$ and $\Omega(2^{\sqrt{D}})$ samples, respectively. Their results have no restriction on the number of messages sent by each user, but they only hold against the special case of \emph{robust} protocols. Our results provide stronger lower bounds when the number of messages per user is less than $\sqrt{D}$, and apply to the most general $\shuffledDP$ model without the robustness restriction.

\subsection{Lower Bounds for Two-Party DP Protocols}

Finally, we consider another model of distributed DP, called the \emph{two-party} model~\cite{mcgregor2010limits}, denoted $\twopartyDP$. In this model, there are two parties, each holding part of the dataset. The DP guarantee is enforced on the view of each party (i.e., the transcript, its private randomness, and its input). See Section~\ref{sec:two-party} for a  formal treatment.

McGregor et al.~\cite{mcgregor2010limits} studied the $\twopartyDP$ and proved an interesting separation of $\Omega_{\eps}(n)$ between the global sensitivity and $\eps$-DP protocol in this model. However, this lower bound does not extend to the approximate-DP case (where $\delta > 0$); in this case, the largest known gap (also proved in~\cite{mcgregor2010limits}) is only $\tilde{\Omega}_{\eps}(\sqrt{n})$, and it was left as an open question if this can be improved\footnote{The conference version of the paper~\cite{mcgregor2010limits} actually claimed to also have a lower bound $\Omega_{\eps}(n)$ for the approximate-DP case as well. However, it was later found to be incorrect; see~\cite{McGregorMPRTV11} for more discussions.}. We answer this question by showing that the gap of $\tilde{\Omega}_{\eps}(n)$ holds even against approximate-DP protocols:

\begin{theorem} \label{thm:two-party-lb}
For any $\eps = O(1)$ and any sufficiently large $n \in \N$, there is a function $f\colon \{0, 1\}^{2n} \to \R$ whose global sensitivity is one and such that no $(\eps, o(1/n))$-$\twopartyDP$ protocol can compute $f$ to within an error of $o(n/\log n)$.
\end{theorem}

The above bound is tight up to a logarithmic factors in $n$, as it is trivial to achieve an error of $n$. 

The proof of Theorem~\ref{thm:two-party-lb} is unlike others in the paper; in fact, we only employ simple reductions starting from the hardness of inner product function already shown in~\cite{mcgregor2010limits}. Specifically, our function is a sum of blocks of inner product modulo 2. While this function is not symmetric, we show (Theorem~\ref{thm:two-party-lb-symmetric}) that it can be easily symmetrized.

\subsection{Discussions and Open Questions}\label{sec:conc}

In this work, we study DP in distributed models, including the local and shuffle settings. By building on the moment matching method and using the newly defined notion of dominated protocols, we give novel lower bounds in both models for three fundamental problems: \DE, \selection, and \paritylearning. While our lower bounds are (nearly) tight in a large setting of parameters, there are still many interesting open questions, three of which we highlight below:

\begin{itemize}

\item \textbf{$\shuffledDP$ Lower Bounds for Protocols with Unbounded Number of Messages.} Our connection between $\shuffledDP$ and dominated protocols becomes weaker as $k \to \infty$ (Lemma~\ref{lm:SDP-to-dominated-Toy}). As a result, it cannot be used to establish lower bounds against $\shuffledDP$ protocols with a possibly unbounded number of messages. In fact, we are not aware of any separation between central DP and $\shuffledDP$ without a restriction on the number of messages and without the robustness restriction. This remains a fundamental open question. (In contrast, separations between central DP and $\localDP$ are well-known, even for basic functions such as binary summation~\cite{chan2012optimal}.)

\item \textbf{Lower Bounds against Interactive Local/Shuffle Model.} Our lower bounds hold in the \emph{non-interactive} local and shuffle DP models, where all users send their messages simultaneously in a single round. While it seems plausible that our lower bounds can be extended to the \emph{sequentially interactive} local DP model~\cite{duchi2013local} (where each user speaks once but not simultaneously), it is unclear how to extend them to the fully interactive local DP model.

The situation for $\shuffledDP$ however is more complicated. Specifically, we are not aware of a formal treatment of an interactive setting for the shuffle model, which would be the first step in providing either upper or lower bounds. We remark that certain definitions could lead to the model being as powerful as the central model (in terms of achievable accuracy and putting aside communication constraints); see e.g.,~\cite{ishai2006cryptography} on how to perform secure computations under a certain definition of the shuffle model.

\item \textbf{$\shuffledDP^1$ Lower Bounds for $\DE$ with Larger $\delta$.} All but one of our lower bounds hold as long as $\delta = n^{-\omega(1)}$, which is a standard assumption in the DP literature. The only exception is that of Theorem~\ref{theo:single-message-LB}, which requires $\delta = 2^{-\Omega(\log^c n)}$ for some constant $c > 0$. It is interesting whether this can be relaxed to $\delta = n^{-\omega(1)}$. %(Lemma~\ref{lm:key-bound-2} takes the blame.)

\end{itemize}

\subsection{Organization}
We describe in Section~\ref{sec:tech_overv} the techniques underlying our results.%, and discuss open questions in Section~\ref{sec:conc}. 
Some basic definitions and notation are given in Section~\ref{sec:preliminaries}. We prove our main lower bounds for $\DE$ (Theorems~\ref{theo:LDP-LB-strong} and~\ref{theo:single-message-LB}) in Section~\ref{sec:lb_single_msg_shuffle_count_distinct}. In Section~\ref{sec:dominated_protocols}, we define dominated protocols and prove Lemma~\ref{lm:SDP-to-dominated-Toy}. Our lower bounds for \selection and \paritylearning are then proved in Section~\ref{sec:selection_parity_learning}. Theorem~\ref{theo:lowb-distinct-elements} is then proved in Section~\ref{sec:lb_count_distinct_local}. Our $\shuffledDP^{\geq 1}$ protocol for \DE is presented and analyzed in Section~\ref{sec:mult_msg_shuffle_prot_count_dist}. Finally, our lower bounds in the two-party model (in particular, Theorem~\ref{thm:two-party-lb}) are proved in Section~\ref{sec:two-party}. Some deferred proofs appear in Appendices~\ref{app:TV-Poi} and~\ref{app:missing-proof-HS}. The connection to the SQ model is presented in Appendix~\ref{sec:SQ_connection}. Finally, in Appendix~\ref{app:up-selection}, we describe the $\shuffledDP^k$ protocol for \selection with sample complexity $\Tilde{O}(D/\sqrt{k})$.

\begin{comment}
\begin{theorem}\label{thm:learning_SQ_pure}
	If a learning task can be solved by an $\eps$-DP $k$-message shuffle protocol with $n$ samples and probability $0.99$, then it can be solved by an algorithm making $O(2^{k} \cdot e^{\eps} \cdot n)$ SQ queries with tolerance $\tau = \Theta\left(\frac{1}{n \cdot 2^{k} \cdot e^{\eps}}\right)$.
\end{theorem}

\begin{theorem}\label{thm:learning_SQ_approximate}
	If a learning task can be solved by an $(\eps,o(1/n))$-DP $k$-message shuffle protocol with $n$ samples and probability $0.99$, then it can be solved by an algorithm making $O(n^{k + 1} \cdot e^{\eps})$ SQ queries with tolerance $\tau = \Theta\left(\frac{1}{n^{k+1} \cdot e^{\eps}}\right)$.
\end{theorem}
\end{comment}

	\section{Overview of Techniques}\label{sec:tech_overv}
\everymath{}

\newcommand{\calW}{\mathcal{W}}
\newcommand{\KL}{\mathrm{KL}}

In this section, we describe the main intuition behind our lower bounds. As alluded to in Section~\ref{sec:intro}, we give two different proofs of the lower bounds for $\DE$ in the $\localDP$ and $\shuffledDP$ settings, each with its own advantages:

\begin{itemize}
	\item \textbf{Proof via Moment Matching.} Our first proof is   technically the hardest in our work. It applies to the much more challenging low-privacy setting (i.e., $(\ln n - O(\ln \ln n),\delta)$-$\localDP$), and shows an $\Omega(n/\polylog(n))$ lower bound on the additive error (Theorem~\ref{theo:LDP-LB-strong}). Together with our new improved connection between $\shuffledDP^1$ and $\localDP$ (Lemma~\ref{lm:SDP-to-LDP}), it also implies the same lower bound for protocols in the $\shuffledDP^1$ model. The key ideas behind the first proof will be discussed in Section~\ref{sec:overview-moment-matching}.
	
	\item \textbf{Proof via Dominated Protocols.} Our second proof has the advantage of giving the optimal $\Omega(n)$ lower bound on the additive error (Theorem~\ref{theo:lowb-distinct-elements}), but only in the constant privacy regime (i.e., $(O(1),\delta)$-$\localDP$), and it is relatively simple compared to the first proof.
	
	Moreover, the second proof technique is very general and is a conceptual contribution: it can be applied to show lower bounds for other fundamental problems (i.e., \selection and \paritylearning; Theorems~\ref{theo:lowb-selection} and~\ref{theo:lowb-learning-parity}) against multi-message $\shuffledDP$ protocols. We will highlight the intuition behind the second proof in Section~\ref{sec:overview-dominated}.
\end{itemize}

%Finally we will discuss the ideas behind our new protocol for distinct elements from Theorem~\ref{theo:upper-bound-DE}.
While our lower bounds also work for the public-coin $\shuffledDP$ models, throughout this section, we focus on private-coin models in order to simplify the presentation. The full proofs extending to public-coin protocols are given in later sections.
% We believe that the lower bounds for private-coin models already contain our most important insights. 

\subsection{Lower Bounds for $\DE$ via Moment Matching}\label{sec:overview-moment-matching}

% Our first proof involves many technical steps.
To clearly illustrate the key ideas behind the first proof, we will focus on the pure-DP case where each user can only send $O(\log n)$ bits. In Section~\ref{sec:lb_single_msg_shuffle_count_distinct}, we generalize the proof to approximate-DP and remove the restriction on communication complexity.

\begin{theorem}[A Weaker Version of Theorem~\ref{theo:LDP-LB-strong}]\label{theo:DE-LB-lowprivacy-toy}
	For $\eps = \ln(n / \log^7 n)$ and $D = n/\log^5 n$, %if $P$ is $\eps$-$\localDP$ and each user sends $O(\log n)$ bits, then it cannot solve $\DE_{n,D}$ with probability at least $0.99$ and error $o(D)$.
	no $\eps$-$\localDP$ protocol where each user sends $O(\log n)$ bits can solve $\DE_{n,D}$ with error $o(D)$.
\end{theorem}

Throughout our discussion, we use $R: [D] \to \calM$ to denote a $\ln(n / \log^7 n)$-$\localDP$ randomizer. By the communication complexity condition of Theorem~\ref{theo:DE-LB-lowprivacy-toy}, we have that $|\calM| \le \poly(n)$.

Our proof is inspired by the lower bounds for estimating distinct elements in the property testing model, e.g.,~\cite{ValiantV17,WY2019chebyshev}. In particular, we use the so-called \emph{Poissonization} trick. To discuss this trick, we start with some notation. For a vector $\vlambda \in \R^{D}$, we use $\vPoi(\vlambda)$ to denote the joint distribution of $D$ independent Poisson random variables:
\[
\vPoi(\vlambda) := (\Poi(\vlambda_1),\Poi(\vlambda_2),\dotsc,\Poi(\vlambda_n)).
\]

For a distribution $\vU$ on $\R^{D}$, we define the corresponding mixture of multi-dimensional Poisson distributions as follows:
\[
\Ex[\vPoi(\vU)] := \Ex_{\vlambda \leftarrow \vU} \vPoi(\vlambda).
\]

For two random variables $X$ and $Y$ supported on $\R^{\calM}$, we use $X + Y$ to denote the random variable distributed as a sum of two independent samples from $X$ and $Y$.

\highlight{Shuffling the Outputs of the Local Protocol.} Our first observation is that the analyzer for any local protocol computing $\DE$ should achieve the same accuracy if it only sees the histogram of the randomizers' outputs. This holds because only seeing the histogram of the outputs is equivalent to shuffling the outputs by a uniformly random permutation, which is in turn equivalent to shuffling the users in the dataset uniformly at random. Since shuffling the users in a dataset does not affect the number of distinct elements, it follows that only seeing the histogram does not affect the accuracy. Therefore, we only have to consider the histogram of the outputs of the local protocol computing $\DE$. For a dataset $W$, we use $\Hist_{R}(W)$ to denote the distribution of the histogram with randomizer $R$.

\highlight{Poissonization Trick.} Given a distribution $\calD$ on $\calM$, suppose we draw a sample $m \leftarrow \Poi(\lambda)$, and then draw $m$ samples from $\calD$. If we use $N$ to denote the random variable corresponding to the histogram of these $m$ samples, it follows that each coordinate of $N$ is independent, and $N$ is distributed as $\vPoi( \lambda \vmu)$, where $\vmu_i = \calD_i$ for each $i \in \calM$.

We can now apply the above trick to the context of local protocols (recall that by our first observation, we can focus on the histogram of the outputs). Suppose we build a dataset by drawing a sample $m \leftarrow \Poi(\lambda)$ and then adding $m$ users with input $z$. By the above discussion, the corresponding histogram of the outputs with randomizer $R$ is distributed as $\vPoi(\lambda \cdot R(z))$, where $R(z)$ is treated as an $|\calM|$-dimensional vector corresponding to its probability distribution.

\highlight{Moment-Matching Random Variables.} Our next ingredient is the following construction of two moment-matching random variables used in~\cite{WY2019chebyshev}. Let $L \in \mathbb{N}$ and $\Lambda = \Theta(L^2)$. There are two random variables $U$ and $V$ supported on $\{0\} \cup [1,\Lambda]$, such that $\Ex[U] = \Ex[V] = 1$ and $\Ex[U^j] = \Ex[V^j]$ for every $j \in [L]$. Moreover $U_0 - V_0 > 0.9$. That is, $U$ and $V$ have the same moments up to degree $L$, while the probabilities of them being zero differs significantly. We will set $L = \log n$ and hence $\Lambda = \Theta(\log^2 n)$.

\highlight{Construction of Hard Distribution via Signal/Noise Decomposition.} Recalling that $D = n / \log^5 n$, we will construct two input distributions for $\DE_{n,D}$.\footnote{In fact, in our presentation the number of inputs in each dataset from our hard distributions will not be exactly $n$, but only concentrated around $n$. This issue can be easily resolved by throwing ``extra'' users in the dataset; we refer the reader to Section~\ref{sec:low-privacy-lowb} for the details.} A sample from both distributions consists of two parts: a signal part with $D$ many users in expectation, and a noise part with $n - D$ many users in expectation. 
%Our goal are the following: (1) First, the signal parts of the two distributions will make the corresponding the number of distinct elements differ significant. But the corresponding distributions on histogram may not look similar. (2) Second

\newcommand{\calDsignal}{\calD_{\sf signal}}
\newcommand{\calDnoise}{\calD_{\sf noise}}

Formally, for a distribution $W$ over $\mathbb{R}^{\ge 0}$ and a subset $E \subset [D]$, the dataset distributions $\calDsignal^W$ and $\calDnoise^{E}$ are constructed as follows: 

\begin{itemize}
	\item \textbf{$\calDsignal^W$:} for each $i \in [D]$, we independently draw $\lambda_i \leftarrow W$, and $n_i \leftarrow \Poi(\lambda_i)$, and add $n_i$ many users with input $i$. 
	
	\item \text{$\calDnoise^{E}$:} for each $i \in E$, we independently draw $n_i \leftarrow \Poi((n-D) / |E|)$, and add $n_i$ many users with input $i$. 
\end{itemize}

We are going to fix a ``good'' subset $E$ of $[D]$ such that $|E| \le 0.02 \cdot D$ (we will later specify the other conditions for being ``good''). Therefore, when it is clear from the context, we will use $\calDnoise$ instead of $\calDnoise^{E}$. 

% \pasin{How about we use $\calD \star \calD'$ for convolution of the two distributions, and $\alpha \calD + (1 - \alpha) \calD'$ for mixture of the two distribution?}
% \lnote{Ah I tried a bit and it looks pretty weird for several reasons. For example shall we write something like 
% 	\[
% 	\bigstar_{i \in [D]} \vPoi(W \cdot R(i)) \star \bigstar_{i \in [E]} \vPoi((n-D)/|E| \cdot R(i))
% 	\]? Looks not very good... But if we don't use bigstar but use $\sum$ there then it seems inconsistent...}
% \lnote{I am thinking about the following solution, ``by default'' use $+$ as convolution for random variables, try to minimize its usage as mixture and remark it when we really have to. I think the only place they both appear is Sec 9 and I tried a bit to make it cleaner, especially Lemma 9.6. Also clarified it in the preliminary.}

Our two hard distributions are then constructed as $\calD^{U} := \calDsignal^{U} + \calDnoise$ and $\calD^{V} := \calDsignal^{V} + \calDnoise$. Using the fact that $\Ex[U] = \Ex[V] = 1$, one can verify that there are $D$ users in each of $\calDsignal^{U}$ and $\calDsignal^{V}$ in expectation. Similarly, one can also verify there are $n - D$ users in $\calDnoise$ in expectation. Hence, both $\calD^{U}$ and $\calD^{V}$ have $n$ users in expectation. In fact, the number of users from both distributions concentrates around $n$.

We now justify our naming of the signal/noise distributions. First, note that the number of distinct elements in the signal parts $\calDsignal^{U}$ and $\calDsignal^{V}$ concentrates around $(1 - \Ex[e^{-U}]) \cdot D$ and $(1 - \Ex[e^{-V}]) \cdot D$ respectively. By our condition that $U_0 - V_0 > 0.9$, it follows that the signal parts of $\calD^U$ and $\calD^V$ separates their numbers of distinct elements by at least $0.4 D$. Second, note that although $\calDnoise$ has $n - D \gg D$ many users in expectation, they are from the subset $E$ of size less than $0.02 \cdot n$. Therefore, these users collectively cannot change the number of distinct elements by more than $0.02 \cdot n$, and the numbers of distinct elements in $\calD^U$ and $\calD^V$ are still separated by $\Omega(D)$.

\highlight{Decomposition of Noise Part.} To establish the desired lower bound, it now suffices to show for the local randomizer $R$, it holds that $\Hist_{R}(\calD^{U})$ and $\Hist_{R}(\calD^{V})$ are very close in statistical distance. For $W \in \{U,V\}$, we can decompose $\Hist_{R}(\calD^{W})$ as
\[
\Hist_{R}(\calD^{W}) = \sum_{i \in [D]} \vPoi(W \cdot R(i)) + \sum_{i \in [E]} \vPoi((n-D)/|E| \cdot R(i)).
\] 
By the additive property of Poisson distributions, letting $\vnu = (n-D) / |E| \cdot \sum_{i \in [E]} R(i)$, we have that $\sum_{i \in [E]} \vPoi((n-D)/|E| \cdot R(i)) = \vPoi(\vnu)$. 

Our key idea is to decompose $\vnu$ carefully into $D+1$ nonnegative vectors $\vnu^{(0)},\vnu^{(1)},\dotsc,\vnu^{(D)}$, such that $\vnu = \sum_{i=0}^{D} \vnu^{(i)}$. Then, for $W \in \{U,V\}$, we have
\[
\Hist_{R}(\calD^{W}) = \vPoi(\vnu^{(0)}) + \sum_{i \in [D]} \vPoi(W \cdot R(i) + \vnu^{(i)}).
\]
To show that $\Hist_{R}(\calD^{U})$ and $\Hist_{R}(\calD^{V})$ are close, it suffices to show that for each $i \in [D]$, it is the case that $\vPoi(U \cdot R(i) + \vnu^{(i)})$ and $\vPoi(V \cdot R(i) + \vnu^{(i)})$ are close. We show that they are close when $\vnu^{(i)}$ is sufficiently large on every coordinate compared to $R(i)$.

\begin{lemma}[Simplification of Lemma~\ref{lm:key-bound-2}]\label{lm:moment-matching-toy}
	For each $i \in [D]$, and every $\vlambda \in (\R^{\ge 0})^{\calM}$, if $\vlambda_z \ge 2 \Lambda^2 \cdot R(i)_z$ for every $z \in \calM$, then\footnote{We use $\|\calD_1 - \calD_2\|_{TV}$ to denote the total variation (aka statistical) distance between two distributions $\calD_1, \calD_2$.}
	\[
	\| \Ex[\vPoi(U \cdot R(i) + \vlambda)] - \Ex[\vPoi(V \cdot R(i) + \vlambda)] \|_{TV} \le \frac{1}{n^2}.
	\]
\end{lemma}

To apply Lemma~\ref{lm:moment-matching-toy}, we simply set $\vnu^{(i)} = (2 \Lambda^2) \cdot R(i)$ and $\vnu^{(0)} = \vnu - \sum_{i \in [D]} \vnu^{(i)}$. Letting $\vmu = \sum_{i \in [D]} R(i)$, the requirement that $\vnu^{(0)}$ has to be nonnegative translates to $\vnu_z \ge 2 \Lambda^2 \cdot \vmu_z$, for each $z \in \calM$.

\highlight{Construction of a Good Subset $E$.} So we want to pick a subset $E \subset [D]$ of size at most $0.02 \cdot D$ such that the corresponding $\vnu^E = (n-D) / |E| \cdot \sum_{i \in [E]} R(i)$ satisfies $\vnu^E_z \ge 2 \Lambda^2 \cdot \vmu_z$ for each $z \in \calM$. We will show that a simple random construction works with high probability: i.e., one can simply add each element of $[D]$ to $E$ independently with probability $0.01$.

More specifically, for each $z \in \calM$, we will show that with high probability $\vnu^E_z \ge 2 \Lambda^2 \cdot \vmu_z$. Then the correctness of our construction follows from a union bound (and this step crucially uses the fact that $|\calM| \le \poly(n)$). 

Now, let us fix a $z \in \calM$. Let $m^* = \max_{i \in [D]} R(i)_z$. Since $R$ is $\ln(n/\log^7 n)$-DP, it follows that $\vnu_z \ge \frac{n-D}{n/\log^7 n} \cdot m^* \ge \frac{\log^7 n}{2} \cdot m^*$. We consider the following two cases:
\begin{enumerate}
\item If $m^* \ge \vmu_z / \log^2 n$, we immediately get that $\vnu_z \ge \log^5 n /2 \cdot \vmu_z \ge 2 \Lambda^2 \cdot \vmu_z$ (which uses the fact that $\Lambda = \Theta(\log^2 n)$).
\item If $m^* < \vmu_z / \log^2 n$, then in this case, the mass $\vmu_z$ is distributed over at least $\log^2 n$ many components $R(i)_z$. Applying Hoeffding's inequality shows that with high probability over $E$, it is the case that $\vnu^E_z \ge \Theta(n/D) \cdot \vmu_z \ge \Lambda^2 \cdot \vmu_z$ (which uses the fact that $D = n/\log^5 n$).
\end{enumerate}

See the proof of Lemma~\ref{lm:E-is-good-whp} for a formal argument and how to get rid of the assumption that $|\calM| \le \poly(n)$.

\highlight{The Lower Bound.} From the above discussions, we get that
\[
\|\Hist_{R}(\calD^{U}) - \Hist_R(\calD^{V})\|_{TV} \le \sum_{i=1}^{D} \| \Ex[\vPoi(U \cdot R(i) + \vnu^{(i)})] - \Ex[\vPoi(V \cdot R(i) + \vnu^{(i)})] \|_{TV} \le 1/n.
\]

Hence, the analyzer of the local protocol with randomizer $R$ cannot distinguish $\calD^U$ and $\calD^V$, and thus it cannot solve $\DE_{n,D}$ with error $o(D)$ and $0.99$ probability. See the proof of Theorem~\ref{theo:LDP-LB-strong-private-coin} for a formal argument and how to deal with the fact that there may not be exactly $n$ users in dataset from $\calD^{U}$ or $\calD^{V}$.

\highlight{Single-Message $\shuffledDP$ Lower Bound.} To apply the above lower bound to $\shuffledDP^1$ protocols, the natural idea is to resort to the connection between the $\shuffledDP^1$ and $\localDP$ models. In particular,~\cite{Cheu18} showed that $(\eps,\delta)$-$\shuffledDP^1$ protocols are also $(\eps + \ln n,\delta)$-$\localDP$.

It may seem that the $\ln n$ privacy guarantee is very close to the $\ln n - O(\ln \ln n)$ one in Theorem~\ref{theo:LDP-LB-strong}. But surprisingly, it turns out (as was stated in Theorem~\ref{th:weak_privacy_local_prot_DE}) that there is a $\left(\ln n + O(1)\right)$-$\localDP$ protocol solving $\DE_{n,n}$ (hence also $\DE_{n,D}$) with error $O(\sqrt{n})$. Hence, to establish the $\shuffledDP^1$ lower bound (Theorem~\ref{theo:single-message-LB}), we rely on the following stronger connection between $\shuffledDP^1$ and $\localDP$ protocols.

\begin{lemma}[Simplification of Lemma~\ref{lm:SDP-to-LDP}]\label{lm:SDP-to-LDP-toy}
	For every $\delta \le 1/n^{\omega(1)}$, if the randomizer $R$ is $(O(1),\delta)$-$\shuffledDP^1$ on $n$ users, then $R$ is $\left(\ln(n\log^2n /\log\delta^{-1}),n^{-\omega(1)}\right)$-$\localDP$.
\end{lemma}
Setting $\delta = 2^{-\log^k n}$ for a sufficiently large $k$ and combining Lemma~\ref{lm:SDP-to-LDP-toy} and Theorem~\ref{theo:LDP-LB-strong} gives us the desired lower bound against $\shuffledDP^1$ protocols. 

%One key idea in our proof is the Poissonization trick. That is, given a distribution $\calD$ on $\calM$. Suppose we take a sample $x \leftarrow \Poi(\lambda)$ and then take $x$ samples from $\calD$. If we use $N$ to denote the random variable corresponding to the histogram of these $x$ samples, it follows that each coordinate of $N$ is independent, and $N$ is distributes as $\Poi( \lambda \vmu)$, where $\vmu_i = \calD_i$ for each $i \in \calM$.

\subsection{Lower Bounds for \DE and \selection via Dominated Protocols}\label{sec:overview-dominated}

\newcommand{\calWuniform}{\calW^{\sf uniform}}
%\newcommand{\calWhalf}{\calW^{\sf half}}

%We first overview the ideas behind our proofs of Theorem~\ref{theo:lowb-selection} and Theorem~\ref{theo:lowb-learning-parity}. 

We will first describe the proof ideas behind Theorem~\ref{theo:lowb-distinct-elements}, which is restated below. Then, we apply the same proof technique to obtain lower bounds for \selection (the lower bound for \paritylearning is established similarly; see Section~\ref{sec:lowb-learning-parity} for details).

\begin{lemma}[Detailed Version of Theorem~\ref{theo:lowb-distinct-elements}]\label{lm:dominated-DE-lowb-toy}
	For $\eps = o(\ln n) $, no $(\eps,o(1/n))$-dominated protocol can solve $\DE$ with error $o(n/e^{\eps})$.% and at least $0.99$ probability.
\end{lemma}

\highlight{Hard Distributions for $\DE_{n,n}$.} We now construct our hard instances for $\DE_{n,n}$. For simplicity, we assume $n = 2^D$ for an integer $D$, and identify the input space $[n]$ with $\{0,1\}^D$ by a fixed bijection. Let $\calU_{D}$ be the the uniform distribution over $\{0,1\}^{D}$. For $(\ell,s) \in [2] \times \{0,1\}^D$, we let $\calD_{\ell,s}$ be the uniform distribution on $\{x \in \{0,1\}^D : \langle x,s\rangle = \ell \}$.

\begin{comment}
We also let 
$$
\calD_{\ell,s}^{\alpha} = \alpha \cdot \calD_{\ell,s} + (1 - \alpha) \cdot \calU_{D},
%(\frac{1}{2} + \frac{\alpha}{2}) \cdot \calD_{\ell,s} + (\frac{1}{2} - \frac{\alpha}{2}) \cdot \calD_{1-\ell,s},
$$
which is a mixture of $\calD_{\ell,s}$ and $\calU_{D}$.
\end{comment}
We also use $\calD_{\ell,s}^{\alpha}$ to denote the mixture of $\calD_{\ell,s}$ and $\calU_{D}$ which outputs a sample from $\calD_{\ell,s}$ with probability $\alpha$ and a sample from $\calU_{D}$ with probability $1-\alpha$.

For a parameter $\alpha > 0$, we consider the following two dataset distributions with $n$ users:

\begin{itemize}
	\item $\calWuniform$: each user gets an i.i.d. input from $\calU_{D}$. That is, $\calWuniform := \calU_{D}^{\otimes n}$.
	
	\item $\calW^{\alpha}$: to sample a dataset from $\calW^{\alpha}$, we first draw $(\ell,s)$ from $[2] \times \{0,1\}^D$ uniformly at random, then each user gets an i.i.d. input from $\calD_{\ell,s}^\alpha$. Formally, $\calW^\alpha := \Ex_{(\ell,s) \leftarrow [2] \times \{0,1\}^D} (\calD_{\ell,s}^{\alpha})^{\otimes n}$.
\end{itemize}

Since for every $\ell,s$, it holds that $|\supp(\calD_{\ell,s}^{1})| \le n/2$, the number of distinct elements from any dataset in $\calW^{1}$ is at most $n/2$. On the other hand, since $\calU_{D}$ is a uniform distribution over $n$ elements, a random dataset from $\calWuniform = \calW^{0}$ has roughly $(1-e^{-1}) \cdot n > n/2$ distinct elements with high probability. Hence, the expected number of distinct elements of datasets from $\calW^{\alpha}$ is controlled by the parameter $\alpha$. A simple but tedious calculation shows that it is approximately $(1-e^{-1} \cdot \cosh(\alpha)) \cdot n$, which can be approximated by $(1 - e^{-1} \cdot (1+\alpha^2)) \cdot n$ for $n^{-0.1} < \alpha < 0.01$ (see Proposition~\ref{prop:DE-bound} for more details). Hence, any protocol solving $\DE$ with error $o(\alpha^2 n)$ should be able to distinguish between the above two distributions. Our goal is to show that this is impossible for $(\eps,o(1/n))$-dominated protocols.

\highlight{Bounding KL Divergence for Dominated Protocols.} Our next step is to upper-bound the statistical distance $\|\Hist_{R}(\calWuniform) - \Hist_R(\calW^{\alpha})\|_{TV}$. As in previous work~\cite{ullman2018tight,GhaziGKPV19,EdmondsNU20}, we may upper-bound the KL divergence instead. By the convexity and chain-rule properties of KL divergence, it follows that
\begin{align}
\KL(\Hist_R(\calW^\alpha) || \Hist_{R}(\calWuniform)) &\le \Ex_{(\ell,s) \leftarrow [2] \times \{0,1\}^D} \KL(R(\calD_{\ell,s}^{\alpha})^{\otimes n} || R(\calU_{D})^{\otimes n}) \notag\\
&=  n \cdot \Ex_{(\ell,s) \leftarrow [2] \times \{0,1\}^D} \KL(R(\calD_{\ell,s}^{\alpha}) || R(\calU_{D})). \label{eq:avg-KL-DE}
\end{align}

\highlight{Bounding the Average KL Divergence between a Family and a Single Distribution.} We are now ready to introduce our general tool for bounding average KL divergence quantities like~\eqref{eq:avg-KL-DE}. We first set up some notation. Let $\calI$ be an index set and $\{ \lambda_{v} \}_{v \in \calI}$ be a family of distributions on $\calX$, let $\pi$ be a distribution on $\calI$, and $\mu$ be a distribution on $\calX$. For simplicity, we assume that for every $x \in \calX$ and $v \in \calI$, it holds that $(\lambda_v)_x \le 2 \cdot \mu_x$ (which is true for $\{\calD_{\ell,s}^{\alpha}\}_{(\ell,s) \in [2] \times \{0,1\}^D}$ and $\calU_D$).

\begin{theorem}\label{theo:many-vs-one-toy}
	Let $W \colon \mathbb{R} \to \mathbb{R}$ be a concave function such that for all functions $\psi \colon \calX \to \R^{\ge 0}$ satisfying $\psi(\mu) \le 1$, it holds that 
		\[
		\Ex_{v \leftarrow \pi} \left[(\psi(\lambda_v) - \psi(\mu))^2\right] \le W(\|\psi\|_{\infty}).
		\] 
	Then for an $(\eps,\delta)$-dominated randomizer $R$, it follows that
	\[
	\Ex_{v \leftarrow \pi} [\KL(R(\lambda_v)||R(\mu))] \le O\left( W(2e^{\eps}) + \delta \right).
	\]
\end{theorem}

\highlight{Bounding~\eqref{eq:avg-KL-DE} via Fourier Analysis.} To apply Theorem~\ref{theo:many-vs-one-toy}, for $f\colon \calX \to \R^{\ge 0}$ with $f(\calU_{D}) = \Ex_{x \in \{0,1\}^D} [f(x)] \le 1$, we want to bound
\[
\Ex_{(\ell,s) \leftarrow [2] \times \{0,1\}^D} [(f(\calD_{\ell,s}^\alpha) - f(\calU_{D}))^2] = \Ex_{s \in \{0,1\}^D} \alpha^2 \cdot \hat{f}(s)^2.
\]
By Parseval's Identity (see Lemma~\ref{lm:parseval}), 
\[
\sum_{s \in \{0,1\}^D} \hat{f}(s)^2 = \Ex_{x \in \{0,1\}^D} f(x)^2 \le f(\calU_D) \cdot \|f\|_{\infty} \le \|f\|_{\infty}.
\]
Therefore, we can set $W(L) := \alpha^2 \cdot \frac{L}{2^D}$, and apply Theorem~\ref{theo:many-vs-one-toy} to obtain
\[
\Ex_{(\ell,s) \leftarrow [2] \times \{0,1\}^D} \KL(R(\calD_{\ell,s}^{\alpha}) || R(\calU_{D})) \le O(\alpha^2 \cdot e^{\eps} / n + \delta).
\]

We set $\alpha$ such that $\alpha^2 = c / e^{\eps}$ for a sufficiently small constant $c$ and note that $\delta = o(1/n)$. It follows that $
\KL(\Hist_R(\calW^\alpha) || \Hist_{R}(\calWuniform)) \le 0.01$, and therefore $\| \Hist_R(\calW^\alpha) - \Hist_{R}(\calWuniform) \|_{TV} \le 0.1$ by Pinsker's inequality. Hence, we conclude that $(\eps,o(1/n))$-dominated protocols cannot solve $\DE_{n,n}$ with error $o(n/e^{\eps})$, completing the proof of Lemma~\ref{lm:dominated-DE-lowb-toy}. Now Theorem~\ref{theo:lowb-distinct-elements} follows from Lemma~\ref{lm:dominated-DE-lowb-toy} and the fact that $(\eps,\delta)$-$\localDP$ protocols are also $(\eps,\delta)$-dominated.

\highlight{Lower Bounds for \selection against Multi-Message $\shuffledDP$ Protocols.} Now we show how to apply Theorem~\ref{theo:many-vs-one-toy} and Lemma~\ref{lm:SDP-to-LDP-toy} to prove lower bounds for \selection. For $(\ell,j) \in [2] \times [D]$, let $\calD_{\ell,j}$ be the uniform distribution on all length-$D$ binary strings with $j$th bit being $\ell$. Recall that $\calU_D$ is the uniform distribution on $\{0,1\}^D$. Again we aim to upper-bound the average-case KL divergence
$
\Ex_{(\ell,j) \leftarrow [2] \times [D]} \KL(R(\calD_{\ell,j}) || R(\calU_{D})).
$

To apply Theorem~\ref{theo:many-vs-one-toy}, for $f\colon \calX \to \R^{\ge 0}$ with $f(\calU_{D}) = \Ex_{x \in \{0,1\}^D} [f(x)] \le 1$, we want to bound
\[
\Ex_{(\ell,j) \leftarrow [2] \times [D]} [(f(\calD_{\ell,j}^\alpha) - f(\calU_{D}))^2] = \Ex_{j \in [D]} \hat{f}(\{j\})^2.
\]
By the Level-1 Inequality (see Lemma~\ref{lm:level-1}), it is the case that
\[
\sum_{j \in [D]} \hat{f}(\{j\})^2 \le O(\log \|f\|_{\infty}).
\]
Therefore, we can set $W(L) := c_1 \cdot \frac{\log L}{D}$ for an appropriate constant $c_1$, and apply Theorem~\ref{theo:many-vs-one-toy} to obtain
\[
\Ex_{(\ell,j) \leftarrow [2] \times [D]} \KL(R(\calD_{\ell,j}) || R(\calU_{D})) \le O\left(\frac{\eps}{D} + \delta\right).
\]
Combining this with Lemma~\ref{lm:SDP-to-LDP-toy} completes the proof (see the proofs of Lemma~\ref{lm:dominated-to-lowerbound} and Theorem~\ref{theo:lowb-selection} for the details).

\everymath{\displaystyle}

	\section{Preliminaries}\label{sec:preliminaries}

\subsection{Notation}

For a function $f \colon \calX \to \R$, a distribution $\calD$ on $\calX$, and an element $z \in \calX$, we use $f(\calD)$ to denote $\Ex_{x \leftarrow \calD}[f(x)]$ and $\calD_z$ to denote $\Pr_{x \leftarrow \calD}[x = z]$.  For a subset $E \subset \calX$, we use $\calD_E$ to denote $\sum_{z \in E} \calD_z = \Pr_{x \leftarrow \calD}[x \in E]$. We also use $\calU_D$ to denote the uniform distribution over $\{0,1\}^D$.

For two distributions $\calD_1$ and $\calD_2$ on sets $\calX$ and $\calY$ respectively, we use $\calD_1 \otimes \calD_2$ to denote their product distribution over $\calX \times \calY$.  For two random variables $X$ and $Y$ supported on $\R^{D}$ for $D \in \mathbb{N}$, we use $X + Y$ to denote the random variable distributed as a sum of two independent samples from $X$ and $Y$. For any set $\calS$, we denote by $\calS^*$ the set consisting of sequences on $\calS$, i.e., $\calS^* = \cup_{n \ge 0} \calS^n$.  For $x \in \R$, let $[x]_+$ denote
$\max(x, 0)$. For a predicate $P$, we use $\mathbb{1}[P]$ to denote the corresponding Boolean value of $P$, that is, $\mathbb{1}[P] = 1$ if $P$ is true, and $0$ otherwise.

For a distribution $\calD$ on a finite set $\calX$ and an event $\calE \subset \calX$ such that $\Pr_{z \leftarrow \calD}[z \in \calE] > 0$, we use $\calD | \calE$ to denote the conditional distribution such that 
\[
(\calD | \calE)_z = \begin{cases}
\frac{\calD_z}{\Pr_{z \leftarrow \calD}[z \in \calE]} \quad &\text{if $z \in \calE$,}\\
0 \quad &\text{otherwise.}
\end{cases}
\]

Slightly overloading the notation, we also use $\alpha \cdot \calD_1 + (1-\alpha) \cdot \calD_2$ to denote the mixture of distributions $\calD_1$ and $\calD_2$ with mixing weights $\alpha$ and $(1-\alpha)$ respectively. Whether $+$ means mixture or convolution will be clear from the context unless explicitly stated.

\subsection{Differential Privacy}

We now recall the basics of differential privacy that we will need.  Fix a finite set $\calX$, the space of user reports. A {\it dataset} $X$ is an element of $\calX^*$, namely a tuple consisting of elements of $\calX$.  Let $\hist(X) \in \mathbb{N}^{|\calX|}$ be the \emph{histogram} of $X$: for any $x \in \calX$, the $x$th component of $\hist(X)$ is the number of occurrences of $x$ in the dataset $X$. We will consider datasets $X, X'$ to be \emph{equivalent} if they have the same histogram (i.e., the ordering of the elements $x_1, \ldots, x_n$ does not matter). For a multiset $\calS$ whose elements are in $\calX$, we will also write $\hist(\calS)$ to denote the histogram of $\calS$ (so that the $x$th component is the number of copies of $x$ in $\calS$).

Let $n \in \mathbb{N}$, and consider a dataset $X = (x_1, \ldots, x_n) \in \calX^n$.  For an element $x \in \calX$, let $f_X(x) = \frac{\hist(X)_x}{n}$ be the {\it frequency} of $x$ in $X$, namely the fraction of elements of $X$ that are equal to $x$.
Two datasets $X, X'$ are said to be {\it neighboring} if they differ in a single element, meaning that we can write (up to equivalence) $X = (x_1, x_2, \ldots, x_n)$ and $X' = (x_1', x_2,  \ldots, x_n)$. In this case, we write $X \sim X'$. Let $\calZ$ be a set; we now define the differential privacy of a randomized function $P \colon \calX^n \rightarrow \calZ$ as follows.
\begin{definition}[Differential privacy (DP)~\cite{DworkMNS06,DworkKMMN06}]
	\label{def:dp}
	A randomized algorithm $P \colon \calX^n \rightarrow \calZ$ is {\it $(\eps, \delta)$-DP} if for every pair of neighboring datasets $X \sim X'$ and for every set $\calS \subset \calZ$, we have
	$$
	\Pr[P(X) \in \calS] \leq e^\eps \cdot \Pr[P(X') \in \calS] + \delta,
	$$
	where the probabilities are taken over the randomness in $P$.  Here, $\eps \geq 0$ and $\delta \in [0,1]$.
\end{definition}

If $\delta = 0$, then we use $\epsilon$-DP for brevity and informally refer to it as \emph{pure-DP}; if $\delta > 0$, we refer to it as \emph{approximate-DP}.  We will use the following post-processing property of DP.
\begin{lemma}[Post-processing, e.g.,~\cite{DworkR14}]
	\label{lem:post_process}
	If $P$ is $(\eps, \delta)$-DP, then for every randomized function $A$, the composed function $A \circ P$ is $(\eps, \delta)$-DP.
\end{lemma}

DP is nicely characterized by the following divergence between distributions, which will be used throughout the paper.

\begin{definition}[Hockey Stick Divergence]\label{defi:hockey-stick-divergence}
	For any $\eps > 0$, the \emph{$e^{\eps}$-hockey stick divergence} between distributions $\calD$ and $\calD'$ is defined as $d_{\eps}(\calD||\calD') := \sum_{x \in \supp(\calD)}[\calD_x - e^{\eps}\cdot \calD'_x]_{+}$.
\end{definition}

We next list two useful facts about the hockey stick divergence between distributions.

\begin{prop}\label{prop:key-facts-hockey-stick}
	Let $\calD$ and $\calD'$ be any  distributions.  Then, the following hold:
	
	\begin{enumerate}
		\item Let $\calD_{com}$ be another distribution.  Then, for any function $f$, it holds that
		\[
		d_{\eps}(f(\calD \otimes \calD_{com})||f(\calD' \otimes \calD_{com})) \le d_{\eps}(\calD||\calD').
		\]
		
		\item Suppose we can decompose $\calD = \sum_{i \in \calI} \alpha_i \calD_i $ and $\calD' = \sum_{i \in \calI} \beta_i \calD'_i$, where $\alpha_i$'s and $\beta_i$'s are tuples of positives reals summing up to $1$ and $\calD_i$'s and $\calD'_i$'s are distributions,
		then
		\[
		d_{\eps}(\calD||\calD') \le \sum_{i \in \calI} \alpha_i \cdot d_{\eps + \ln(\beta_i / \alpha_i)}(\calD_i||\calD'_i).
		\]
	\end{enumerate}
	
\end{prop}
\begin{proof}
	Item (1) follows from the post-processing property of DP, together with the definition of the hockey stick divergence.
	
	To prove Item (2), we note that
	\begin{align*}
	d_{\eps}(\calD||\calD') &= \sum_{x \in \supp(\calD)}[\calD_x - e^{\eps}\cdot \calD'_x]_{+}\\
	&= \sum_{x \in \supp(\calD)}\left[\sum_{i \in \calI} \alpha_i (\calD_i)_x - e^{\eps}\cdot \left(\sum_{i \in \calI} \beta_i (\calD'_i)_x\right)\right]_{+}\\
	&\le \sum_{i \in \calI} \sum_{x \in \supp(\calD_i)} \left[\alpha_i (\calD_i)_x - e^{\eps}\cdot \beta_i (\calD'_i)_x\right]_{+}\\
	&\le \sum_{i \in \calI} \alpha_i \cdot \sum_{x \in \supp(\calD_i)} \left[(\calD_i)_x - e^{\eps}\cdot \beta_i/\alpha_i (\calD'_i)_x\right]_{+}\\
	&\le \sum_{i \in \calI} \alpha_i \cdot d_{\eps + \ln(\beta_i / \alpha_i)}(\calD_i||\calD'_i).
	\qedhere
	\end{align*}
\end{proof}

\subsection{Shuffle Model}\label{sec:model}
We briefly review the \emph{shuffle model} of DP~\cite{bittau2017prochlo, erlingsson2019amplification, Cheu18}. 
The input to the model is a dataset $(x_1, \ldots, x_n) \in \calX^n$, where item $x_i \in \calX$ is held by user $i$. A protocol
$P\colon \calX \rightarrow \calZ$ in the shuffle model consists of three algorithms:
\begin{itemize}%[nosep]
	\item The {\it local randomizer} $R\colon \calX \rightarrow \calM^*$ takes as input the data of one user, $x_i \in \calX$, and outputs a sequence $(y_{i,1}, \ldots, y_{i,m_i})$ of {\it messages}; here $m_i$ is a positive integer.
	
	To ease discussions in the paper, we will further assume that the randomizer $R$ pre-shuffles its messages. That is, it applies a random permutation $\pi\colon [m_i] \to [m_i]$ to the sequence $(y_{i,1}, \ldots, y_{i,m_i})$ before outputting it.\footnote{Therefore, for every $x \in \calX$ and any two tuples $z_1,z_2 \in \calM^*$ that are equivalent up to a permutation, $R(x)$ outputs them with the same probability.} 
	
	\item The {\it shuffler} $S\colon \calM^* \rightarrow \calM^*$ takes as input a sequence of elements of $\calM$, say $(y_1, \ldots, y_m)$, and outputs a random permutation, i.e., the sequence $(y_{\pi(1)}, \ldots, y_{\pi(m)})$, where $\pi \in S_m$ is a uniformly random permutation on $[m]$. The input to the shuffler will be the concatenation of the outputs of the local randomizers.
	\item The {\it analyzer} $A\colon \calM^* \rightarrow \calZ$ takes as input a sequence of elements of $\calM$ (which will be taken to be the output of the shuffler) and outputs an answer in $\calZ$ that is taken to be the output of the protocol~$P$.
\end{itemize}
We will write $P = (R, S, A)$ to denote the protocol whose components are given by $R$, $S$, and $A$.  The main distinction between the shuffle and local model is the introduction of the shuffler $S$ between the local randomizer and the analyzer. As in the local model, the analyzer is untrusted in the shuffle model; hence privacy must be guaranteed with respect to the input to the analyzer, i.e., the output of the shuffler. Formally, we have: % We next define what it means for $P$ to be differentially private in the shuffled model.
\begin{definition}[DP in the Shuffle Model,~\cite{erlingsson2019amplification, Cheu18}]
	\label{def:dp_shuffled}
	A protocol $P = (R, S, A)$ is {\it $(\eps, \delta)$-DP} if, for any dataset $X = (x_1, \ldots, x_n)$, the algorithm 
	$$
	(x_1, \ldots, x_n) \mapsto S(R(x_1), \ldots, R(x_n))
	$$
	is $(\eps, \delta)$-DP. 
\end{definition}
Notice that the output of $S(R(x_1), \ldots, R(x_n))$ can be simulated by an algorithm that takes as input the {\it multiset} consisting of the union of the elements of $R(x_1), \ldots, R(x_n)$ (which we denote as $\bigcup_i R(x_i)$, with a slight abuse of notation) and outputs a uniformly random permutation of them. Thus, by Lemma~\ref{lem:post_process}, it can be assumed without loss of generality for privacy analyses that the shuffler simply outputs the multiset $\bigcup_i R(x_i)$. 
For the purpose of analyzing the accuracy of the protocol $P = (R, S, A)$, we define its \emph{output} on the dataset $X = (x_1, \ldots, x_n)$ to be $P(X) := A(S(R(x_1), \ldots, R(x_n)))$. We also remark that the case of {\it local DP}, formalized in Definition~\ref{def:dp_local}, is a special case of the shuffle model where the shuffler $S$ is replaced by the identity function:

\begin{definition}[Local DP~\cite{kasiviswanathan2011can}]
	\label{def:dp_local}
	A protocol $P = (R,A)$ is {\it $(\eps, \delta)$-DP in the local model} (or {\it $(\eps, \delta)$-locally DP}) if the function $x \mapsto R(x)$ is $(\eps, \delta)$-DP. 
\end{definition}
We say that the {\it output} of the protocol $P$ on an input dataset $X = (x_1, \ldots, x_n)$ is $P(X) := A(R(x_1), \ldots, R(x_n))$.

We denote DP in the shuffle model by $\shuffledDP$, and the special case where each user can send at most\footnote{We may assume w.l.o.g. that each user sends \emph{exactly} $k$ messages; otherwise, we may define a new symbol $\perp$ and make each user sends $\perp$ messages so that the number of messages becomes exactly $k$.\label{fn:equal-msg}} $k$ messages by $\shuffledDP^k$. We denote DP in the local model by $\localDP$.

\paragraph*{Public-Coin DP.} 

The default setting for local and shuffle models is private-coin, i.e., there is no randomness shared between the randomizers and the analyzer.   We will also study the public-coin variants of the local and shuffle models. In the public-coin setting, each local randomizer also takes a public random string $\alpha \leftarrow \{0,1\}^{*}$ as input. The analyzer is also given the public random string $\alpha$. We use $R_{\alpha}(x)$ to denote the local randomizer with public random string being fixed to $\alpha$. At the start of the protocol, all users jointly sample a public random string from a publicly known distribution $\calDpb$.

Now, we say that a protocol $P = (R,A)$ is \emph{$(\eps,\delta)$-DP in the public-coin local model}, if the function 
$$
x \underset{\alpha \leftarrow \calDpb}{\mapsto} (\alpha,R_{\alpha}(x))
$$ 
is $(\eps,\delta)$-DP. 

Similarly, we say that a protocol $P = (R, S, A)$ is \emph{$(\eps, \delta)$-DP in the public-coin shuffle model}, if for any dataset $X = (x_1, \ldots, x_n)$, the algorithm 
$$
(x_1, \ldots, x_n) \underset{\alpha \leftarrow \calDpb}{\mapsto} (\alpha, S(R_{\alpha}(x_1), \ldots, R_{\alpha}(x_n)) )
$$
is $(\eps, \delta)$-DP. 

%We will also use $\pubR_{\alpha}(x)$ to denote the randomized function $x \mapsto (\alpha,R_{\alpha}(x)) (\alpha \leftarrow \calDpb)$ for convenience.

\subsection{Useful Divergences}
\label{sec:divergences}

We will make use of two important divergences between distributions, the KL-divergence and the $\chi^2$-divergence, defined as
\[
KL(P||Q) = \Ex_{z \leftarrow P} \log\left(\frac{P_z}{Q_z}\right)\quad\text{and}\quad
\chi^2(P||Q) = \Ex_{z \leftarrow Q} \left[ \frac{P_z - Q_z}{Q_z} \right]^2.
\]
We rely on the key fact that $\chi^2$-divergence upper-bounds KL-divergence~\cite{GS02}, that is,
\[
\KL(P||Q) \le \chi^2(P||Q).
\]
We will also use Pinsker's inequality, whereby the total variation distance lower-bounds the KL-divergence:
\[
\KL(P||Q) \geq \frac{2}{\ln 2} \| P - Q \|_{TV}^2.
\]

%\lnote{ADD something about variables with negative correlations. Check~\cite[Proposition~7 and 11]{DubhashiR98}}

\subsection{Fourier Analysis}

We now review some basic Fourier analysis and then introduce two inequalities that will be heavily used in our proofs. For a function $f\colon\{0,1\}^{D} \to \R$, its Fourier transform is given by the function $\hat{f}(S) := \Ex_{x \leftarrow \calU_D} [f(x) \cdot (-1)^{\sum_{i \in S} x_i}]$. We also define $\|f\|_2^2 = \Ex_{x \leftarrow \calU_D} [f(x)^2]$. For $k \in \mathbb{N}$, we define the \emph{level-$k$ Fourier weight} as $\bfW^{k}[f] := \sum_{S\subseteq [D], |S| = k} \hat{f}(S)^2$. For convenience, for $s \in \{0,1\}^D$, we will also write $\hat{f}(s)$ to denote $f(\chi_s)$, where $\chi_s$ is the set $\{ i : i \in [D] \wedge s_i = 1 \}$. One key technical lemma is the Level-1 Inequality from~\cite{ODonnell14}, which was also used in~\cite{GhaziGKPV19}.

\begin{lemma}[Level-1 Inequality]\label{lm:level-1}
	Suppose $f\colon\{0,1\}^D \to \R_{\ge 0}$ is a non-negative-valued function with $f(x) \in [0,L]$ for all $x \in \{0,1\}^D$, and $\Ex_{x\sim\calU_D}[f(x)] \le 1$. Then, $\bfW^1[f] \le 6 \ln (L+1)$.
\end{lemma}
We also need the standard Parseval's identity.
\begin{lemma}[Parseval's Identity]\label{lm:parseval}
	For all functions $f\colon\{0,1\}^D \to \R$, 
	\[
	\|f\|_2^2 = \sum_{S \subseteq [D]} \hat{f}(S)^2.
	\]
\end{lemma}
	\section{Low-Privacy $\localDP$ and $\shuffledDP^1$ Lower Bounds for \DE}\label{sec:lb_single_msg_shuffle_count_distinct}

In this section, we prove Theorem~\ref{theo:LDP-LB-strong} and Theorem~\ref{theo:single-message-LB}. In Section~\ref{sec:poisson-prelim}, we introduce some necessary definitions and notation.
In Section~\ref{sec:low-privacy-lowb}, we prove our lower bound for low-privacy (private-coin) $\localDP$ protocols computing $\DE$.
In Section~\ref{sec:SDP-to-LDP-stronger}, we show the improved connection between $\shuffledDP^1$ and $\localDP$, which implies our lower bounds for $\shuffledDP^1$ protocols for $\DE$.
In Section~\ref{sec:adaption-public-coin}, we describe how to adapt the proof to public-coin protocols.

\subsection{Preliminaries}\label{sec:poisson-prelim}

\begin{comment}
For a vector $\vlambda \in \R^{D}$, we use $\vPoi(\vlambda)$ to denote the joint distribution of independent $D$ Poisson distributions:
\[
\vPoi(\vlambda) := (\Poi(\vlambda_1),\Poi(\vlambda_2),\dotsc,\Poi(\vlambda_n)).
\]

For a distribution $\vU$ on $\R^{D}$, we define the corresponding mixture of multi-dimensional Poisson distributions as follows:
\[
\Ex[\vPoi(\vU)] := \Ex_{\vlambda \leftarrow \vU} \vPoi(\vlambda).
\]

already in tech-overview
\end{comment}

Recall that we use the notations $\vPoi(\vlambda)$ and $\Ex[\vPoi(\vU)]$ to denote multi-dimensional Poisson distributions and their mixtures, respectively (see Section~\ref{sec:overview-moment-matching} for the precise definitions). 

We also recall the key additive property of multi-dimensional Poisson distributions: for $\valpha,\vbeta \in \R^{D}$, we have that
\[
\vPoi(\valpha) + \vPoi(\vbeta) = \vPoi(\valpha + \vbeta).
\]

\subsection{Low-Privacy $\localDP$ Lower Bounds for $\DE$}\label{sec:low-privacy-lowb}

%\begin{theorem}\label{theo:single-message-LB-private-coin}
%	For all $\eps = O(1)$, there are $\delta = 2^{-\Theta(\log^8 n)}$ and $D = \Theta(n/\log^4 n)$ such that $(\eps,\delta)$-DP private coin protocol in the shuffle model with $n$ users cannot solve $\DE_{n,D}$ with error $o(D)$ and probability at least $0.99$.
%\end{theorem}

%The main result of this section is the following theorem, which is a more detailed version of Theorem~\ref{theo:LDP-LB-strong}. 
We will first prove the low-privacy $\localDP$ lower bounds in the private-coin setting, which is captured by the following theorem.

\begin{theorem}[The Private-Coin Case of Theorem~\ref{theo:LDP-LB-strong}]\label{theo:LDP-LB-strong-private-coin}
	For some $\eps = \ln(n / \Theta(\log^6 n))$ and $D = \Theta(n/\log^4 n)$, if $P$ is a private-coin $(\eps, n^{-\omega(1)})$-$\localDP$ protocol, then it cannot solve $\DE_{n,D}$ with error $o(D)$ and probability at least $0.99$.
\end{theorem}

\subsubsection{Technical Lemmas}

Now we need the following construction from~\cite{WY2019chebyshev} (which uses a classical result from \cite[2.11.1]{timan2014theory}).

\begin{lemma}[\cite{WY2019chebyshev}]\label{lm:construction-two-RV}
	There is a constant $c$ such that, for all $L \in \mathbb{N}$, there are two distributions $U$ and $V$ supported on $\{0\} \cup [1,\Lambda]$ for $\Lambda = c \cdot L^2$, such that $\Ex[U] = \Ex[V] = 1$, $U_0 - V_0 > 0.9$, and $\Ex[U^j] = \Ex[V^j]$ for every $j \in [L]$.
\end{lemma}

The following lemma is crucial for our proof. Its proof uses the moment matching technique~\cite{WY16polyapprox,JiaoHW18,WY2019chebyshev,HanLec19}, and can be found in Appendix~\ref{app:TV-Poi}.

\begin{lemma}\label{lm:key-bound-2}
	Let $U,V$ be two random variables supported on $[0,\Lambda]$ such that $\Ex[U^j] = \Ex[V^j]$ for all $j \in \{1,2,\dotsc,L\}$, where $L \ge 1$. Let $D \in \mathbb{N}$ and $\vtheta,\vlambda \in (\R^{\ge 0})^{D}$ such that $\|\vtheta\|_1 = 1$. Let $\calD_{\vtheta}$ be the distribution over $[D]$ corresponding to $\vtheta$. Suppose that
	\[
	\Pr_{i \leftarrow \calD_{\vtheta}} [\vlambda_i \ge 2 \Lambda^2 \cdot \vtheta_i] \ge 1 -  \frac{1}{2\Lambda}.
	\]
	Then,
	\[
	\|\Ex[\vPoi(U \vtheta + \vlambda)] - \Ex[\vPoi(V \vtheta + \vlambda)] \|_{TV}^2 \le \frac{1}{L!}.
	\]
\end{lemma}

Finally, we need an observation that for a $\localDP$ protocol $P$ solving $\DE$, we can assume without loss of generality that the analyzer of $P$ only sees the histogram of the messages.

\begin{lemma}\label{lm:histogram-enough}
	For any $\localDP$ protocol $P = (R,A)$ for $\DE$, there exists an analyzer $A'$ which only sees the histogram of the messages, and achieves the same accuracy and error as that of $A$.
\end{lemma}
\begin{proof}
	Let $n$ be the number of users. Given the histogram, $A'$ first constructs a sequence of messages $S \in \calM^n$ consistent with the histogram. Then, it applies a random permutation $\pi\colon [n] \to [n]$ to $S$ to obtain a new sequence $\pi(S)$. Finally, it simply outputs $A(\pi(S))$.
	
	Note that applying a random permutation on the messages is equivalent to applying a random permutation on the user inputs in the dataset. Hence, the new protocol $P' = (R,A')$ is equivalent to running $P$ on a random permutation of the dataset. The lemma follows from the fact that a random permutation does not change the number of distinct elements.
\end{proof}

\subsubsection{Construction of the Hard Dataset Distributions} 

In the rest of the section, we use $n$ to denote a parameter controlling the number of users. The actual number of users $\bar{n}$ will be later set to a number in the interval $[n,2n]$. In the following, we fix a randomizer $R\colon \calX \to \calM$ which is $(\eps_R,\delta)$-$\localDP$ on $\bar{n}$ users,
% ($\bar{n} \in [n,2n]$ and will be set later),
for some $\eps_R = \Theta(\bar{n}/\log^6 \bar{n})$ to be specified later. Before constructing our two hard distributions over datasets, we set some parameters that will be used in the construction:

\begin{itemize}
	\item We set $L = \log n$ and note that $\frac{1}{L!} \le 1/n^4$ for large enough $n$. 
	
	\item Applying Lemma~\ref{lm:construction-two-RV}, for $\Lambda = \Theta(L^2) = \Theta(\log^2 n)$, we obtain two random variables $U$ and $V$ supported on $\{0\} \cup [1,\Lambda]$, such that $\Ex[U] = \Ex[V] = 1$, $U_0 - V_0 > 0.9$, and $\Ex[U^j] = \Ex[V^j]$ for every $j \in [L]$.
	
	\item We set $\Gamma = 8 \Lambda^2 = \Theta(\log^4 n)$ and $D = n / \Gamma = \Theta(n/\log^4 n)$.  We are going to construct instances where inputs are from the universe $\calX = [D]$. 
	
	\item We set $\bar{n} = n + D - n^{0.99}$.
	
	\item Let $W = (\log^2 n) \cdot 4 \Lambda^2$. We set $\eps_R$ so that $n / 2^{\eps_R} = W = \Theta(\log^6 n)$. Hence, $R$ is $(\ln(n/W),n^{-\omega(1)})$-$\localDP$.
\end{itemize}

%Let $\eps_1 > 0$ be a small constant to be specified later.

%STOPED HERE

Now, for a distribution $U$ over $\mathbb{R}^{\ge 0}$ and a non-empty subset $E$ of $[D]$, the dataset distribution $\calD^{U,E}$ is constructed as follows:

\begin{enumerate}
	\item For each $i \in [D]$, we draw $\lambda_i \leftarrow U$, and $n_i \leftarrow \Poi(\lambda_i)$, and add $n_i$ many users with input $i$. 
	
	\item For each $j \in E$, we draw $m_j \leftarrow \Poi(n/|E|)$, and add $m_j$ many users with input $j$.\footnote{Note that we here use a Poisson distribution slightly differently from the construction in Section~\ref{sec:tech_overv} (namely, $\Poi(n/|E|)$ instead of $\Poi((n-D)/|E|)$) in order to simply the later calculations.}
	%We select a random subset $E$ of $[D]$ by including each item independently with probability $\eps_1$, 
\end{enumerate}

For clarity of exposition, we will use the histogram of the protocol to denote the histogram of the messages in the transcript of the protocol. Our goal is to show that for some ``good'' subset $E \subseteq [D]$, the following hold:

\begin{enumerate}
	\item The distributions of the histogram of the protocol under $\calD^{U,E}$ and $\calD^{V,E}$ are very close.
	
	\item With high probability, the number of distinct elements in datasets from $\calD^{U,E}$ is $\Omega(D)$ smaller than in datasets from $\calD^{V,E}$.
\end{enumerate}

Clearly, given the above two conditions and Lemma~\ref{lm:histogram-enough}, no protocol with randomizer $R$ can estimate the number of distinct elements within $o(D)$ error and with constant probability.

\subsubsection{Conditions on a Good Subset $E$}

Given a subset $E$, we let $\vnu^E = \sum_{i \in E} R(i) \cdot \frac{n}{|E|}$. We also set $\vmu = \sum_{i \in [D]} R(i)$. We now specify our conditions on a subset $E \subset [D]$ being good. Let $\eps_1 = 0.01$. We say $E$ is good if the following two conditions hold:

\begin{enumerate}
	\item $0 < |E| < 2\eps_1 \cdot |D|$.
	
	\item For each $i \in [D]$,
	\[
	\Pr_{z \leftarrow R(i)} [\vnu^E_z \ge 2\Lambda^2 \cdot \vmu_z] \ge 1 - 1/2\Lambda.
	\]
\end{enumerate}

%Given a particular subset $E$ fo $D$, we can consider the distribution $\calD^{U,E}$ and $\calD^{V,E}$ which are $\calD^{U}$ and $\calD^{V}$ conditioned on the subset selected in the second phase being $E$.

We claim that a good subset $E$ exists. In fact, we give a probabilistic construction of $E$ that succeeds with high probability:

\begin{lemma}\label{lm:E-is-good-whp}
	If we include each element of $i \in [D]$ in $E$ independently with probability $\eps_1$, then $E$ is good with probability at least $1 - n^{-\omega(1)}$.
\end{lemma}

\subsubsection{The Lower Bound}

Before proving Lemma~\ref{lm:E-is-good-whp}, we show that for a good $E$, the distributions $\calD^{U,E}$ and $\calD^{V,E}$ satisfy our desired properties, and thereby imply our $\localDP$ lower bound. For a dataset distribution $\calD$, we use $\Hist_R(\calD)$ to denote the corresponding distribution of the histogram of the transcript, if all users apply the randomizer $R$. For a dataset $I$, we use $\DE(I)$ to denote the number of distinct elements in it.

\begin{lemma}\label{lm:good-E-conditions}
	For a good subset $E$ of $[D]$, the following hold:
	
	\begin{enumerate}
		\item We have that
		\[
		\|\Hist_R(\calD^{U,E}) - \Hist_R(\calD^{V,E})\|_{TV} \le 1/n.
		\]
		
		\item There are two constants $\tau_1 < \tau_2$ such that
		\[
		\Pr_{I_1 \leftarrow \calD^{U,E}}[\DE(I_1) < \tau_1 \cdot D] \ge 1 - n^{-\omega(1)},
		\]
		and
		\[
		\Pr_{I_2 \leftarrow \calD^{V,E}}[\DE(I_2) > \tau_2 \cdot D] \ge 1 - n^{-\omega(1)}.
		\]
	\end{enumerate}
\end{lemma}
\begin{proof}\item
	\paragraph*{Proof of Item~(1).} In the following, we use $\vnu$ to denote $\vnu^E$ for simplicity. We first construct $D$ vectors $\{ \vnu^{(i)} \}_{i \in [D]}$ as follows: for each $z \in \calM$, if $\vnu_z \ge 2\Lambda^2 \cdot \vmu_z$, then for all $i \in [D]$ we set $(\vnu^{(i)})_z = R(i)_z \cdot 2\Lambda^2$, otherwise we set $(\vnu^{(i)})_z = 0$ for all $i \in [D]$. Note that for each $z$, we have that $\left(\sum_{i \in [D]} \vnu^{(i)}\right)_z \le \vnu_z$. Let $\vnu^{(0)} := \vnu - \left(\sum_{i \in [D]} \vnu^{(i)}\right)$. By definition, it follows that $\vnu^{(0)}$ is a non-negative vector. Now, $\Hist_R(\calD^{U,E})$ and $\Hist_R(\calD^{V,E})$ can be seen as distributions over histograms in $\mathbb{N}^{\calM}$. Let $X_1,X_2,\dotsc,X_D$ be $D$ independent random variables distributed as $U$. By the construction of $\calD^{U,E}$, we have that
	\begin{align*}
	\Hist_R(\calD^{U,E}) =& \vPoi(\vnu) + \sum_{i=1}^{D} \vPoi(X_i \cdot R(i))\\
	=& \vPoi(\vnu^{(0)}) + \sum_{i=1}^{D} \vPoi(X_i \cdot R(i) + \vnu^{(i)}).
	\end{align*}
	Similarly, let $Y_1,Y_2,\dotsc,Y_D$ be $D$ independent random variables distributed as $V$. We have that
	\[
	\Hist_R(\calD^{V,E}) = \vPoi(\vnu^{(0)}) + \sum_{i=1}^{D} \vPoi(Y_i \cdot R(i) + \vnu^{(i)}).
	\]
	Since for each $i \in [D]$, we have that
	\[
	\Pr_{z \leftarrow R(i)} [(\vnu^{(i)})_z \ge 2\Lambda^2 \cdot R(i)_z] \ge 1 - 1/2\Lambda.
	\]
	Applying Lemma~\ref{lm:key-bound-2}, for each $i \in [D]$, we have that
	\[
	\|\vPoi(X_i \cdot R(i) + \vnu^{(i)}) - \vPoi(Y_i \cdot R(i) + \vnu^{(i)})\|_{TV} \le \left( \frac{1}{L!} \right)^{1/2} \le 1/n^2.
	\]
	Therefore,
	\[
	\| \calD^{U,E} - \calD^{V,E} \|_{TV} \le \sum_{i=1}^{D} \| \Poi(U \cdot R(i) + \vnu^{(i)}) - \Poi(V \cdot R(i) + \vnu^{(i)})\|_{TV} \le D \cdot \left( \frac{1}{L!} \right)^{1/2} \le 1/n.
	\]

	\paragraph*{Proof of Item~(2).} Let $\gamma_U = \Ex[e^{-U}]$ and $\gamma_V = \Ex[e^{-V}]$.
	%STOPED HERE
	By Lemma~\ref{lm:construction-two-RV}, we have that $U_0 \ge 0.9$, $V_0 \le 0.1$, and $U,V$ are supported on $\{0\} \cup [1,\Lambda]$. Hence, it follows that $\gamma_U \ge U_0 \ge 0.9$ and $\gamma_V \le e^{-1}(1-V_0) + V_0 \cdot e^{-1} \le 0.5$. 
	
	Now, consider the construction of $\calD^{U,E}$. For every $i \in [D]$, at least one user with input $i$ is added to the dataset during phase (1) with probability $1 - \gamma_U$. Moreover, these events are mutually independent. Hence, by a simple Chernoff bound, with probability at least $1- n^{-\omega(1)}$, the number of distinct elements in the dataset after phase (1) is no greater than $(1 - \gamma_U + 0.01) \cdot D$. Since the second phase can add at most $|E| = 0.02 D$ many distinct elements, we can set $\tau_1 = 1 - \gamma_U + 0.03$.
	
	Similarly, for instances generated from $\calD^{V,E}$, with probability at least $1- n^{-\omega(1)}$, the number of distinct elements in the dataset after phase (1) is at least $(1 - \gamma_V - 0.01) \cdot D$. We can set $\tau_2 = 1 - \gamma_V - 0.01$.
	
	By our condition on $\gamma_U$ and $\gamma_V$, we have that $\tau_2 > \tau_1$, which completes the proof.
\end{proof}

We are now ready to prove Theorem~\ref{theo:LDP-LB-strong-private-coin}. One complication is that datasets from $\calD^{U,E}$ and $\calD^{V,E}$ may not have the same number of users. We address this issue by ``throwing out'' extra users randomly and obtain distributions over datasets with exactly $\bar{n}$ many users.

\begin{proofof}{Theorem~\ref{theo:LDP-LB-strong-private-coin}}
	Consider the $\calD^{U,E}$ and $\calD^{V,E}$. By a simple Chernoff bound, we have that with probability at least $1-n^{-\omega(1)}$, the number of users lies in $[n + D - n^{0.99},n+D+n^{0.99}]$.
	
	Recall that $\bar{n} = n + D - n^{0.99}$. We construct the distribution $\bar{\calD}^{U,E}$ as follows: to generate a dataset from $\bar{\calD}^{U,E}$, we take a sample dataset $I$ from $\calD^{U,E}$, and if there are $n_I > \bar{n}$ users in $I$, we delete $n_I - \bar{n}$ users uniformly at random, and output $I$. We similarly construct another distribution $\bar{\calD}^{V,E}$. Note that with probability at least $1 - n^{-\omega(1)}$, we delete at most $2n^{0.99}$ users in the construction of $\bar{\calD}^{U,E}$ (as well as in that of $\bar{\calD}^{V,E}$).
	
	Now, both $\bar{\calD}^{U,E}$ and $\bar{\calD}^{V,E}$ output datasets with exactly $\bar{n}$ users with probability $1 - n^{-\omega(1)}$. This means that if there is a protocol solving \DE with $\bar{n}$ users with error $o(D)$, then by Lemma~\ref{lm:histogram-enough}, Item~(2) of Lemma~\ref{lm:good-E-conditions} and since $2n^{0.99} = o(D)$, the analyzer of the protocol should be able to distinguish $\Hist_R(\bar{\calD}^{U,E})$ and $\Hist_R(\bar{\calD}^{V,E})$ with at least a constant probability. Therefore, we have that
	\[
	\|\Hist_R(\bar{\calD}^{U,E}) - \Hist_R(\bar{\calD}^{V,E})\|_{TV} = \Omega(1).
	\]
	
	On the other hand, $\Hist_R(\bar{\calD}^{U,E})$ (respectively, $\Hist_R(\bar{\calD}^{V,E})$) can also be constructed by taking a sample from $\Hist_R(\calD^{U,E})$ (respectively, $\Hist_R(\calD^{V,E})$) and throwing out some random messages until at most $\bar{n}$ messages remain. Since post-processing does not increase statistical distance, by Item~(1) of Lemma~\ref{lm:good-E-conditions}, we have that
	\[
	\|\Hist_R(\bar{\calD}^{U,E}) - \Hist_R(\bar{\calD}^{V,E})\|_{TV} \le \|\Hist_R(\calD^{U,E}) - \Hist_R(\calD^{V,E})\|_{TV} \le 1/n,
	\]
	a contradiction.
\end{proofof}

\subsubsection{A Probabilistic Construction of Good $E$}

We need the following proposition for the proof of Lemma~\ref{lm:E-is-good-whp}.

\begin{prop}\label{prop:simple-prop}
	Let $R\colon \calX \to \calM$ be an $(\eps,\delta)$-$\localDP$ randomizer. For every $i,j \in \calX$, it follows that
	\[
	\Pr_{z \leftarrow R(i)}[R(i)_z \ge 2e^{\eps} \cdot R(j)_z] \le 2\delta.
	\]
\end{prop}
\begin{proof}
	%Suppose otherwise $\Pr_{z \leftarrow R(i)}[R(i)_z \ge 2e^{\eps} \cdot R(j)_z] > 2\delta$. 
	Let $\calT$ be the set $\{ z : R(i)_z \ge 2e^{\eps} \cdot R(j)_z \wedge z \in \calM \}$. Since $R$ is $(\eps,\delta)$-$\localDP$, it follows that
	\[
	R(i)_{\calT} \le e^{\eps} \cdot R(j)_{\calT} + \delta.
	\]
	
	By the definition of the set $\calT$, it follows that 
	\[
	R(j)_{\calT} = \sum_{z \in \calT} R(j)_z \le \frac{1}{2 e^{\eps}} \cdot \sum_{z \in \calT} R(i)_z = \frac{1}{2 e^{\eps}} \cdot R(i)_{\calT}.
	\]
	
	Putting the above two inequalities together, we have
	\[
	R(i)_{\calT} \le \frac{1}{2} \cdot R(i)_{\calT} + \delta,
	\]  
	which in turn implies that
	\[
	 \Pr_{z \leftarrow R(i)}[R(i)_z \ge 2e^{\eps} \cdot R(j)_z] = R(i)_{\calT} \le 2\delta. \qedhere
	\]
\end{proof}

Finally, we prove Lemma~\ref{lm:E-is-good-whp} (restated below).

\begin{reminder}{Lemma~\ref{lm:E-is-good-whp}.}
	If we include each element $i \in [D]$ in $E$ independently with probability $\eps_1 = 0.01$, then $E$ is good with probability at least $1 - n^{-\omega(1)}$.
\end{reminder}

\begin{proof}	
	\newcommand{\eLight}{\calE_{\sf light}}
	\newcommand{\eHeavy}{\calE_{\sf heavy}}
	\newcommand{\eSize}{\calE_{\sf size}}
	Let $\eSize$ be the event that $0 < |E| < 2 \eps_1 \cdot |D|$. By a simple Chernoff bound, it follows that
	\[
	\Pr_{E}[\eSize] \ge 1 - \exp(-\Omega(|D|)) \ge 1 - n^{-\omega(1)}.
	\]
	Therefore, the first condition for $E$ being good is satisfied with probability $1 - n^{-\omega(1)}$. In the following, we will condition on the event $\eSize$.
	
	Recall that $\vnu^E = \sum_{i \in E} R(i) \cdot \frac{n}{|E|}$ and $\vmu = \sum_{i \in [D]} R(i)$. In the rest of the proof, we will focus on the second condition for $E$ being good, namely that for each $i \in [D]$, it is the case that 
	\[
	\Pr_{z \leftarrow R(i)} [\vnu^E_z \ge 2\Lambda^2 \cdot \vmu_z] \ge 1 - 1/2\Lambda.
	\]
	
	In the following, we fix $i \in [D]$, and show that the previous inequality holds for $i$ with high probability. Therefore, we can then conclude that $E$ is good with high probability by a union bound.
	
	We also let $\vmx \in \R^{\calM}$ be such that $\vmx_z = \max_{i \in [D]} R(i)_z$ for all $z \in \calM$. Now, for $z \in \calM$, if $\vmu_z \le \vmx_z \cdot \log^2 n$, we say $z$ is \emph{light}; otherwise we say $z$ is \emph{heavy}.

	We define
	
	\[
	\eLight := \left[ \Pr_{z \leftarrow R(i)} [\vnu^E_z < 2\Lambda^2 \cdot \vmu_z \mbox{ and } \text{$z$ is light}] \le 1/4\Lambda\right],
	\]
	and
	\[
	\eHeavy := \left[ \Pr_{z \leftarrow R(i)} [\vnu^E_z < 2\Lambda^2 \cdot \vmu_z \mbox{ and } \text{$z$ is heavy}] \le 1/4\Lambda\right].
	\]
	It suffices to show that both $\Pr_E[\eLight |\eSize]$ and $\Pr_E[\eHeavy |\eSize]$ are very large. Note that $\vnu^E$ is not defined when $|E| = 0$. But since we only care about $\Pr_E[\eLight |\eSize]$ and $\Pr_E[\eHeavy |\eSize]$, this corner case is excluded by conditioning on $\eSize$.
	
	\paragraph*{Proving that $\Pr_E[\eLight |\eSize] = 1$.} By Proposition~\ref{prop:simple-prop} and the fact that $R$ is $(\ln(n/W),n^{-\omega(1)})$-$\localDP$, for every $x \in [D]$
	\[
	\Pr_{z \leftarrow R(i)} [ R(i)_z \ge (2 n/W) \cdot R(x)_z] \le n^{-\omega(1)}.
	\]
	
	By a union bound over all elements in $E$, we have that
	\[
	\Pr_{z \leftarrow R(i)} [ R(i)_z \ge (2 n/W) \cdot \frac{\vnu_z}{n}] \le n^{-\omega(1)},
	\]
	which is equivalent to
	\[
	\Pr_{z \leftarrow R(i)} [ \vnu_z \le W/2 \cdot R(i)_z] \le n^{-\omega(1)}.
	\]
	Similarly, for $j \in [D]$, we also have that
	\[
	\Pr_{z \leftarrow R(j)} [ \vnu_z \le W/2 \cdot R(j)_z] \le n^{-\omega(1)}.
	\]
	Again since $R$ is $(\ln(n/W),n^{-\omega(1)})$-$\localDP$, we have that
	\begin{align*}
	&
	\Pr_{z \leftarrow R(i)} [ \vnu_z \le W/2 \cdot R(j)_z] 
	= \Ex_{z \leftarrow R(j)} \frac{R(i)_z}{R(j)_z} \cdot \mathbb{1}[ \vnu_z \le W/2 \cdot R(j)_z]\\
	\le& \Ex_{z \leftarrow R(j)} (2n/W) \cdot \mathbb{1}\left[\frac{R(i)_z}{R(j)_z} \le 2n/W \right]\cdot \mathbb{1}[ \vnu_z \le W/2 \cdot R(j)_z] + \Ex_{z \leftarrow R(j)} \frac{R(i)_z}{R(j)_z} \cdot \mathbb{1}\left[\frac{R(i)_z}{R(j)_z} > 2n/W \right]\\
	\le& (2n /W) \cdot n^{-\omega(1)} + \Ex_{z \leftarrow R(i)} \mathbb{1}\left[\frac{R(i)_z}{R(j)_z} > 2n/W \right]\\
	\le& n^{-\omega(1)}.
	\end{align*}
	Therefore, by a union bound over $j \in [D]$,
	\[
	\Pr_{z \leftarrow R(i)} [ \vnu_z \le W/2 \cdot \vmx_z] \le n^{-\omega(1)}.
	\]
	
	Now we are ready to prove that $\Pr_{E}[\eLight | \eSize] = 1$. We will show that $\eLight$ holds for every nonempty $E$. We have that
	\begin{align*}
	 &\Pr_{z \leftarrow R(i)} [\vnu^E_z < 2\Lambda^2 \cdot \vmu_z \mbox{ and } \text{$z$ is light}]\\
  \le&\Pr_{z \leftarrow R(i)} [\vnu^E_z < 2\Lambda^2 \cdot \vmu_z \mbox{ and } \text{$z$ is light} \mbox{ and } \vnu_z > W/2 \cdot \vmx_z] + \Pr_{z \leftarrow R(i)} [ \vnu_z \le W/2 \cdot \vmx_z]\\
  \le&\Pr_{z \leftarrow R(i)} [\vnu^E_z < 2\Lambda^2 \cdot \vmu_z \mbox{ and } \vmu_z \le \vmx_z \cdot \log^2 n \mbox{ and } \vnu_z > W/2 \cdot \vmx_z] + n^{-\omega(1)} \tag{z is light implies $\vmu_z \le \vmx_z \cdot \log^2 n$}\\
  \le& n^{-\omega(1)}.
	\end{align*}
	The last inequality follows from the fact that $\vmu_z \le \vmx_z \cdot \log^2 n$ and $\vnu_z > W/2 \cdot \vmx_z$ together imply that $\vnu_z > W/2 \cdot \vmx_z \ge \frac{W/2}{\log^2 n} \vmu_z \ge 2 \Lambda^2 \cdot \vmu_z$ (recall that $W = \log^2 4 \Lambda^2$). Hence $\Pr_{z \leftarrow R(i)} [\vnu^E_z < 2\Lambda^2 \cdot \vmu_z \mbox{ and } \vmu_z \le \vmx_z \cdot \log^2 n \mbox{ and } \vnu_z > W/2 \cdot \vmx_z] = 0$ as the three inequalities cannot be simultaneously satisfied.
	%With high probability over $z \leftarrow R(i)$, we have that $\vnu_z > W/2 \cdot \vmx_z$. If we further assume $z$ is light, it follows that $\vnu_z > W/2 \cdot \vmx_z \ge \frac{W/2}{\log^2 n} \vmu_z$. By our choice of $W$, we have that $\frac{W/2}{\log^2 n} \ge 2\Lambda^2$. Therefore, when both $z$ is light and $\vnu_z \le W/2 \cdot \vmx_z$, it follows $\vnu_z > 2 \Lambda^2 \vmu_z$.
	
	\paragraph*{Proving that $\Pr[\eHeavy | \eSize]$ is large.}
	
	Now, for a heavy $z$, we have that $\vmu_z \ge \vmx_z \cdot \log^2 n$. In particular, fix a heavy $z$, and define the random variable $X_i := \mathbb{1}[i \in E] \cdot R(i)_z$ for each $i \in [D]$. Note that the $X_i$'s are independent variables over $[0,R(i)_z]$ and $\Ex\left[\sum_{i\in [D]} X_i\right] = \eps_1 \cdot \vmu_z$. Letting $S = \sum_{i \in [D]} X_i$, by Hoeffding's inequality, we have that
	\[
	\Pr_{E}[ |S - \Ex[S]| \ge \frac{1}{2} \Ex[S]] \le 2 \exp\left(- \frac{2 \cdot (\frac{1}{2} \Ex[S])^2}{\sum_{i \in [D]} R(i)_z^2}\right).
	\]
	Note that
	\[
	\sum_{i \in [D]} R(i)_z^2 \le \sum_{i \in [D]} R(i)_z \cdot \vmx_z \le \vmu_z \cdot \vmx_z.
	\]
	Plugging in, it follows that
	\[
	\Pr_{E}[ S \le \eps_1 \vmu_z / 2 ] \le 2 \exp\left(- \frac{\eps_1^2\vmu_z^2 / 2}{\vmu_z \cdot \vmx_z}\right) \le  2 \exp(- \eps_1^2 / 2 \cdot \log^2 n) \le n^{-\omega(1)}.
	\]
	Note that $\vnu^E_z = S \cdot \frac{n}{|E|}$, and that $|E| \le 2\eps_1 D$ with probability at least $1 - n^{-\omega(1)}$. By a union bound, we have that $\vnu^E_z \ge \eps_1 \vmu_z / 2 \cdot \frac{n}{2\eps_1 D} = \vmu_z \cdot \Gamma / 4$ with probability at least $1 - n^{-\omega(1)}$.
	Noting that $\Gamma / 4 = 2\Lambda^2$, we have that
	\[
	\Ex_{E}\left[\Pr_{z \leftarrow R(i)} [\vnu^E_z < 2\Lambda^2 \cdot \vmu_z \wedge \text{$z$ is heavy}] \right] \le n^{-\omega(1)}.
	\]
	Recall that $\Pr[\eSize] \ge 1 - n^{-\omega(1)}$. By Markov's inequality, we have $\Pr[\eHeavy|\eSize] \ge 1 - n^{-\omega(1)}$.
\end{proof}

\subsection{$\shuffledDP^1$ Implies $\localDP$ with Stronger Privacy Bound}\label{sec:SDP-to-LDP-stronger}

In this section, we prove a stronger connection between $\shuffledDP^1$ and $\localDP$ than previously known. Together with Theorem~\ref{theo:LDP-LB-strong-private-coin}, it implies the private-coin version of Theorem~\ref{theo:single-message-LB}. 

We first need a technical lemma which gives a lower bound on the hockey stick divergence between $\Ber(\alpha) + \Bin(m,p)$ and $\Ber(\beta) + \Bin(m,p)$. We defer its proof to Appendix~\ref{app:missing-proof-HS}.

\begin{lemma}\label{lm:key-HS-lowb}
	There exists an absolute constant $c_0$ such that, for every integer $m \ge 1$, three reals $\alpha,\beta,\eps > 0$ such that $\alpha > e^{\eps}\beta$, letting $\Delta = \alpha - e^{\eps}\beta$ and supposing $4 \frac{e^{\eps}}{\Delta} \beta < 1/2$, it holds that
	\[
	d_{\eps}(\Ber(\alpha) + \Bin(m,\beta)|| \Ber(\beta) + \Bin(m,\beta)) \ge \Delta \cdot \frac{1}{2\sqrt{2m}} \cdot \exp\left(-c_0 \cdot m \cdot \frac{e^{\eps}}{\Delta} \beta \cdot \left[\log(\Delta^{-1}) + 1\right]\right).
	\]
\end{lemma}

We are now ready to prove the main lemma of this subsection.

\begin{lemma}\label{lm:single-msg-SDP-to-LDP}
    For all $\eps = O(1)$, there is a constant $c > 0$ such that for all $\delta \le \delta_0 \le 1/n$ if the randomizer $R$ is $(\eps,\delta)$-$\shuffledDP^1$ on $n$ users, then $R$ is $\left(\ln\left(n \Big/\frac{c \ln \delta^{-1}}{\ln \delta_0^{-1}}\right),\delta_0\right)$-$\localDP$.
\end{lemma}
\begin{proof}
	Let $\eps = O(1)$. Note that we can assume that $\delta \le \delta_0^{\omega(1)} \le n^{-\omega(1)}$, as otherwise $\frac{\ln \delta^{-1}}{\ln \delta_0^{-1}} \le O(1)$ and in this case the theorem follows directly from the fact that $R$ is $(\eps + \ln n,\delta)$-$\localDP$~\cite{Cheu18}. 
	
	Suppose that $R$ is $(\eps,\delta)$-$\shuffledDP^1$ on $n$ users.
	Let $c$ be a constant to be fixed later, $D = n \Big/\frac{c \ln \delta^{-1}}{\ln \delta_0^{-1}}$ and $E \subseteq \calM$ be an event. Our goal is to show that
	\[
	R(x)_E \le R(y)_E \cdot D + \delta_0,
	\]
	for all $x,y \in \calX$.
	
	Fix two $x,y \in \calX$. Let $\alpha = R(x)_E$ and $\beta = R(y)_E$. Note that without loss of generality we can assume that $\beta \le 1/D$, as otherwise clearly $\alpha \le 1 \le D \cdot \beta + \delta_0$.
	
	Let $W_1 = x y^{n-1} $ and $W_2 = y^n$ be two neighboring datasets, and $X,Y$ be the random variables corresponding to the number of occurrences of the event $E$ in the transcript, when running the protocol with randomizer $R$ on datasets $W_1$ and $W_2$, respectively.

	From the assumption that $R$ is $(\eps,\delta)$-$\shuffledDP^1$, we have $d_{\eps}(\Hist_R(W_1)||\Hist_R(W_2))  \le \delta$. Then, by the post-processing property of DP (Lemma~\ref{lem:post_process}), it follows that $d_{\eps}(X||Y) \le \delta$.
	
	We have that
	\[
	X = \Ber(\alpha) + \Bin(n-1,\beta)\text{ and } 
	Y = \Ber(\beta) + \Bin(n-1,\beta).
	\] 
	
	The goal now is to show that if $\alpha > \beta \cdot D + \delta_0$, then $X$ and $Y$ do not satisfy $(\eps,\delta)$-DP (i.e., $d_{\eps}(X||Y) > \delta$), thus obtaining a contradiction.
	
	Now, assume that $\alpha > D \cdot \beta + \delta_0$. Since $\eps = O(1)$, we have that $\Delta = \alpha - e^{\eps} \beta > \frac{D}{2} \cdot \beta + \delta_0$. Note that $4 \frac{e^{\eps}}{\Delta} \beta = O(D^{-1}) < 1/2$.
	
	Letting $m = n- 1$ and applying Lemma~\ref{lm:key-HS-lowb} for a universal constant $c_0$, it follows that
	\[
	d_{\eps}(\Ber(\alpha) + \Bin(m,\beta)|| \Ber(\beta) + \Bin(m,\beta)) 
	\ge \Delta \cdot \frac{1}{2\sqrt{2m}} \cdot \exp\left(-c_0 \cdot m \cdot \frac{e^{\eps}}{\Delta} \beta \cdot \left[\log(\Delta^{-1}) + 1\right]\right).
	\]
	Noting that $\eps = O(1)$, we have that
	\[
	m \cdot \frac{e^{\eps}}{\Delta} \beta \cdot \left[\log(\Delta^{-1}) + 1\right] \le O\left(m \frac{1}{D \beta} \cdot \beta \cdot \log \delta_0^{-1}\right) = O(c \ln \delta^{-1}). 
	\]
	We now set the constant $c$ to be small enough so that
	\[
	c_0 \cdot m \cdot \frac{e^{\eps}}{\Delta} \beta \cdot \left[\log(\Delta^{-1}) + 1\right] \le \frac{1}{2} \ln \delta^{-1}.
	\]
	Plugging in and recalling that $\delta \le \delta_0^{\omega(1)} \le n^{-\omega(1)}$, we get that
	\[
	d_{\eps}(X||Y) = d_{\eps}(\Ber(\alpha) + \Bin(m,\beta)|| \Ber(\beta) + \Bin(m,\beta)) \ge \delta_0 \cdot \frac{1}{2\sqrt{2m}} \sqrt{\delta} > \delta,
	\]
	a contradiction.
	%which contradicts the assumption that $R$ is $(\eps,\delta)$-$\shuffledDP^1$. 
\end{proof}

Finally, we are ready to prove our $\shuffledDP^1$ lower bound for $\DE$ in the private-coin case.

\begin{theorem}\label{theo:single-message-LB-private-coin}
	For all $\eps = O(1)$, there are $\delta = 2^{-\Theta(\log^8 n)}$ and $D = \Theta(n/\log^4 n)$ such that no private-coin $(\eps,\delta)$-$\shuffledDP^1$ protocol on $n$ users can solve $\DE_{n,D}$ with error $o(D)$ and probability at least $0.99$.
\end{theorem}
\begin{proof}
	We set $\delta_0 = 2^{-\log^2 n}$ and $\delta = 2^{-c\log^8 n}$ for a constant $c$ to be specified shortly. By Theorem~\ref{theo:LDP-LB-strong-private-coin}, it follows that the corresponding randomizer $R$ is $(\ln(\Theta(n/c\log^6 n),n^{-\omega(1)})$-$\localDP$. Setting $c$ to be sufficiently large and combining with Theorem~\ref{theo:LDP-LB-strong-private-coin} completes the proof.
\end{proof}

\subsection{Generalizing to Public-Coin Protocols}\label{sec:adaption-public-coin}

Finally, we generalize our proof for the private-coin case to the public-coin case, and prove Theorem~\ref{theo:single-message-LB} (restated below).

\begin{reminder}{Theorem~\ref{theo:single-message-LB}.}
	For all $\eps = O(1)$, there are $\delta = 2^{-\Theta(\log^8 n)}$ and $D = \Theta(n/\log^4 n)$ such that no public-coin $(\eps,\delta)$-$\shuffledDP^1$ protocol on $n$ users can solve $\DE_{n,D}$ with error $o(D)$ and probability at least $0.99$.
\end{reminder}

Fix $R$ to be a public-coin randomizer with public randomness $\alpha$ from distribution $\calDpb$. We first generalize Lemma~\ref{lm:single-msg-SDP-to-LDP}, and show that if $R$ is $\shuffledDP^1$, then with high probability over $\alpha \leftarrow \calDpb$, $R_{\alpha}$ satisfies the similar $\localDP$ guarantee as in Lemma~\ref{lm:single-msg-SDP-to-LDP}.

\begin{lemma}\label{lm:single-msg-SDP-to-LDP-pbcoin}
	For all $\eps = O(1)$, there is a constant $c > 0$ such that for all $\delta \le \delta_0 \le 1/n$ if the public-coin randomizer $R\colon \calX \to \calM$ with public randomness distribution $\calDpb$ is $(\eps,\delta)$-$\shuffledDP^1$ on $n$ users, and if $|\calX| \le n$, then with probability at least $1-\delta_0$ over $\alpha \leftarrow \calDpb$, it is the case that $R_{\alpha}$ is $\left(\ln\left(n \Big/\frac{c \ln \delta^{-1}}{\ln \delta_0^{-1}}\right),\delta_0\right)$-$\localDP$.
\end{lemma}
\begin{proofsketch}
	\newcommand{\Wgood}{W_{good}}
	
	Similar to the proof of Lemma~\ref{lm:single-msg-SDP-to-LDP}, we can assume that $\delta \le \delta_0^{\omega(1)} \le n^{-\omega(1)}$ without loss of generality.
	
	By the definition of public-coin $(\eps,\delta)$-\DP in the shuffle model, for every two neighboring datasets $W_1$ and $W_2$, we have that
	\[
	\Ex_{\alpha \leftarrow \calDpb} [d_{\eps}(\Hist_{R_\alpha}(W_1)||\Hist_{R_\alpha}(W_2))] \le \delta.
	\]
	
	Observe that the proof of Lemma~\ref{lm:single-msg-SDP-to-LDP} only considers the $|\calX|^2$ pairs of neighboring datasets of the form $W_1 = x y^{n-1}$ and $W_2 = y^n$ for all $x,y \in \calX$. We use $\Wgood$ to denote the set of such pairs.
	
	Using the assumption that $|\calX| \le n$, we have that
	\[
	\Ex_{\alpha \leftarrow \calDpb} \sum_{(W_1,W_2) \in \Wgood} [d_{\eps}(\Hist_{R_\alpha}(W_1)||\Hist_{R_\alpha}(W_2))] \le |\calX|^2 \delta \le n^2 \delta.
	\]
	Thus, by Markov's inequality, with probability at least $1 - \delta_0$ over $\alpha \leftarrow \calDpb$, we have that
	\[
	\sum_{(W_1,W_2) \in \Wgood} [d_{\eps}(\Hist_{R_\alpha}(W_1)||\Hist_{R_\alpha}(W_2))] \le n^2 \delta /\delta_0 \le \delta^{0.9},
	\]
 where the last inequality follows from our assumption that $\delta \le \delta_0^{\omega(1)} \le n^{-\omega(1)}$. We say an $\alpha$ is \emph{good} if it satisfies the above inequality. In particular, for all good $\alpha$ and all pairs $(W_1,W_2) \in \Wgood$, we have that
	\[
		d_{\eps}(\Hist_{R_\alpha}(W_1)||\Hist_{R_\alpha}(W_2)) \le \delta^{0.9}.
	\]
	The proof of Lemma~\ref{lm:single-msg-SDP-to-LDP} then implies that $R_{\alpha}$ is  $\left(\ln\left(n \Big/\frac{c \ln \delta^{-1}}{\ln \delta_0^{-1}}\right),\delta_0\right)$-$\localDP$, for a constant $c$ depending on $\eps$.
\end{proofsketch}

Now we are ready to prove Theorem~\ref{theo:single-message-LB}.
\begin{proofof}{Theorem~\ref{theo:single-message-LB}}
	We use $\bar{\calD}^{U,E}$ and $\bar{\calD}^{V,E}$ to denote the same distributions constructed in the proof of Theorem~\ref{theo:single-message-LB-private-coin}. We moreover use the same notation as in Section~\ref{sec:low-privacy-lowb}.
	
	By a simple application of Markov's inequality and noting that our assumed protocol solves $\DE_{n,D}$ with error $o(D)$ and probability at least $0.99$, it follows that with probability at least $0.9$ over $\alpha \leftarrow \calDpb$, if $|E| \le 0.02 \cdot |D|$, then
	\[
	\| \Hist_{R_\alpha}(\bar{\calD}^{U,E}) - \Hist_{R_\alpha}(\bar{\calD}^{V,E})\|_{TV} = \Omega(1).
	\]
	By Lemma~\ref{lm:single-msg-SDP-to-LDP-pbcoin}, with probability at least $1-\delta_0$ over $\alpha \leftarrow \calDpb$, we have that $R_{\alpha}$ is $\left(\ln\left(n \Big/\frac{c \ln \delta^{-1}}{\ln \delta_0^{-1}}\right),\delta_0\right)$-$\localDP$. We say that such an $\alpha$ is \emph{good}.
	
	For all good $\alpha$ and for a good subset $E$, when the randomizer is set to $R_{\alpha}$ (note that the definition of a good subset depends on the randomizer $R$), by a similar argument as in Theorem~\ref{theo:single-message-LB-private-coin}, we have
	\[
	\| \Hist_{R_\alpha}(\bar{\calD}^{U,E}) - \Hist_{R_\alpha}(\bar{\calD}^{V,E})\|_{TV} = o(1).
	\]
	
	Now, by Lemma~\ref{lm:E-is-good-whp}, if we construct $E$ by including each element of $D$ with probability $0.01$, then for every good $\alpha$, we know that $E$ is good for randomizer $R_{\alpha}$ with probability at least $1-n^{-\omega(1)}$. By a union bound, there exists a fixed choice of $E$ such that $E$ is good for randomizer $R_{\alpha}$ with probability at least $1 - 1/n$ over $\alpha \leftarrow \calDpb$. In the following, we fix $E$ to be such a choice.
	
	Finally, by a union bound, it follows that with probability at least $0.9 - \delta_0 - 1/n > 0$, the above two inequalities hold simultaneously, a contradiction.
\end{proofof}

Theorem~\ref{theo:LDP-LB-strong} follows exactly using a similar argument as in the proof of Theorem~\ref{theo:single-message-LB} (in fact, it is simpler in the local case because there is no need to apply Lemma~\ref{lm:single-msg-SDP-to-LDP-pbcoin}).

\section{$(\eps,\delta)$-Dominated Algorithms}\label{sec:dominated_protocols}

In~\cite{Cheu18}, it was shown that an $(\eps,\delta)$-$\shuffledDP^1$ protocol on $n$ users is also $(\eps + \ln n,\delta)$-$\localDP$, thereby reducing the problem of proving lower bounds for $\shuffledDP^1$ protocols to proving lower bounds on $\localDP$ protocols with low privacy properties.
However, it is known that such a connection does not hold even for $\shuffledDP^2$ protocols~\cite[Section 4.1]{balcer2019separating}. 

Recall the definition of \emph{$(\eps,\delta)$-dominated algorithms} from Definition~\ref{defi:dominated-algo}. In this section, we will show that $\shuffledDP^k$ protocols are dominated.

For clarity of exposition, we will assume that each user sends exactly $k$ messages; this is without loss of generality (see Footnote~\ref{fn:equal-msg}).
%We always use $k$ to denote the number of messages from each user. Let $R: \calX \to \calM^{k}$ be the randomizer and we assume all users use the same randomizer. 
To handle public-coin protocols, we need a relaxed version of dominated algorithms.

\begin{definition}[Dominated Algorithms]
	For a distribution $\mu$ on $\calX$, we say an algorithm $R$ is
	\emph{$(\eps,\delta,\mu)$-dominated}, if for the distribution $\calD_\mu = \Ex_{x \leftarrow \mu} R(x)$, there exists a distribution $\calD$ on $\calM^k$ such that 
	\[
	d_{\eps}\left( \calD_\mu||\calD\right) \le \delta.
	\]
	In this case, we also say $R$ is $(\eps,\delta,\mu)$-dominated by $\calD$.
\end{definition}

\subsection{Approximate-$\shuffledDP$ Protocols are Dominated}

Next we show that approximate $\shuffledDP$ protocols are dominated. For this purpose, we introduce the concept of ``pseudo-locally private'' algorithms, which is a special case of dominated algorithms, and may be interesting in its own right.

\subsubsection{Merged Randomizer}

%A partition of $[k]$ is a collection of disjoint, non-empty subsets of $[k]$ whose union is $[k]$.  Let $\calB([k])$ be the set of all partitions of $[k]$. 

Let $\calB_{n,k}$ be the set of all $k$-sized subsets of the set $[n] \times [k]$. For $\calF \in \calB_{n,k}$ and a randomizer $R$, we define the \emph{merged randomizer} of $R$ with respect to $\calF$, denoted by $R^{\calF}$, as follows:

\begin{framed}
	\begin{center}
		$R^{\calF}(x)$
	\end{center}
	\begin{itemize}
		\item Given an input $x$, for each $j \in [n]$, we simulate $R(x)$ with independent random coins to get an output $w_j \in \calM^k$.
		
		\item Assume that $\calF$ consists of elements $(x_1,y_1),\dotsc,(x_{k},y_{k}) \in [n] \times [k]$ indexed in lexicographical order. We construct a $k$-tuple $z \in \calM^k$ such that $z_{i} = (w_{x_i})_{y_i}$ for each $i \in [k]$.
		
		\item We pre-shuffle $z$ before outputting it. That is, we draw a permutation $\pi\colon [k] \to [k]$ uniformly at random, shuffle $z$ according to $\pi$ to obtain a new $k$-tuple $\widetilde{z}$ (by setting $\widetilde{z}_i = z_{\pi(i)}$ for each $i \in [k]$), and output $\widetilde{z}$.
	\end{itemize}
\end{framed}

That is, $R^{\calF}(x)$ runs $R(x)$ several times, and merges the obtained outputs according to $\calF$. We now define a distribution $\calD_{n,k}$ on $\calB_{n,k}$ as follows: to draw a sample from $\calD_{n,k}$, we simply draw $k$ items $\{(x_i,y_i)\}_{i \in [k]}$ without replacement from the set $[n] \times [k]$. 

Finally, for a randomizer $R$, we define the randomizer $\Rrand$ as follows: Given an input $x$, we first draw $\calF$ from $\calD_{n,k}$, and then simulate $R^{\calF}$ on the input $x$ and output its output. 

\begin{comment}
\paragraph*{Output Probabilities for $R^{\calF}$.}

For $z \in \calM^{k}$ and a subset $S \subset [k]$, we use $z_{S} \in \calM^{|S|}$ to denote the projection of $z$ on $S$, that is, $(z_{S})_i := z_{S_i}$, where $S_i$ is the $i$-th element from $S$.

For $\calF \in \calB([k])$, and $x \in \{0,1\}^D$, we let $p^{\calF}_{x,z} = \prod_{S \in \calF} p_{x,z_{|S}}$, 

For $E \subset \calM^k$, we also write $p^{\calF}_{x,E} = \sum_{z \in E} p^{\calF}_{x,z}$, 

It's not hard to see that $p^{\calF}_{x,z}$ is the probability that $R^{\calF}(x) = z$.
\end{comment}

\subsubsection{Pseudo-Locally Private Algorithms}

We are now ready to define pseudo-locally private algorithms.

\begin{definition}
	We say that an algorithm $R$ is \emph{$(\eps,\delta)$-pseudo-locally private}, if for all $x,y \in \calX$ and all $E \subset \calM^k$,
	\[
	\Pr[R(x) \in E] \le e^{\eps} \cdot \Pr[\Rrand(y) \in E] + \delta.
	\]
\end{definition}

\begin{remark} \label{remark:pseudo-local-implies-dominated}
	If $R$ is $(\eps,\delta)$-pseudo-locally private, then clearly $R$ is also $(\eps,\delta)$-dominated; we can simply take $\calD = \Rrand(y^*)$ for any fixed $y^*$.
\end{remark}

\subsubsection{Multi-Message $\shuffledDP$ Protocols are Pseudo-Locally Private}

Our most crucial observation here is an analogue of \cite[Theorem 6.2]{Cheu18} for multi-message $\shuffledDP$ protocols. Namely, we show that any multi-message $\shuffledDP$ protocol is pseudo-locally private.

\begin{lemma}\label{lm:SDP-to-local-approx}
	If $R$ is $(\eps,\delta)$-$\shuffledDP^k$ on $n$ users, then it is $(\eps + k(1 + \ln n),\delta)$ pseudo-locally private.
\end{lemma}
\begin{proof}
	Suppose otherwise, i.e., that there are $x,y$ and $E \subseteq \calM^k$ such that
	\begin{align*}
	\Pr[R(x) \in E] > (e n)^{k} e^{\eps} \cdot \Ex_{\calF \leftarrow \calD_{n,k}} \left[\Pr\left[R^{\calF}(y) \in E\right]\right] + \delta.
	\end{align*}
	
	Note that since both $R$ and $R^{\calF}$ pre-shuffle their outputs before outputting them, we can assume that $E$ is a union of equivalence classes of $k$-tuples (we say two $k$-tuples $u,v \in \calM^k$ are equivalent if $v$ can be obtained by $u$ via applying a permutation).\footnote{Too see this, we can take $E$ to be $\{ z : R(x)_z > (e n)^{k} e^{\eps} \Rrand(y)_z \wedge z \in \calM^k \}$. One can see that if $u$ and $v$ are equivalent up to a permutation, then either both $u$ and $v$ are in $E$, or neither of them is in $E$.}
	
	Consider two datasets $X_0 = y^n$ and $X_1 = y^{n-1} x$. Let $P$ be the corresponding $\shuffledDP$ protocol with randomizer $R$. For a dataset $X$, we use $P_R(X)$ to denote the random variable of the transcript of $P$ before shuffling. That is, for a dataset $X = (x_i)_{i \in [n]}$, $P_R(X)$ is the concatenation of all $R(X_i)$ for $i$ from $1$ to $n$.
	
	We now define an event $\calE$ as ``there exist $k$ messages in the transcript of $P$ constituting the event $E$''.
	It immediately follows that 
	\begin{align} \label{eq:lower-bound-shuffled}
	\Pr[P_R(X_1) \in \calE] \ge \Pr[R(x) \in E].
	\end{align}
	Furthermore, we claim that
	\begin{align} \label{eq:upper-bound-shuffled}
	\Pr[P_R(X_0) \in \calE] \le \binom{kn}{k}\cdot \Ex_{\calF \leftarrow \calD_{n,k}} \left[\Pr\left[R^{\calF}(y) \in E\right]\right].
	\end{align} 
	To see why the above inequality holds, note that if we pick $k$ messages from $P_R(X_0)$, depending on which users these messages come from, the probability that they constitute $E$ is bounded by $\Pr[R^{\calF}(y) \in E]$ for a certain $\calF \in \calB_{n,k}$. 
	
	Moreover, if we pick these $k$ messages uniformly at random from all $\binom{kn}{k}$ possible $k$-tuples, the corresponding $\calF$ is distributed according to $\calD_{n,k}$. Therefore, we can apply a union bound over all $\binom{kn}{k}$ possible $k$-tuples and sum up the corresponding $\Pr[R^{\calF}(y) \in E]$ to obtain an upper bound on $\Pr[P_R(X_0) \in \calE]$. The aforementioned sum is precisely $\binom{kn}{k}$ times the expectation $\Ex_{\calF \leftarrow \calD_{n,k}} \left[\Pr\left[R^{\calF}(y) \in E\right]\right]$.

	Since $\binom{kn}{k} \le (e n)^k$, we may combine~\eqref{eq:lower-bound-shuffled} and~\eqref{eq:upper-bound-shuffled} to obtain
	\begin{equation}
	\Pr[P_R(X_1) \in \calE] >  e^{\eps} \cdot \Pr[P_R(X_0) \in \calE] + \delta. \label{eq:contradiction}
	\end{equation}
	
	Finally, note that applying a random permutation to the transcript does not change whether the event $\calE$ occurs. Therefore,~\eqref{eq:contradiction} contradicts the assumption that $R$ is $(\eps,\delta)$-$\shuffledDP$.
\end{proof}
From Remark~\ref{remark:pseudo-local-implies-dominated}, we get the following corollary:

\begin{cor}\label{cor:SDP-to-dominated}
	If $R$ is $(\eps,\delta)$-$\shuffledDP^k$ on $n$ users, then it is $(\eps + k (1 + \ln n),\delta)$-dominated.
\end{cor}

Next, we generalize Corollary~\ref{cor:SDP-to-dominated} to the public-coin case:

\begin{lemma}\label{lm:SDP-to-dominated-public-coin}
	If $R$ is $(\eps,\delta)$-$\shuffledDP^k$ in the $n$-user public-coin setting with public randomness from $\calDpb$ and $\tau = \eps + k (1 + \ln n)$, then there is a family of distributions $\{ \calD_{\alpha} \}_{\alpha \in \supp(\calDpb)}$ over $\calM^k$ such that for every distribution $\mu$ over $\calX$,
	\[
	\Ex_{\alpha \leftarrow \calDpb} d_{\tau}\left(\Ex_{x \leftarrow \mu} R_{\alpha}(x)||\calD_{\alpha}\right)\le\delta.
	\]
	
	In other words, there are reals $\{ \delta_\alpha \}_{\alpha \in \supp(\calDpb)}$ such that $\Ex_{\alpha \leftarrow \calDpb}[\delta_\alpha] \le \delta$ and $R_\alpha$ is $(\tau,\delta_\alpha,\mu)$-dominated by $\calD_\alpha$.
\end{lemma}
\begin{proof}
	From the proof of Lemma~\ref{lm:SDP-to-local-approx}, it follows that for all $x,y \in \calX$ and $\alpha \in \supp(\calDpb)$, we have that
	\[
	d_{\tau}(R_{\alpha}(x)||\Rrand_\alpha(y)) \le d_{\eps}(P_{R_\alpha}(xy^{n-1})||P_{R_\alpha}(y^{n})).
	\]
	Since $R$ is $(\eps,\delta)$-$\shuffledDP$ on $n$ users, it follows that
	\[
	\Ex_{\alpha \leftarrow \calDpb}[d_{\eps}(P_{R_\alpha}(xy^{n-1})||P_{R_\alpha}(y^{n}))] \le \delta.
	\]
	Putting the above two inequalities together, for all $x,y \in \calX$, we get that
	\[
	\Ex_{\alpha \leftarrow \calDpb} [d_{\tau}(R_{\alpha}(x)||\Rrand_\alpha(y))] \le \delta.
	\]
	We now finish the proof by fixing $y^* \in \calX$, and setting $\calD_{\alpha} = \Rrand_\alpha(y^*)$ for every $\alpha \in \supp(\calDpb)$.
\end{proof}

	\subsection{Bounding KL Divergence for Dominated Randomizers}

%\badih{Might be better to make this a subsection of Section~\ref{sec:dominated_protocols}, instead of having it as a separate section?}

%\lnote{agreed.}

In this subsection, we prove the technical lemma bounding average-case KL divergences for dominated randomizers.

As before, we use $\calX$ and $\calM$ to denote the input space and the message space respectively. For a local randomizer $R \colon \calX \to \calM$, we let $p_{x,z} = \Pr[R(x) = z]$.
%We also let $f_z(x) = \frac{p_{x,z}}{q_z}$.

Let $\mu$ be a distribution on $\calX$. 
Let $\calI$ be an index set, $\pi$ be a distribution on $\calI$, and $\{ \lambda_{v} \}_{v \in \calI}$ be a family of distributions on $\calX$.  For a constant $\tau$, we say that $\mu$ \emph{$\tau$-dominates} $\{ \lambda_v \}$ if for all $x \in \calX$ and $v \in \calI$, it holds that $(\lambda_v)_x \le \tau \cdot \mu_x$.

\begin{theorem}\label{theo:many-vs-one}
For a constant $\tau \ge 2$, let $\mu$ be a distribution which $\tau$-dominates a distribution family $\{ \lambda_v \}_{v \in \calI}$. Let $\pi$ be a distribution on $\calI$. Let $W \colon \mathbb{R} \to \mathbb{R}$ be a concave function such that for all functions $\psi \colon \calX \to \R^{\ge 0}$ satisfying $\psi(\mu) \le 1$, it holds that 
		\[
		\Ex_{v \leftarrow \pi} \left[(\psi(\lambda_v) - \psi(\mu))^2\right] \le W(\|\psi\|_{\infty}).
		\] 	
Then for an $(\eps,\delta,\mu)$-dominated randomizer $R$, it holds that
\[
	\Ex_{v \leftarrow \pi} [\KL(R(\lambda_v)||R(\mu))] \le 2 W(2e^{\eps}) + 4 (\tau - 1)^2 \cdot \delta.
\]
\end{theorem}
\begin{proof}
	Let $Q = R(\mu)$. Recall that $p_{x,z} = \Pr[R(x) = z]$. We also set $q_z = \Pr[Q = z]$ and $f_z(x) = \frac{p_{x,z}}{q_z}$.
	
	It follows from the assumption that there exists a distribution $\qD$ that $(\eps,\delta,\mu)$-dominates $R$. Noting that $\chi^2$-divergence upper-bounds KL divergence (see Section~\ref{sec:divergences}), it follows that
	\begin{align*}
	\Ex_{v \leftarrow \pi} \KL(R(\lambda_v) || Q) &\le \Ex_{v \leftarrow \pi} \chi^2(R(\lambda_v) || Q) \\ 
	&\le \Ex_{v \leftarrow \pi} \Ex_{z \leftarrow Q} \left[ \frac{R(\lambda_v)_z - q_z}{q_z} \right]^2\\
	&= \Ex_{z \leftarrow Q} \Ex_{v \leftarrow \pi} \left[ f_z(\lambda_v) - 1 \right]^2.
	\end{align*}
	We will further decompose $f_z = g_z + h_z$ so that $\|g_z\|_{\infty}$ is small and $\Ex_{z\leftarrow Q} h_z(\mu)$ is small. Formally, for each $z \in \calM$, we define a truncation level 
	\[
	L_z = \frac{2 e^{\eps} \cdot q_z^{\calD}}{q_z}.
	\] 
	Then, we define $g_z$ and $h_z$ as follows
	\[
	g_{z}(x) := \begin{cases}
	f_{z}(x) & \text{if } f_{z}(x) \le L_z,\\
	0         & \text{otherwise},
	\end{cases}
	\quad\quad\text{and}\quad\quad
	h_z(x) := f_z(x) - g_z(x).
	\]
	Fix a $z$ in the support of $Q$. Noting that 
	$$
	g_z(\mu) + h_z(\mu) = f_z(\mu) = \frac{\Ex_{x \leftarrow \mu}[p_{x,z}]}{q_z} = 1,
	$$
	we get
	\begin{align}
	\Ex_{v \leftarrow \pi} \left[ f_z(\lambda_v) - 1 \right]^2 &= 
	\Ex_{v \leftarrow \pi} \left[ (g_z(\lambda_v) - g_z(\mu)) + (h_z(\lambda_v) - h_z(\mu)) \right]^2 \notag \\
	&\le 2 \cdot \Ex_{v \leftarrow \pi} [(g_z(\lambda_v) - g_z(\mu))^2] +  2 \cdot \Ex_{v \leftarrow \pi} [(h_z(\lambda_v) - h_z(\mu))^2]. \label{eq:KL-two-sqaures}
	\end{align}
	To simplify the notation, in the following we set $\hat{g}_z(\lambda_v) := g_z(\lambda_v) - g_z(\mu)$ and 
	$\hat{h}_z(\lambda_v) = h_z(\lambda_v) - h_z(\mu)$. We will bound the two terms in~\eqref{eq:KL-two-sqaures} separately. 
	
	\paragraph{Bounding $\Ex_{v \leftarrow \pi} \hat{g}_z(\lambda_v)^2$.} 
	
	Since $W$ is concave, noting that $\|g_z\|_{\infty} \le L_z$ and $g_z(\mu) \le 1$, it follows that
	\[
	\Ex_{z \leftarrow Q}\Ex_{v\leftarrow \pi} \hat{g}_z(\lambda_v)^2 \le \Ex_{z \leftarrow Q} W(L_z) \le W\left(\Ex_{z\leftarrow Q} L_z\right),
	\]
	where the second step uses Jensen's inequality. From the definition of $L_z$, we have that
	\[
	\Ex_{z\leftarrow Q} L_z = \sum_{z} q_z \cdot \frac{2e^{\eps} \cdot \qD_z}{q_z} = 2 e^{\eps} \cdot \sum_{z} \qD_z = 2 e^{\eps},
	\]
	where the last equality follows from the fact that $\qD$ is a distribution. We therefore obtain
	\[
	\Ex_{z \leftarrow Q}\Ex_{v\leftarrow \pi} \hat{g}_z(\lambda_v)^2 \le W(2 e^{\eps}).
	\]
	%\paragraph{Bounding $\Ex_{z \leftarrow Q}\Ex_{v \leftarrow \pi} \hat{h}_z(\lambda_v)^2$.}

	\paragraph{Bounding $\Ex_{v \leftarrow \pi} \hat{h}_z(\lambda_v)^2$.}  Since $\mu$ $\tau$-dominates $\{ \lambda_v \}$, it follows that
	$$
	|\hat{h}_z(\lambda_v)| = |h_z(\lambda_v) - h_z(\mu)| \le \max\left\{h_z(\mu), \tau \cdot h_z(\mu) - h_z(\mu)\right\} \le (\tau - 1) h_z(\mu).
	$$
	Therefore,
	\[
	\Ex_{z \leftarrow Q}\Ex_{v \leftarrow \pi} \hat{h}_z(\lambda_v)^2 \le (\tau - 1)^2 \cdot \Ex_{z \leftarrow Q} h_z(\mu)^2 \le (\tau - 1)^2 \cdot \Ex_{z \leftarrow Q} h_z(\mu),
	\]
	where the last inequality holds because $h_z(\mu) \le 1$.
	
	By the definition of $h_z$, it follows that
	\[
	\Ex_{z \leftarrow Q} h_z(\mu) = \sum_{z \in \calM} \Ex_{x \leftarrow \mu}\Big[p_{x,z} \cdot \mathbb{1}[f_z(x) > L_z] \Big].
	\]
	Let $\mathcal{T}_x = \{z \in \calM : f_z(x) > L_z\}$. For $z \in \mathcal{T}_x$, we get that 
	\begin{align*}
	p_{x,z} &> L_z \cdot q_z \\
	&> \frac{2 e^{\eps} \cdot q_z^{\calD}}{q_z} \cdot q_z = 2 \cdot e^{\eps} \cdot \qD_z.
	\end{align*}
	In particular, the above means that $\Ex_{x \leftarrow \mu} p_{x,\mathcal{T}_x} \le 2 \delta$, as otherwise 
	\[
	\Ex_{x \leftarrow \mu}[p_{x,\mathcal{T}_x}] -  e^{\eps} \cdot \qD_{\mathcal{T}_x} \ge \Ex_{x \leftarrow \mu}[p_{x,\mathcal{T}_x}]/2 > \delta,
	\]
	contradicting the fact that $R$ is $(\eps,\delta,\mu)$-dominated by $\qD$. Hence, we have
	\[
	\Ex_{z \leftarrow Q} h_z(\mu) = \sum_{z \in \calM} \Ex_{x \leftarrow \mu}\Big[p_{x,z} \cdot \mathbb{1}[f_z(x) > L_z] \Big] = \Ex_{x \leftarrow \mu} p_{x,\mathcal{T}_x} \le 2 \delta.
	\]
	Putting everything together, it follows that
	\[
	\Ex_{z \leftarrow Q}\Ex_{v \leftarrow \pi} \hat{h}_z(\lambda_v)^2 \le 2 (\tau - 1)^2 \cdot \delta.
	\]
	
	\paragraph*{Final Bound.} Combining our bounds on $\Ex_{z \leftarrow Q}\Ex_{v \leftarrow \pi} \hat{g}_z(\alpha)^2$ and $\Ex_{z \leftarrow Q}\Ex_{v \leftarrow \pi} \hat{h}_z(\lambda_v)^2$, we conclude that
	\[
	\Ex_{v \leftarrow \pi} [\KL(R(\lambda_v)||R(\mu))] \le 2 W(2e^{\eps}) + 4 (\tau - 1)^2 \cdot \delta.
	\qedhere
	\]
\end{proof}

\section{Lower Bounds for \selection and \paritylearning}\label{sec:selection_parity_learning}

In this section, we prove lower bounds for \selection and \paritylearning in the $\shuffledDP$ model. We begin with some notation.
\subsection{Notation}
For $(\ell,s) \in [2] \times \{0,1\}^D$, let $\calD_{\ell,s}$ be the uniform distribution on $\{x \in \{0,1\}^D : \langle x,s\rangle = \ell \}$. Recall that $\calU_D$ is the uniform distribution on $\{0,1\}^D$.

For $j \in [D]$, let $e_j$ be the $D$-bit string such that only the $j$-th bit is $1$, and the other bits are $0$. For $(\ell,j) \in [2] \times [D]$, we denote by $\calD_{\ell,e_j}$ the uniform distribution on all length-$D$ Boolean strings with $j$-th bit being $\ell$. For simplicity, we also use $\calD_{\ell,j}$ to denote $\calD_{\ell,e_j}$ when the context is clear.

We need the following simple proposition.

\begin{prop}\label{prop:fourier-fact}
	For every function $f\colon \{0,1\}^D \to \R$ and $s \in \{0,1\}^D$,
	\[
	\hat{f}(s) = \frac{1}{2} (f(\calD_{0,s}) - f(\calD_{1,s})).
	\]
\end{prop}
\begin{proof}
	By definition, we have that
	\[
	\hat{f}(s) = \Ex_{x \leftarrow \calU_{D}} (-1)^{\langle s,x \rangle} f(x)  = \frac{1}{2} (f(\calD_{0,s}) - f(\calD_{1,s})).\qedhere
	\]
\end{proof}
\subsection{Lower Bound for \selection}\label{sec:lowb-selection}

We begin with lower bounds for \selection.

\begin{lemma}\label{lm:KL-bound-sel}
	For $\eps > 0$, suppose $R$ is $(\eps, \delta, \calU_{D})$-dominated, then we have
	\[
	\Ex_{(\ell,j) \in [2]\times[D]} [\KL(R(\calD_{\ell,j}) || R(\calU_D))] \le O\left(\frac{\eps}{D} + \delta\right).
	\]
\end{lemma}
\begin{proof}
	To apply Theorem~\ref{theo:many-vs-one}, we set the index set as $\calI = [2] \times [D]$, the distribution $\pi$ to the uniform distribution over $\calI$, $\{ \lambda_v \}_{v \in \calI} = \{ \calD_v \}_{v \in \calI}$, and $\mu = \calU_{D}$.
	
	Clearly, $\mu$ $2$-dominates $\{\lambda_v\}$. Let $f$ be a function such that $\|f\|_{\infty} = L$ and $f(\mu) \le 1$, it follows that
	\begin{align*}
	\Ex_{v \leftarrow \pi} (f(\mu) - f(\lambda_v))^2 &=\Ex_{(\ell,j) \in [2] \times [D]} (f(\calD_{\ell,j}) - f(\calU_D))^2\\
	&=\Ex_{(\ell,j) \in [2] \times [D]} \frac{1}{4}(f(\calD_{\ell,j}) - f(\calD_{1-\ell,j}))^2\\
	&=\Ex_{(\ell,j) \in [2] \times [D]} \hat{f}(\{j\})^2. \tag{\text{Proposition~\ref{prop:fourier-fact}}}
	\end{align*}
	By Lemma~\ref{lm:level-1}, it follows that
	\[
	\Ex_{v \leftarrow \pi} (f(\mu) - f(\lambda_v))^2 
	= \Ex_{(\ell,j) \in [2] \times [D]} \hat{f}(\{j\})^2
	\leq O\left( \frac{\ln L}{D} \right).
	\]
	Therefore, we can set $W(L) := \frac{C \cdot \ln L}{D}$ for a large enough constant $C$ and note that $W$ is a concave function. By Theorem~\ref{theo:many-vs-one}, it follows that
	\[
	\Ex_{(\ell,j) \in [2]\times[D]} [\KL(R(\calD_{\ell,j}) || R(\calU_D))] \le O( W(2e^{\eps}) + \delta ) \le O\left(\frac{\eps}{D} + \delta\right).
	\qedhere
	\]
\end{proof}

\begin{lemma}\label{lm:dominated-to-lowerbound}
	For a public-coin randomizer $R$ with public randomness from $\calDpb$, if there is a family of distributions $\{ \calD_{\alpha} \}_{\alpha \in \supp(\calDpb)}$ over $\calM^k$ such that
	\[
	\Ex_{\alpha \leftarrow \calDpb} d_{\eps}\left(\Ex_{x \leftarrow \calU_{D}} R_{\alpha}(x)||\calD_{\alpha}\right) \le o(1/D),
	\]
	then a public-coin protocol with randomizer $R$ needs at least $\Omega\left(\frac{D\log D}{\eps}\right)$ samples to solve \selection with probability at least $0.99$. 
\end{lemma}
\begin{proof}
	Let $L,J$ be uniformly random over $[2] \times [D]$, and $X_1,X_2,\dotsc,X_n$ be $n$ i.i.d. samples from $D_{L,J}$. For each $i \in [n]$, we draw $Z_i$ from $R(X_i)$. 
	
	Let $P_{\alpha}(Z_1,Z_2,\dotsc,Z_n)$ be the output of the protocol with public randomness fixed to $\alpha$, and let $F_{\alpha}(Z_1,Z_2,\dotsc,Z_n) := (1,P_{\alpha}(Z_1,Z_2,\dotsc,Z_n))$. Assuming $n \ge \Theta(\log D)$, it follows that $F_{\alpha}(Z_1,Z_2,\dotsc,Z_n) = (L,J)$ with probability at least $0.99 - 0.01 \ge 0.98$ over the randomness of $\alpha \leftarrow \calDpb$ and randomness in $F_{\alpha}$, conditioned on the event $L = 1$.
	
	By Markov's inequality, with probability at least $0.8$ over $\alpha \leftarrow \calDpb$, we have that $F_{\alpha}(Z_1,Z_2,\dotsc,Z_n) = (L,J)$ with probability at least $0.8$ over the randomness in $F_\alpha$ conditioned on the event $L = 1$.
	
	From our assumption and Markov's inequality, it follows that with probability at least $0.99$ over $\alpha \leftarrow \calDpb$, we have $d_{\eps}\left(\Ex_{x \leftarrow \calU_{D}} R_{\alpha}(x)||\calD_{\alpha}\right) \le o(1/D)$. That is, $R_{\alpha}$ is $(\eps,o(1/D),\calU_{D})$-dominated.
	
	By a union bound, with probability at least $0.99 - 0.2 > 0$ over $\alpha \leftarrow \calDpb$, we have $F_{\alpha}(Z_1,Z_2,\dotsc,Z_n) = (L,J)$ with probability at least $0.8/2 \ge 0.4$, and $R_{\alpha}$ is $(\eps,o(1/D),\calU_{D})$-dominated. In the following, we fix such an $\alpha$.
	
	By Lemma~\ref{lm:KL-bound-sel}, for $\beta \le O\left(\frac{\eps}{D}\right)$, we have that
	\[
	\Ex_{(\ell,j) \in [2]\times[D]} [\KL(R_{\alpha}(\calD_{\ell,j}) || R_{\alpha}(\calU_D))] \le \beta.
	\]
	By Fano's inequality,
	\begin{align*}
	\Pr[F_\alpha(Z_1,Z_2,\dotsc,Z_n) = (L,J)] &\le \frac{1+I((Z_1,Z_2,\dotsc,Z_n);(L,J))}{\log 2D}\\
	&\le\frac{1+n \cdot I(Z_1;(L,J))}{\log 2D}.
	\end{align*}
	We also have that
	\begin{align*}
	I(Z_1;(L,J)) = \KL((Z_1,L,J) || Z_1 \otimes (L,J)) 
	&= \Ex_{L,J \leftarrow [2] \times [D]} \KL((Z_1 | L,J) || Z_1)\\ 
	&= \Ex_{L,J \leftarrow [2] \times [D]} \KL(R(\calD_{L,J}) || R(\calU_D))\\
	&\le \beta.
	\end{align*}
	Plugging in, we obtain
	\[
	\frac{1+ n \cdot \beta}{\log 2D} \ge \Pr[F(Z_1,Z_2,\dotsc,Z_n) = (L,J)] \ge 0.4.
	\]
	Hence, we deduce that $n = \Omega(\log D \cdot \beta^{-1}) = \Omega\left(\frac{D\log D}{\eps}\right)$.
\end{proof}

We are now ready to prove Theorem~\ref{theo:lowb-selection} (restated below).

\begin{reminder}{Theorem~\ref{theo:lowb-selection}.}
	For any $\eps = O(1)$, if $P$ is a public-coin $(\eps,o(1/D))$-$\shuffledDP^k$ protocol solving \selection with probability at least $0.99$, then $n \ge \Omega\left(\frac{D}{k}\right)$.
\end{reminder}
\begin{proof}
	Without loss of generality, we assume that $n \le \poly(D)$. Applying Lemma~\ref{lm:SDP-to-dominated-public-coin} and letting $\tau = \eps + k (1 + \ln n)$, we get that
	\[
	\Ex_{\alpha \leftarrow \calDpb} d_{\tau}\left( \Ex_{x \leftarrow \calU_D} R_{\alpha}(x) || \calD_{\alpha} \right) \le \delta,
	\]
	for a distribution family $\{\calD_{\alpha}\}_{\alpha \in \supp(\calDpb)}$.
	
	Therefore, by Lemma~\ref{lm:dominated-to-lowerbound}, it follows that $n \ge \Omega\left(\frac{D\log D}{\tau}\right) = \Omega\left(\frac{D}{k}\right)$.
\end{proof}

\subsection{Lower Bound for \paritylearning}\label{sec:lowb-learning-parity}

We next prove our lower bound for \paritylearning.

%We also let $\calD_{\ell,s}^{\alpha} = (\frac{1}{2} + \frac{\alpha}{2}) \cdot \calD_{\ell,s} + (\frac{1}{2} - \frac{\alpha}{2}) \cdot \calD_{1-\ell,s}$, and $P_{\ell,s}^\alpha$ be the distribution $R(\calD_{\ell,s}^\alpha)$.  

\begin{lemma}\label{lm:KL-bound-par}
	For $\eps > 0$, suppose $R$ is $(\eps, \delta, \calU_{D})$-dominated. We have that
	\[
	\Ex_{\ell,s \in [2]\times \{0,1\}^D} [\KL(R(\calD_{\ell,s}) || R(\calU_D))] \le  \frac{4 e^{\eps}}{2^D} + 4\delta.
	\]
\end{lemma}
\begin{proof}
	To apply Theorem~\ref{theo:many-vs-one}, we set the index set as $\calI = [2] \times \{0,1\}^D$, distribution $\pi$ to be the uniform distribution over $\calI$, $\{ \lambda_v \}_{v \in \calI} = \{ \calD_v \}_{v \in \calI}$, and $\mu = \calU_{D}$.
	
	Clearly, $\mu$ $2$-dominates $\{\lambda_v\}$. Let $f$ be a function such that $\|f\|_{\infty} = L$ and $f(\mu) \le 1$.  It follows that
	\begin{align*}
	\Ex_{v \leftarrow \pi} |f(\mu) - f(\lambda_v)|^2 &=\Ex_{\ell,s \in [2] \times \{0,1\}^D} |f(\calD_{\ell,s}) - f(\calU_D)|^2\\
	&=\Ex_{\ell,s \in [2] \times \{0,1\}^D} \frac{1}{4}|f(\calD_{\ell,s}) - f(\calD_{1-\ell,s})|^2\\
	&=\Ex_{\ell,s \in [2] \times \{0,1\}^D} \hat{f}(s)^2. \tag{\text{Proposition~\ref{prop:fourier-fact}}}
	\end{align*}
	By Lemma~\ref{lm:parseval}, it follows that
	\[
	\sum_{s \in \{0,1\}^{D}} \hat{f}(s)^2 = \Ex_{x \leftarrow \calU_{D}} f(x)^2 \le \|f\|_{\infty} \cdot f(\calU_{D}) \le L.
	\]
Therefore, we can set $W(L) := \frac{L}{2^D}$. In this case, $W$ is clearly concave. By Theorem~\ref{theo:many-vs-one}, it follows that
	\[
	\Ex_{\ell,s \in [2]\times \{0,1\}^D} [\KL(R(\calD_{\ell,s}) || R(\calU_D))] \le 2 W(2e^{\eps}) + 4 (\tau - 1)^2 \cdot \delta \le \frac{4 e^{\eps}}{2^D} + 4\delta.
	\qedhere
	\] 
\end{proof}

Now we apply Lemma~\ref{lm:KL-bound-par} to the \paritylearning problem. Recall that in \paritylearning, there is a random hidden element $s \in \{0,1\}^D$, and each user gets a random element $x$ together with the inner product $\langle s,x \rangle$ over $\F_2$. Appending the label to the vector, each user indeed gets a random sample from the set $\{ x \in \{0,1\}^{D+1} : \langle x,(s,1) \rangle = 0 \}$, where $(s,1)$ is the $(D+1)$-dimensional vector obtained by appending $1$ to the end of the vector $s$. In other words, each user gets a random sample from the distribution $\calD_{0,(s,1)}$.
 
\begin{lemma}\label{lm:dominated-to-lowerbound-par}
	For a public-coin randomizer $R$ with public randomness from $\calDpb$, if there is a family of distributions $\{ \calD_{\alpha} \}_{\alpha \in \supp(\calDpb)}$ over $\calM^k$ such that,
	\[
	\Ex_{\alpha \leftarrow \calDpb} d_{\eps}\left(\Ex_{x \leftarrow \calU_{D}} R_{\alpha}(x)||\calD_{\alpha}\right) \le o(1/n),
	\]
	where $n$ is the number of samples, then a public-coin protocol with randomizer $R$ needs at least $\Omega\left( 2^{D} / e^{\eps} \right)$ samples to solve \paritylearning with probability at least $0.99$. 
\end{lemma}
\begin{proof}
	Suppose there is a public-coin protocol $P$ with randomizer $R$ solving \paritylearning with probability at least $0.99$. For a dataset $W$, we use $P(W)$ (respectively, $P_{\alpha}(W)$) to denote the output of $P$ on $W$ (with public randomness fixed to $\alpha$).
	
	Consider running $P$ on $n$ uniformly random samples from $\{0,1\}^{D+1}$. We note that for at least a $0.99$ fraction of $s \in \{0,1\}^{D}$, we have that
	\[
	\Ex_{\alpha \leftarrow \calDpb} \left[ \Pr[P_{\alpha}(\calU_{D+1}^{\otimes n}) = s] \right] = \Pr[P(\calU_{D+1}^{\otimes n}) = s] \le 0.01. 
	\]
	From the assumption that $P$ solves \paritylearning, for all $s \in \{0,1\}^{D}$, we have
	\[
	\Ex_{\alpha \leftarrow \calDpb} \left[ \Pr[P_{\alpha}(\calD_{0,(s,1)}^{\otimes n}) = s] \right] = \Pr[P(\calD_{0,(s,1)}^{\otimes n}) = s] \ge 0.99.
	\]
	By a union bound, with probability at least $0.5$ over $\alpha \leftarrow \calDpb$, we have
	\begin{equation}
	\Pr[P_{\alpha}(\calU_{D+1}^{\otimes n}) = s] \le 0.1~~\text{and}~~ \Pr[P_\alpha(\calD_{0,(s,1)}^{\otimes n}) = s] \ge 0.9~~\text{for at least a $0.5$ fraction of $s \in \{0,1\}^D$.}
	\label{eq:condition-alpha}
	\end{equation}
	From our assumption and Markov's inequality, with probability at least $0.99$ over $\alpha \leftarrow \calDpb$, we have that $d_{\eps}\left(\Ex_{x \leftarrow \calU_{D}} R_{\alpha}(x)||\calD_{\alpha}\right) \le o(1/n)$. That is, $R_{\alpha}$ is $(\eps,o(1/n),\calU_{D})$-dominated. 
	
	By a union bound, there exists an $\alpha \in \supp(\calDpb)$ such that $R_{\alpha}$ is $(\eps,o(1/D),\calU_{D})$-dominated and~\eqref{eq:condition-alpha} is satisfied.
	
	By Lemma~\ref{lm:KL-bound-par}, we have that
	\[
	\Ex_{\ell,s \in [2]\times \{0,1\}^{D+1}} [\KL(R_\alpha(\calD_{\ell,s}) || R_\alpha(\calU_D))] \le  \frac{4 e^{\eps}}{2^{D+1}} + o(1/n),
	\]
	which implies
	\[
	\Ex_{s \in \{0,1\}^{D}} [\KL(R_\alpha(\calD_{0,(s,1)}) || R_\alpha(\calU_D))] \le  O\left(\frac{e^{\eps}}{2^D}\right) + o(1/n).
	\]
	Supposing $n = o(2^D/e^{\eps})$ for the sake of contradiction, it follows that
	\[
	\Ex_{s \in \{0,1\}^{D}} [\KL(R_\alpha(\calD_{0,(s,1)})^{\otimes n} || R_\alpha(\calU_{D+1})^{\otimes n})] \le o(1).
	\]
	
	Since there is at least a $0.5$ fraction of $s \in \{0,1\}^D$ satisfying the conditions in~\eqref{eq:condition-alpha}, it follows that there exists an $s \in \{0,1\}^D$ satisfying these conditions and $\KL(R_\alpha(\calD_{0,(s,1)})^{\otimes n} || R_\alpha(\calU_{D+1})^{\otimes n})] = o(1)$,
	which, by Pinsker's inequality, implies that
	\[
	\| R_\alpha(\calD_{0,(s,1)})^{\otimes n} - R_\alpha(\calU_{D+1})^{\otimes n}\|_{TV} \le o(1),
	\]
	and
	\[
	\Pr[P_\alpha(\calD_{0,(s,1)}^{\otimes n}) = s] \le \Pr[P_{\alpha}(\calU_{D+1}^{\otimes n}) = s] + o(1) \le 0.01 + o(1),
	\]
	a contradiction.
\end{proof}

We are now ready to prove Theorem~\ref{theo:lowb-learning-parity}.

\begin{reminder}{Theorem~\ref{theo:lowb-learning-parity}.}
	For any $\eps = O(1)$, if $P$ is a public-coin $(\eps,o(1/n))$-$\shuffledDP^k$ protocol solving \paritylearning with probability at least $0.99$, then $n \ge \Omega(2^{D/(k+1)})$.
\end{reminder}
\begin{proofof}{Theorem~\ref{theo:lowb-learning-parity}}
	
	Applying Lemma~\ref{lm:SDP-to-dominated-public-coin} and letting $\tau = \eps + k (1 + \ln n)$, we have that
	\[
	\Ex_{\alpha \leftarrow \calDpb}  d_{\tau}\left( \Ex_{x \leftarrow \calU_D} R_{\alpha}(x) || \calD_{\alpha} \right) \le o(1/n),
	\]
	for a distribution family $\{\calD_{\alpha}\}_{\alpha \in \supp(\calDpb)}$.
	
	By Lemma~\ref{lm:dominated-to-lowerbound-par}, $n \ge \Omega(2^D/e^{\tau}) \ge \Omega(2^D / (e n)^k)$. It then follows that $n^{k+1} \ge \Omega(2^D / e^{k})$ and consequently $n \ge \Omega(2^{D/(k+1)})$.
\end{proofof}

%\subsection{Lower Bounds for Learning Sparse Parity}

	%\input{query-release-shuffle}
	
\section{Lower Bound for \DE with Maximum Hardness }\label{sec:lb_count_distinct_local}

In this section, we prove Theorem~\ref{theo:lowb-distinct-elements}, which gives a $\Omega(n)$ lower bound on the error of $\localDP$ protocols for $\DE$.

%We start with some preliminaries.

\subsection{Preliminaries}

For $(\ell,s) \in [2] \times \{0,1\}^D$, recall that $\calD_{\ell,s}$ is the uniform distribution on $\{x \in \{0,1\}^D : \langle x,s\rangle = \ell \}$. As in~Section~\ref{sec:overview-dominated},
we also use $\calD_{\ell,s}^{\alpha}$ to denote the mixture of $\calD_{\ell,s}$ and $\calU_{D}$ which outputs a sample from $\calD_{\ell,s}$ with probability $\alpha$ and a sample from $\calU_{D}$ with probability $1-\alpha$. Note that $\calD_{\ell,s}^{\alpha}$ can also be interpreted as the mixture of $\calD_{\ell,s}$ and $\calD_{1-\ell,s}$ that outputs a sample from $\calD_{\ell,s}$ with probability $\frac{1}{2} + \frac{\alpha}{2}$, and a sample from $\calD_{1-\ell,s}$ with probability $\frac{1}{2} - \frac{\alpha}{2}$.
\begin{comment}
We also let 
$$
\calD_{\ell,s}^{\alpha} = \left(\frac{1}{2} + \frac{\alpha}{2} \right) \cdot \calD_{\ell,s} + \left(\frac{1}{2} - \frac{\alpha}{2} \right) \cdot \calD_{1-\ell,s},
$$
\end{comment}
We next estimate the number of distinct elements in $n$ samples taken from $\calD^{\alpha}_{\ell,s}$.

\begin{prop}\label{prop:DE-bound}
	Set $D = \log n$. For $\alpha \in (0,0.01)$ and any $(\ell,s) \in [2] \times \{0,1\}^D$, let $X$ be the number of distinct elements in $n$ samples drawn from $\calD_{\ell,s}^{\alpha}$. We have that
	\[
	\Pr\left[ \left|X - (1 - e^{-1}\cosh(\alpha)) \cdot n \right| > 10\sqrt{n}\right] < 0.01.
	\]
	
\end{prop}
\begin{proof}
	In the following, we identify the index space $[n]$ with $\{0,1\}^{\log n}$ in the natural way. For $i \in [n]$, we use $X_i$ to denote the indicator of whether $i$ occurs in the $n$ samples taken from $\calD^{\alpha}_{\ell,s}$. Note that these $X_i$'s are not independent, but they are negatively correlated~\cite[Proposition~7 and 11]{DubhashiR98}, and hence a Chernoff bound still applies.
	
	Let $i$ be an element in the support of $\calD_{\ell,s}$.  Note that $i$ equals one sample from $\calD_{\ell,s}^\alpha$ with probability
	\[
	2^{-D+1} \cdot \left(\frac{1}{2} + \frac{\alpha}{2} \right) = \frac{1+\alpha}{n}.
	\]
	Therefore, $i$ occurs in $n$ i.i.d. samples from $\calD_{\ell,s}^\alpha$ with probability
	\begin{align*}
	p_1 := 1 - (1 - (1+\alpha) / n)^{n} &= 1 - e^{ \ln(1-(1+\alpha) / n) \cdot n} \\
								 &= 1 - e^{ (-(1+\alpha)/n + \Theta((1+\alpha)/n)^2) \cdot n }\\
								 &= 1 - e^{-(1+\alpha)} \cdot e^{\Theta(1/n)}.
	\end{align*}
	Therefore, we have that
	\[
	\left| p_1 - (1 - e^{-(1+\alpha)}) \right| \le O(1/n).
	\]
	Similarly, for an element $i$ in the support of $\calD_{\ell,s}$, $i$ equals one sample from $\calD_{\ell,s}^\alpha$ with probability
	\[
	2^{-D+1} \cdot \left( \frac{1}{2} - \frac{\alpha}{2} \right) = \frac{1-\alpha}{n}.
	\]
	Hence, by a similar calculation, $i$ occurs in $n$ i.i.d. samples from $\calD_{\ell,s}^\alpha$ with probability
	\begin{align*}
		p_2 := 1 - (1 - (1-\alpha) / n)^{n} = 1 - e^{-(1-\alpha)} \cdot e^{\Theta(1/n)},
	\end{align*}
	and
	\[
	\left| p_2 - (1 - e^{-(1-\alpha)}) \right| \le O(1/n).
	\]
	Hence, we have that
	\begin{align*}
	\mu = \Ex\left[ \sum_{i \in [n]} X_i \right] &= \Ex\left[ \sum_{i \in \supp(\calD_{\ell,s})} X_i \right] + \Ex\left[ \sum_{i \in \supp(\calD_{\ell,s})} X_i \right]\\
	&= (p_1 + p_2) \cdot \frac{n}{2}.
	\end{align*}
	Let
	\[
	\nu = (1 - e^{-1+\alpha}) \cdot n/2 + (1 - e^{-1-\alpha}) \cdot n / 2 = (1 - e^{-1}\cosh(\alpha)) \cdot n,
	\]
	where the last equality holds since $\cosh(\alpha) := \frac{e^{\alpha} + e^{-\alpha}}{2}$. Let $X = \sum_{i=1}^{n} X_i$. Using the Chernoff bound and the fact that $|\nu - \mu | \le O(1)$, we have that
	\[
	\Pr\left[ \left|X - (1 - e^{-1}\cosh(\alpha)) \cdot n \right| > 10\sqrt{n}\right] < 0.01,
	\]
	which completes the proof.
\end{proof}

\subsection{$\localDP$ Lower bound}

\begin{lemma}\label{lm:KL-bound-par-alpha}
	For any $\eps > 0$ and  $\alpha \in [0,1]$, if $R$ is $(\eps, \delta,\calU_D)$-dominated, then we have that
	\[
	\Ex_{\ell,j \in [2]\times \{0,1\}^D} [\KL(P^{\alpha}_{\ell,j} || Q)] \le  \alpha^2 \cdot \frac{4 e^{\eps}}{2^D} + 4\delta.
	\]
\end{lemma}
\begin{proof}
	We follow closely the proof of Lemma~\ref{lm:KL-bound-par}.
	To apply Theorem~\ref{theo:many-vs-one}, we set the index set as $\calI = [2] \times \{0,1\}^D$, the distribution $\pi$ to be the uniform distribution over $\calI$, $\{ \lambda_v \}_{v \in \calI} = \{ \calD_v^{\alpha} \}_{v \in \calI}$, and $\mu = \calU_{D}$.
	
	Clearly, $\mu$ $2$-dominates $\{\lambda_v\}$. Let $f$ be a function such that $\|f\|_{\infty} = L$ and $f(\mu) = 1$. It follows that
	\[
	\Ex_{v \leftarrow \pi} |f(\mu) - f(\lambda_v)|^2 =\Ex_{(\ell,s) \in [2] \times \{0,1\}^D} \alpha^2 \cdot \hat{f}(s)^2.
	\]
	Recall that
	\[
	\Ex_{(\ell,s) \in [2] \times \{0,1\}^D} \hat{f}(s)^2 \le \frac{L}{2^D}.
	\]
	Therefore, we can set $W(L) := \alpha^2 \cdot \frac{L}{2^D}$. Clearly, $W$ is a concave function. By Theorem~\ref{theo:many-vs-one}, it follows that
	\[
	\Ex_{\ell,j \in [2]\times[D]} [\KL(R(\calD_{\ell,j}) || R(\calU_D))] \le 2 W(2e^{\eps}) + 4 (\tau - 1)^2 \cdot \delta \le \alpha^2 \cdot \frac{4 e^{\eps}}{2^D} + 4\delta.
	\qedhere
	\] 
\end{proof}

We now show that the \DE function is hard for $(\eps,\delta)$-local algorithms.

\begin{reminder}{Theorem~\ref{theo:lowb-distinct-elements}.}
	For $\eps \le 0.49 \cdot \ln n $, if $P$ is a public-coin $(\eps,o(1/n))$-$\localDP$ protocol, then it cannot compute $\DE_{n,n}$ with error $o(n/e^{\eps})$ and probability at least $0.99$.
\end{reminder}
\begin{proof}
	Let $D = \log n$. We identify the input space $[n]$ with $\{0,1\}^{D}$ in the natural way. Suppose there is a public-coin $(\eps,o(1/n))$-$\localDP$ protocol $P$ solving $\DE_{n,n}$ with error $o(n/e^{\eps})$ and probability at least $0.99$.
	 
	Let $R$ with public randomness from $\calDpb$ be the randomizer used in $P$. For a dataset $W$, we use $P(W)$ (respectively, $P_{\gamma}(W)$) to denote the output of $P$ on the dataset $W$ (with public randomness fixed to $\gamma$).

	Setting $\alpha^2 = \frac{1}{20 e^{\eps}}$, we let $\mu_\alpha = (1 - e^{-1} \cosh(\alpha)) \cdot n$ and $\mu_0 = (1- e^{-1}) \cdot n$.
	
	By our assumption on $P$, Proposition~\ref{prop:DE-bound} and a union bound, it follows that for every $(\ell,s) \in [2] \times \{0,1\}^D$, we have
	\[
	\Ex_{\gamma \leftarrow \calDpb}\left[\Pr\left[\left| P_{\gamma}((\calD^{\alpha}_{\ell,s})^{\otimes n}) - \mu_\alpha \right| \le \frac{n}{1000 e^{\eps}} + 10 \sqrt{n} \right] \right] \ge 0.98.
	\]
	
	Similarly, we have
	\[
	\Ex_{\gamma \leftarrow \calDpb}\left[\Pr\left[\left| P_{\gamma}(\calU_{D}^{\otimes n}) - \mu_0 \right| \le \frac{n}{1000 e^{\eps}} + 10 \sqrt{n} \right] \right] \ge 0.98.
	\]
	
	Note that by our choice of $\eps$, we have $\frac{n}{1000 e^{\eps}} + 10 \sqrt{n} < \frac{n}{800 e^{\eps}}$. By a union bound, it follows that with probability at least $0.5$ over $\gamma \leftarrow \calDpb$, we have
	\begin{align}
	\Pr\left[\left| P_{\gamma}((\calD^{\alpha}_{\ell,s})^{\otimes n}) - \mu_\alpha \right| < \frac{n}{800 e^{\eps}}\right] \ge 0.8 ~~\text{and}~~
	\Pr\left[\left| P_{\gamma}(\calU_{D}^{\otimes n}) - \mu_0 \right| < \frac{n}{800 e^{\eps}}\right] \ge 0.8 \notag \\
	\text{for at least a 0.5 fraction of $(\ell,s) \in [2] \times \{0,1\}^D$.} \label{eq:condition-on-gamma}
	\end{align}
	
	By the definition of public-coin $\localDP$ protocols, we have that with probability at least $0.99$ over $\gamma \leftarrow \calDpb$, $R_{\gamma}$ is $(\eps,o(1/n),\calU_{D})$-dominated. By a union bound, there exists a $\gamma$ such that $R_{\gamma}$ is $(\eps,o(1/n),\calU_{D})$-dominated and it satisfies the condition in~\eqref{eq:condition-on-gamma}. We fix such a $\gamma$.
	
	By Lemma~\ref{lm:KL-bound-par-alpha}, it follows that
	\begin{align*}
	\Ex_{(\ell,s) \in [2]\times \{0,1\}^D} [\KL(R_\gamma(\calD^{\alpha}_{\ell,s}) || R_\gamma(\calU_D))] &\le \alpha^2 \cdot \frac{2 e^{\eps}}{2^D} + o(1/n).
	\end{align*} 
	
	Recall that $\alpha^2 = \frac{1}{20 e^{\eps}}$, the above further simplifies to
	\[
	\Ex_{(\ell,s) \in [2]\times \{0,1\}^D} [\KL(R_\gamma(\calD^{\alpha}_{\ell,s}) || R_\gamma(\calU_D))] \le \frac{1}{10n} + o(1/n).
	\]
	
	Let $S$ be the set of $(\ell,s)$ satisfying the conditions on $(\ell,s)$ stated in~\eqref{eq:condition-on-gamma}. Since $S$ contains at aleast a $0.5$ fraction of $[2] \times \{0,1\}^{D}$, it follows that
	\[
	\Ex_{(\ell,s) \in S} [\KL(R_\gamma(\calD^{\alpha}_{\ell,s}) || R_\gamma(\calU_D))] \le \frac{1}{5n} + o(1/n).
	\]
	
	This means that there exists a pair $(\ell,s) \in S$ such that $\KL(R_\gamma(\calD^{\alpha}_{\ell,s}) || R_\gamma(\calU_D)) \le 1/5n + o(1/n)$. We fix such a pair $(\ell,s)$.
	
	We have
	$
	\KL(R_\gamma(\calD^{\alpha}_{\ell,s})^{\otimes n} || R_\gamma(\calU_D)^{\otimes n}) \le 1/5 + o(1).
	$
	By Pinsker's inequality, it follows that
	\[
	\| R_\gamma(\calD^{\alpha}_{\ell,s})^{\otimes n} - R_\gamma(\calU_D)^{\otimes n}) \|_{TV} \le \sqrt{1/2 \cdot 1/5 + o(1)} \le 0.4.
	\]
	
	Since $(\ell,s) \in S$, it follows that
	\begin{equation}
	\Pr\left[\left| P_{\gamma}(\calU_D^{\otimes n}) - \mu_\alpha \right| < \frac{n}{800 e^{\eps}}\right] \ge 
	\Pr\left[\left| P_{\gamma}((\calD^{\alpha}_{\ell,s})^{\otimes n}) - \mu_\alpha \right| < \frac{n}{800 e^{\eps}}\right] - 0.4 \ge 0.4. \label{eq:cond-111}
	\end{equation}
	On the other hand, we also have
	\begin{equation}
	\Pr\left[\left| P_{\gamma}(\calU_D^{\otimes n}) - \mu_0 \right| < \frac{n}{800 e^{\eps}}\right] \ge 0.8.
	\label{eq:cond-222}
	\end{equation}
	
	Note that $|\mu_\alpha - \mu_0| = e^{-1} (\cosh(\alpha) - 1) \cdot n \ge \alpha^2/2 \cdot e^{-1} \cdot n > n/200e^{\eps}$. Hence \eqref{eq:cond-111} and~\eqref{eq:cond-222} give a contradiction.
\end{proof}

%Since $(\eps,\delta)$-$\localDP$ algorithms are $(\eps,\delta)$-dominated, Theorem~\ref{theo:lowb-distinct-elements} follows directly from Lemma~\ref{lm:dominated-DE-hard}.

	\section{Low-Message $\shuffledDP$ Protocols for \DE}\label{sec:mult_msg_shuffle_prot_count_dist}

In this section, we present our low-message $\shuffledDP$ protocols for \DE, thereby proving Theorem~\ref{theo:upper-bound-DE}. 

In Section~\ref{sec:low-message-UB-DE-intuition}, we review the previous protocol of~\cite{BCJM20}, and discuss some intuitions underlying our improvement. In Section~\ref{sec:UB-DE-prelim}, we introduce some necessary definitions and technical tools. Next, in Section~\ref{sec:UB-DE-base-protocol} we present our private-coin protocol (stated in Theorem~\ref{theo:upper-bound-DE-base}) for \DE with error $\tilde{O}(\sqrt{D})$, which uses $1/2 + o(1)$ message per user in expectation when the input universe size is below $n/ \polylog(n)$. We will also show that a simple modification of this protocol is $(\ln(n) + O(1))$-$\localDP$, thereby proving Theorem~\ref{th:weak_privacy_local_prot_DE}. Finally, based on the private-coin protocol, in Section~\ref{sec:UB-DE-main-protocol} we prove Theorem~\ref{theo:upper-bound-DE} by presenting our public-coin protocol for \DE, which uses less than $1$ message per user in expectation without any restriction on the universe size.

\subsection{Intuition} \label{sec:low-message-UB-DE-intuition}
\newcommand{\shuffledMtwoDP}{\typeOfDP{mod2\text{-}shuffle}}

\newcommand{\PBCJM}{P_{\sf BCJM}}
\newcommand{\Pmtwo}{P_{\sf mod 2}}

We now turn to sketch the main ideas behind Theorem~\ref{theo:upper-bound-DE-base} and Theorem~\ref{theo:upper-bound-DE}. It would be instructive to review the $\widetilde{O}(D)$-message $\shuffledDP$ protocol solving $\DE_{n,D}$ with error $O(\sqrt{D})$ from~\cite{BCJM20}.

\highlight{The $\shuffledMtwoDP$ Model.} To gain more insights about their protocol, we consider the following \emph{mod 2 shuffle model} ($\shuffledMtwoDP$), where two messages of the same content ``cancel each other'', i.e.,  the transcript is now a random permutation of messages that appear an odd number of times. 

The DP requirement now applies to this new version of transcript. The same holds for the analyzer, who now can only see the new version of transcript.~\cite{BCJM20} first gave a $\shuffledMtwoDP$ protocol for $\DE$, and then adapted that protocol to the standard $\shuffledDP$ model using the Ishai et al. protocol for secure aggregation~\cite{ishai2006cryptography}.\footnote{They did not explicit specify their protocol in the $\shuffledMtwoDP$ model, but it is implicit in their proof of security.}
%In~\cite{BCJM20}, for each $i \in [D]$, their protocol (referred as $P_{\sf BCJM}$ in the following) runs a copy $P_i$ for computing the sum of $\F_2$, using the IKOS protocol. 

\highlight{Low-Message Protocol in $\shuffledMtwoDP$.} The $\shuffledMtwoDP$ protocol of~\cite{BCJM20} (referred as $\Pmtwo$ in what follows) first sets a parameter $q = \Theta(1/n)$ so that $ \Pr[\Bin(n,q) \equiv 1 \pmod 2] = 1 / (2e^{\eps/2})$.
Next, for each user holding an element $x \in [D]$, the user first sends $x$ with probability $1/2$. Then for each $j \in [D]$, the user sends message $j$ with probability $q$. All these events are independent. 

Finally, if there are $z$ messages in the transcript (i.e., there are $z$ messages occurring an odd number of times in the original transcript), then the analyzer outputs $(2 \cdot z \cdot e^{\eps / 2} - D) / (e^{\eps/2} - 1)$ as the estimate.
Note that a user sends $1/2 + D \cdot q = 1/2 + O(D / n)$ message in expectation.

\highlight{Analysis of the Protocol $\Pmtwo$.} It is shown in~\cite{BCJM20} that the above protocol is $\eps$-DP and solves $\DE_{n,D}$ with error $O(\sqrt{D})$. Here we briefly outline the intuition behind it. 

Let $S$ be the set consisting of all inputs of the users. We can see that every $i \in S$ belongs to the transcript with probability exactly $1/2$; on the other hand, every $i \in [D] \setminus S$ belongs to the transcript with probability exactly $\Pr[\Bin(n,q) \equiv 1 \pmod 2] = \frac{1}{2e^{\eps/2}}$. Moreover, all these events are independent. Therefore, a simple calculation shows that $(2 \Ex[z] e^{\eps / 2} - D) / (e^{\eps/2} - 1) = |S|$, and the accuracy follows from a Chernoff bound. As for the DP guarantee, changing the input of one user only affects the distributions of two messages in the transcript, and it only changes each message's occurrence probability in the transcript from $1/2$ to $1/2e^{\eps/2}$ or vice versa.

\highlight{From $\shuffledMtwoDP$ to $\shuffledDP$.} To obtain an actual $\shuffledDP$ protocol from $\Pmtwo$, the protocol from~\cite{BCJM20} (which we henceforth denote by $\PBCJM$) runs $D$ copies of the protocol for securely computing sum over $\F_2$~\cite{ishai2006cryptography}, such that the $i$-th protocol $P_i$ aims to simulate the number of occurrences of message $i$ modulo $2$. For each user $i$, if it were to send a message $i$ in $\Pmtwo$, it sends one in $P_i$; otherwise it sends zero in $P_i$.

Since the~\cite{ishai2006cryptography} protocol for computing sum over $\F_2$ requires $O\left(\frac{\log(1/\delta)}{\log n} + 1\right)$ messages from each user~\cite{ghazi2019privateEurocrypt,balle_merged}, each user needs to send $O\left(D \cdot \left(\frac{\log(1/\delta)}{\log n}+1\right)\right)$ messages in total. Moreover, from the security condition of $P_i$, for each message $i$ the transcript only reveals the parity of its number of occurrences, which is exactly what we need in order to simulate $\shuffledMtwoDP$ protocols.

\highlight{Our Improvement.} Note that $\PBCJM$ requires significantly more messages per user than that of $\Pmtwo$. Our goal here is to compile $\Pmtwo$ to $\shuffledDP$ in a much more efficient way, ideally with no overhead. In $\Pmtwo$ each user sends only $1/2 + O(D/n)$ message. This means that when translating to $\PBCJM$, users end up sending many zero messages in the $P_i$ subprotocols, which is wasteful.

Our crucial idea for improving on the aforementioned protocol is a very simple yet effective alternative to the secure aggregation protocol over $\F_2$ of~\cite{ishai2006cryptography} used in $\PBCJM$. In our new subprotocol $P_i$, if a user were to send a message $i$ in $\Pmtwo$, it sends one to $P_i$; otherwise it draws $\lambda$ from a noise distribution $\calD$ (such that $\Ex[\calD] \approx \polylog(\delta^{-1}) / n$) and sends $2 \lambda$ many ones to $P_i$. Clearly, our new $P_i$ still maintains the parity of occurrences of each messages, and the expected number of messages is roughly $2 \cdot \Ex[\calD] \cdot D + 1/2 = O(\polylog(\delta^{-1})) \cdot D/n + 1/2$. To show that the resulting protocol is $\shuffledDP$, we build on the techniques of~\cite{ghazi2020private}, which show that the noise added can hide the contribution of a single user.

\subsection{Preliminaries}\label{sec:UB-DE-prelim}

We first recall the definition of the negative binomial distribution.
\begin{definition}\label{defi:NB-distr}
	Let $r > 0$ and $p \in [0,1]$, the negative binomial distribution $\NB(r,p)$ is defined by $\Pr[\NB(r,p) = k] = \binom{k+r-1}{k} (1-p)^{r}p^{k}$ for each non-negative integer $k$.\footnote{For a real number $\alpha$, $\binom{\alpha}{k} := \prod_{i=0}^{k-1} \frac{\alpha-i}{i+1}$.}
\end{definition}

We recall the following key properties of the negative binomial distribution: (1) For $\alpha,\beta > 0$ and $p \in [0,1]$, $\NB(\alpha,p) + \NB(\beta,p)$ has the same distribution as $\NB(\alpha + \beta,p)$; (2) $\Ex[\NB(r,p)] = \frac{pr}{1-p}$.

We will need the following lemma from~\cite{ghazi2020private}.

\begin{lemma}\label{lm:hiding-NB}
	For any $\eps > 0, \delta \in (0,1)$, and $\Delta \in \mathbb{N}$, let $p = e^{-0.1\eps/\Delta}$ and $r = 50 \cdot e^{\eps/\Delta} \cdot \log(\delta^{-1})$. For any $k \in \{-\Delta,-\Delta+1,\dotsc,\Delta - 1,\Delta \}$, $d_{\eps}(k + \NB(r,p) || \NB(r,p)) \le \delta$.
\end{lemma}

The following is a simple corollary of Item (2) of Proposition~\ref{prop:key-facts-hockey-stick} and Lemma~\ref{lm:hiding-NB}.

\begin{cor}\label{cor:hiding-NB}
	For any $\eps > 0, \delta \in (0,1)$, and $\Delta \in \mathbb{N}$, let $p$ and $r$ be as in~Lemma~\ref{lm:hiding-NB}. For any two distributions $X$ and $Y$ on $\{0,1,2,\dotsc,\Delta\}$, $d_{\eps}(X + \NB(r,p) || Y + \NB(r,p)) \le \delta$.
\end{cor}

\subsection{A Private-Coin Base Protocol}\label{sec:UB-DE-base-protocol}

Recall that $\DE_{n,D}$ denotes the restriction of \DE such that every user gets an input from $[D]$, and the goal is compute the number of distinct elements among all users. 

We are now ready to prove Theorem~\ref{theo:upper-bound-DE-base}, which is the private-coin case of Theorem~\ref{theo:upper-bound-DE}. To simplify the privacy analysis of the protocol and ease its application in Section~\ref{sec:UB-DE-main-protocol}, we also allow the input to be $0$, which means that the user's input is not counted. 

\begin{theorem} \label{theo:upper-bound-DE-base}
	For any $\eps \le O(1)$ and $\delta \le 1/n$, there is a private-coin $(\eps,\delta)$-$\shuffledDP$ protocol computing $\DE_{n,D}$ with error $O\left(\sqrt{D} \cdot \eps^{-1}\right)$ with probability at least $0.99$. Moreover, the expected number of messages sent by each user is $\frac{1}{2} + O\left( \frac{\log(1/\delta)^2 \ln(2/\eps)}{\eps} \cdot \frac{D}{n} \right)$.
\end{theorem}

\begin{proof}
Without loss of generality, we can assume that $\epsilon \le 1$. The algorithm requires several global constants that only depend on the values of $n,\eps$, and $\delta$. Algorithm~\ref{algo:global-constants} specifies these constants.
	Here, $c_0$ is a sufficiently large constant to be specified later.
	
	\begin{algorithm}
		\caption{Set-Global-Constants($n$, $\eps$, $\delta$) 
			\label{algo:global-constants}}
		\KwIn{$n$ is the number of users and the pair $(\eps,\delta)$ specifies the DP guarantee.}
		$\eps_0 = \min(\eps/6, 0.01)$\;
		$\Delta = \left\lceil c_0 \cdot \log \delta^{-1} \ln(\eps_0^{-1}) + 1 \right\rceil$\;
		$p = e^{-0.1\eps_0/\Delta}$\;
		$r = 50 \cdot e^{\eps_0/\Delta} \cdot \log(10\delta^{-1})$\;
		$q = \frac{1-(1-e^{-\eps_0})^{1/n}}{2}$\;
	\end{algorithm}
	
	Next, we specify the randomizer and the analyzer of the protocol in Algorithm~\ref{algo:randomizer} and Algorithm~\ref{algo:analyzer} respectively.
	
	\begin{algorithm}
		\caption{Randomizer($x$, $D$, $n$, $\eps$, $\delta$) 
			\label{algo:randomizer}}
		\KwIn{$x \in \{0\} \cup [D]$ is the user's input. $D$ is the universe size.} 
			%\\$n$ is the number of users and the pair $(\eps,\delta)$ specifies the DP requirement.}
		Set-Global-Constants($n$, $\eps$, $\delta$)\;
		Toss a uniformly random coin to get $v \in \{0,1\}$\;
		\If{$v = 1$ and $x \ne 0$}{
			send message $(x)$\;
		}
		\For{$i \in [D]$}{
			Let $ y \leftarrow \Ber(q')$\;
			\If{$y = 1$}{
				send message $(i)$\;
			}
			Let $\eta \leftarrow \NB(r/n,p)$\;
			Send $2 \cdot \eta$ messages $(i)$\;
		}
	\end{algorithm}
	
	\begin{algorithm}
		\caption{Analyzer($S$, $D$, $n$, $\eps$, $\delta$)
			\label{algo:analyzer}}
		\KwIn{$S$ is the multi-set of messages. $D$ is the universe size.}
			%\\$n$ is the number of users and the pair $(\eps,\delta)$ specifies the DP requirement.}
		Set-Global-Constants($n$, $\eps$, $\delta$)\;
		
		\For{$i \in [D]$}{
			Let $y_i$ be the number of message $(i)$ in $S$\;
			$C_i = y_i \bmod{2}$\;
		}
		
		$C = \sum_{i=1}^{D} C_i$\;
		$z = \frac{2Ce^{\eps_0} - D}{e^{\eps_0} - 1}$\;
		\KwRet{$z$}\;
	\end{algorithm}
	
	\paragraph{Accuracy Analysis.} We first analyze the error of our protocol. Let $E$ be the set $\{x_i\}_{i \in [n], x_i \ne 0}$. Recall that the goal is to estimate $|E|$.
	
	For each $i \in [D]$, we analyze the distribution of the random variable $C_i$ in Algorithm~\ref{algo:analyzer}. We observe that: (1) if $i \in E$, then $C_i$ is distributed uniformly at random over $\{0,1\}$; (2) if $i \notin E$, then $C_i$ is distributed as $\Ber\left(\frac{1}{2 e^{\eps_0}} \right)$ by Lemma~\ref{lm:distr-C_i}; (3) $\{C_i\}_{i \in [D]}$ are independent.
	
	\begin{lemma}[{\cite[Lemma~3.5]{BCJM20}}]\label{lm:distr-C_i}
		Let $n,q'$ be specified as in Global-Constants$(n,\eps,\delta)$.  Then, $[ \Bin(n,q') \mod 2]$ is distributed identically to $\Ber\left( \frac{1}{2 e^{\eps_0}} \right)$.
	\end{lemma}
	
	Hence, we have that $\Ex[C] = |E| \cdot \frac{1}{2} + (D - |E|) \cdot \frac{1}{2 e^{\eps_0}}$. Plugging in the equation defining the output $z$, we have $\Ex[z] = \Ex\left[\frac{2Ce^{\eps_0} - D}{e^{\eps_0} - 1}\right] = |E|$. An application of Hoeffding's inequality implies that 
	\[
	\Pr\left[|z - |E|| > c \cdot (\eps_0)^{-1} \cdot \sqrt{D} \right] < 0.01,
	\] for a sufficiently large constant $c$. Hence, with probability at least $0.99$, the error of the protocol is less than $c \cdot (\eps_0)^{-1} \cdot \sqrt{D} = O(\sqrt{D} \cdot \eps^{-1})$.

	\paragraph*{Privacy Analysis.} We now prove that our protocol is indeed $(\eps,\delta)$-\DP.
	Note that the multi-set of messages $S$ can be described by integers $(y_i)_{i \in [D]}$ (corresponding to the histogram of the messages). 
	
	Consider two neighboring datasets $x = (x_1,x_2,\dotsc,x_n)$ and $x' = (x_1',x_2,\dotsc,x_n)$ (without loss of generality, we assume that they differ at the first user).  Let $Y$ and $Y'$ be the corresponding distributions of $(y_i)_{i \in [D]}$ given input datasets $x$ and $x'$. The goal is to show that they satisfy the $(\eps,\delta)$-DP constraint. That is, we have to establish that $d_{\eps}(Y||Y') \le \delta$.
	
	To simplify the analysis, we introduce another dataset $\bar{x} = (0,x_2,\dotsc,x_n)$, and let $\bar{Y}$ be the corresponding distribution of $(y_i)_{i \in [D]}$ given input dataset $\bar{x}$. By the composition rule of $(\eps,\delta)$-DP, it suffices to show that the pairs $(Y,\bar{Y})$ and $(\bar{Y},Y')$ satisfy $(\eps/2,\delta/3)$-DP (note that $\eps <1$, and $\delta/3 + e^{\eps/2} \cdot \delta/3 \le \delta$). By symmetry, it suffices to consider the pair $(Y,\bar{Y})$ and prove that $d_{\eps/2}(Y||\bar{Y}) \le \delta/3$.
	
	Let $i = x_1$, and $m_i$ be the number of times that $i$ appears in $x_2,\dotsc,x_n$. First note that all coordinates in both $Y$ and $\bar{Y}$ are independent, and furthermore the marginal distribution of $Y$ and $\bar{Y}$ on coordinates in $[D] \setminus \{i\}$ are identical. Hence, by Item (1) of Proposition~\ref{prop:key-facts-hockey-stick}, it suffices to establish that $Y_i$ and $\bar{Y}_i$ satisfy $(\eps/2,\delta/3)$-DP.
		
	The distribution of $\bar{Y}_i$ is $\Bin(n,q') + 2 \cdot \NB(r,p) + \Bin(m_i,1/2)$, and the distribution of $Y_i$ is $\Bin(n,q') + 2 \cdot \NB(r,p) + \Bin(m_i+1,1/2)$.\footnote{
		Recall that for two random variables $X$ and $Y$, we use $X + Y$ to denote the random variable distributed as a sum of two independent samples from $X$ and $Y$.} Since $\Bin(m_i+1,1/2) = \Bin(m_i,1/2) + \Ber(1/2)$, it suffices to consider the case where $m_i = 0$ by Item (1) of Proposition~\ref{prop:key-facts-hockey-stick}.
	
	We need the following lemma, whose proof is deferred until we finish the proof of Theorem~\ref{theo:upper-bound-DE-base}.
	
	\begin{lemma}\label{lm:privacy-bound}
		Let $n,q',\lambda$ be specified as in Set-Global-Constants$(n,\eps,\delta)$, $X = \Bin(n,q') + 2 \cdot \NB(r,p)$, and $Y = \Bin(n,q') + 2 \cdot \NB(r,p) + \Ber(1/2)$.  Then,
		\[
		d_{\eps/2}(X||Y) \le \delta/3 \quad\text{and}\quad d_{\eps/2}(Y||X) \le \delta/3.
		\]
	\end{lemma}

	By Lemma~\ref{lm:privacy-bound} and previous discussions, it follows that $D_{\eps}(Y||Y') \le \delta$, which shows that our protocol is $(\eps,\delta)$-\DP as desired.
	
	In the following, we will need the proposition below which gives us an estimate on $q'$.
	
	\begin{prop}\label{prop:bound-on-q'}
		Let $n,q',\eps_0$ be specified as in Set-Global-Constants$(n,\eps,\delta)$. Then, $q' \le O(\ln(\eps_0^{-1}) / n)$.
	\end{prop}
	\begin{proof}
		Since $\eps_0 \le 0.01$, we have $e^{-\eps_0} \le 1 - \eps_0/2$. Hence, $1 - e^{-\eps_0} \ge \eps_0 / 2$. Plugging in the definition of $q'$, it follows that $(1-e^{-\eps_0})^{1/n} \ge e^{\ln(\eps_0/2) / n} \ge 1  + \ln(\eps_0/2) / n$.  Finally, it follows that
		\[
		q' = \frac{1 - (1-e^{-\eps_0})^{1/n}}{2} \le -\ln(\eps_0/2) / 2n = O(\ln(\eps_0^{-1}) / n).\qedhere
		\]
	\end{proof}
	
	\paragraph*{Efficiency Analysis.} We now analyze the message complexity of our protocol. Note that 
	\[
	\Ex[\NB(r/n,p)] = \frac{1}{n} \cdot \frac{pr}{1-p} = O\left(\frac{1}{n} \cdot \frac{\Delta}{\eps_0} \cdot \log(1/\delta)\right) = O\left(\frac{1}{n} \cdot \eps^{-1} \ln(2/\eps) \cdot \log(1/\delta)^2\right).
	\]
	By a straightforward calculation, each user sends 
	\[
	\frac{1}{2} + O(D \cdot \Ex[\NB(r/n,p)] +  D \cdot q') \le \frac{1}{2} + O\left( \frac{\log(1/\delta)^2 \ln(2/\eps)}{\eps} \cdot \frac{D}{n} \right)
	\]
	messages in expectation.
	\end{proof}
	
	Finally, we prove Lemma~\ref{lm:privacy-bound}.
	\begin{proofof}{Lemma~\ref{lm:privacy-bound}}
		We consider bounding $D_{\eps/2}(X||Y)$ first. Note that since $q' = O( \ln (\eps_0^{-1}) /n)$ by Proposition~\ref{prop:bound-on-q'}, we set the constant $c_0$ so that
		\[
		\Pr\left[\Bin(n,q') > c_0 \cdot \log \delta^{-1} \ln (\eps^{-1}) \right] \le \delta/10.
		\]
		Recall that $\Delta = \lceil c_0 \cdot \log \delta^{-1} \ln (\eps^{-1}) + 1 \rceil$, and note that our choices of $r$ and $p$ satisfy Lemma~\ref{lm:hiding-NB} with privacy parameters $\eps_0 \le \eps/6$ and $\delta / 10$.
		
		Now, let $A = \Bin(n,q')$, $N = \NB(r,p)$, and $B = \Ber(1/2)$.
		
		To apply Item (2) of Proposition~\ref{prop:key-facts-hockey-stick}, we are going to decompose $X = A + 2 \cdot N$ and $Y = A + 2 \cdot N + B$ into a  weighted sum of three sub-distributions. 
		
		\paragraph*{Decomposition of $X = A + 2 \cdot N$.} We define three events on $A$ as follows:
		\[
		\calE_{big} = [A > c_0 \cdot \log (\delta^{-1})], \quad\calE_{even} = [A \le c_0 \cdot \log (\delta^{-1}) \wedge A \equiv 0 \bmod 2],
		\]
		and
		\[
		\calE_{odd} = [A \le c_0 \cdot \log (\delta^{-1}) \wedge A \equiv 1 \bmod 2].
		\]
		We let $\alpha_{big} = \Pr_{A}[\calE_{big}]$, $\alpha_{even} = \Pr_{A}[\calE_{even}]$, $\alpha_{odd} = \Pr_{A}[\calE_{even}]$.
		
		From our choice of $c_0$, we have $\alpha_{big} = \Pr\left[\Bin(n,q') > c_0 \cdot \log \delta^{-1} \ln (\eps^{-1})\right] \le \delta / 10$. Let $q = \frac{1}{2 e^{\eps_0}}$. By Lemma~\ref{lm:distr-C_i}, it follows that $|\alpha_{odd} - q| \le \delta / 10$ and $|\alpha_{even} - (1 - q)| \le \delta / 10$. 
		
		Therefore, let $A_{big} := A|\calE_{big}$, $A_{even} := A|\calE_{even}$ and $A_{odd} := A|\calE_{odd}$. We can now decompose $A + 2N$ as a mixture of components $A_{big} + 2N$, $A_{even} + 2N$ and $A_{odd} + 2N$ with corresponding mixing weights $\alpha_{big}$, $\alpha_{even}$ and $\alpha_{odd}$.
		%\[
		%A + 2N = \alpha_{big} \cdot (A_{big} + 2N) + \alpha_{even} \cdot (A_{even} + 2N) + \alpha_{odd} \cdot (A_{odd} + 2N).
		% the equation above contains two usage of + as convolution and mixture, very bad...
		%\]
		
		\paragraph*{Decomposition of $Y = A + 2 \cdot N + B$.}
		Now, we define three events on $(A,B)$ as follows
		\[
		\WT{\calE}_{big} = [A > c_0 \cdot \log (\delta^{-1})], \quad\WT{\calE}_{even} = [A \le c_0 \cdot \log (\delta^{-1}) \wedge A + B \equiv 0 \bmod 2],
		\]
		and
		\[
		\WT{\calE}_{odd} = [A \le c_0 \cdot \log (\delta^{-1}) \wedge A + B \equiv 1 \bmod 2].
		\]
		Similarly, we let $\beta_{big} = \Pr_{A,B}[\WT{\calE}_{big}]$, $\beta_{even} = \Pr_{A,B}[\WT{\calE}_{even}]$, $\beta_{odd} = \Pr_{A,B}[\WT{\calE}_{even}]$.
		
		By our choice of $c_0$, we have $\beta_{big} \le \delta / 10$. Since $\Pr[A + B \equiv 1 \bmod 2] = 1/2$, it follows that $|\beta_{even} - 1/2| \le \delta / 10$ and $|\beta_{odd} - 1/2| \le \delta / 10$. 
		
		Let $(A+B)_{big} := (A+B)|\WT{\calE}_{big}$, $(A+B)_{even} := (A+B)|\WT{\calE}_{even}$ and $(A+B)_{odd} := (A+B)|\WT{\calE}_{odd}$. We therefore decompose $A + 2N + B$ as a mixture of components $(A+B)_{big} + 2N$, $(A+B)_{even} + 2N$ and $(A+B)_{odd} + 2N$ with mixing weights $\beta_{big}$, $\beta_{even}$ and $\beta_{odd}$. 
		
		\paragraph*{Bounding $d_{\eps/2}(X||Y)$.}
		
		By Item (2) of Proposition~\ref{prop:key-facts-hockey-stick}, we have that
		\begin{align*}
		d_{\eps/2}(X||Y) \le& \alpha_{big} \\
		+&\alpha_{even} \cdot d_{\eps/2 + \ln (\beta_{even} / \alpha_{even})}(A_{even} + 2N || (A + B)_{even} + 2N)\\
		+& \alpha_{odd} \cdot d_{\eps/2 + \ln (\beta_{odd} / \alpha_{odd})} (A_{odd} + 2N || (A + B)_{odd} + 2N).   
		\end{align*}
		
		Now, note that $\frac{\beta_{odd}}{\alpha_{odd}} \ge 1$, and $ \frac{\beta_{even}}{\alpha_{even}} \ge \frac{1/2 - \delta/10}{1-q + \delta/10} \ge e^{-2\eps_0}$. It follows that $\eps/2 + \ln (\beta_{even} / \alpha_{even}) \ge \eps/2 - 2\eps_0 \ge \eps_0$ (since $\eps_0 \le \eps/6$), and $\eps/2 + \ln (\beta_{odd} / \alpha_{odd}) \ge \eps/2 \ge \eps_0$.
		
		Hence by Corollary~\ref{cor:hiding-NB}, we have that
		\[
		d_{\eps/2}(X||Y) \le \delta/10 + d_{\eps_0}(A_{odd} + 2N || (A + B)_{odd} + 2N) + d_{\eps_0}(A_{even} + 2N || (A + B)_{even} + 2N) \le \delta / 3.
		\]
		By a similar calculation, we can also bound $d_{\eps}(Y||X)$ by $\delta / 3$. 
	\end{proofof}

\paragraph*{Extension to Robust Shuffle Privacy.} Now we briefly discuss how to generalize the analysis of the above protocol in order to show that it also satisfies the stronger robust shuffle privacy condition. We first need the following formal definition of robust shuffle privacy.

\begin{definition}[\cite{BCJM20}]
	A protocol $P = (R,S,A)$ is $(\eps,\delta,\gamma)$-robustly shuffle differential private if, for all $n \in \mathbb{N}$ and $\gamma' \ge \gamma$, the algorithm $S_{\gamma' n} \circ R^{\gamma' n}$ is $(\eps,\delta)$-DP. In other words, $P$ guarantees $(\eps,\delta)$-shuffle privacy whenever at least a $\gamma$ fraction of users follow the protocol.
\end{definition}

Note that while the above definition requires the privacy condition  to be satisfied whenever there is at least a $\gamma$ fraction of users participating, the accuracy condition is only required when all users participate. That is, if some users drop from the protocol, then the analyzer does not need to output an accurate estimate of $\DE$.

\begin{theorem} \label{theo:upper-bound-DE-base-robust}
	For two constants $\gamma,\eps \in (0,1]$, and $\delta \le 1/n$, there is an $(\eps,\delta,\gamma)$-robustly shuffle differentially private protocol solving $\DE_{n,D}$ with error $O_{\gamma,\eps}\left(\sqrt{D}\right)$ and with probability at least $0.99$. Moreover, the expected number of messages sent by each user is $\frac{1}{2} + O_{\gamma,\eps}\left( \log(1/\delta)^2 \cdot \frac{D}{n} \right)$.\footnote{To make the notation clean, we choose not to analyze the exact dependence on $\eps$ and $\gamma$.}
\end{theorem}
\begin{proofsketch}
	To make the algorithm in Theorem~\ref{theo:upper-bound-DE-base} robustly shuffle private, we need the following modifications:
	\begin{itemize}
		\item In Algorithm~\ref{algo:global-constants}, we set 
		$q' = \frac{1-(1-e^{-\eps_0})^{1/(\gamma n)}}{2}$, instead of $q' = \frac{1-(1-e^{-\eps_0})^{1/n}}{2}$.
		\item In Algorithm~\ref{algo:randomizer}, we let $\eta \leftarrow \NB(r/(\gamma n),p)$, instead of $\eta \leftarrow \NB(r/n,p)$.
		\item In Algorithm~\ref{algo:analyzer}, we set $z = \frac{2C\tau - D}{\tau - 1}$ for $\tau = \frac{1}{1-(1-e^{\eps_0})^{1/\gamma}}$, instead of $z = \frac{2Ce^{\eps_0} - D}{e^{\eps_0} - 1}$.
	\end{itemize}
	
	The first two modifications guarantee that there is enough noise even when only $\gamma \cdot n$ users participate, so that the privacy analysis of Theorem~\ref{theo:upper-bound-DE-base} goes through. We now show that the last modification allows us to obtain an accurate estimate of $\DE_{n,D}$ when all users participate. In the following, we will use the same notation as in the proof of Theorem~\ref{theo:upper-bound-DE-base}. Note that we have
	\[
	q' = \frac{1-(1-e^{-\eps_0})^{1/(\gamma n)}}{2} = \frac{1 - (1-\tau^{-1})^{1/n} }{2}.
	\]
	Hence by Lemma~\ref{lm:distr-C_i}, $C_i$ is distributed as $\Ber(1/2\tau)$ when $i \notin E$. A similar calculation then shows that the error can be bounded by $O_{\tau}(\sqrt{D}) = O_{\eps,\gamma}(\sqrt{D})$.
\end{proofsketch}

\paragraph*{$\ln(O(n))$-$\localDP$ Protocol for $\DE$.}
Finally, we show that the protocol from Theorem~\ref{theo:upper-bound-DE-base} is also $\ln(O(n))$-$\localDP$ with a simple modification, which proves Theorem~\ref{th:weak_privacy_local_prot_DE} (restated below).

\begin{reminder}{Theorem~\ref{th:weak_privacy_local_prot_DE}.}
	There is a $(\ln(n) + O(1))$-$\localDP$ protocol computing $\DE_{n,n}$ with error $O(\sqrt{n})$.
\end{reminder}
\begin{proofsketch}
	Let $D = n$. We consider the following modification of Algorithm~\ref{algo:randomizer}.
	\begin{algorithm}
		\caption{Randomizer($x$, $D$, $n$, $\eps$, $\delta$) 
			\label{algo:randomizer-no-noise}}
		\KwIn{$x \in [D]$ is the user's input. $D$ is the universe size.} 
		%\\$n$ is the number of users and the pair $(\eps,\delta)$ specifies the DP requirement.}
		$\eps_0 = 1$, $q' = \frac{1-(1-e^{-\eps_0})^{1/n}}{2}$\;
		Toss a uniformly random coin to get $v \in \{0,1\}$\;
		\If{$v = 1$}{
			send message $(x)$\;
		}
		\For{$i \in [D] \setminus \{x\}$}{
			Let $ y \leftarrow \Ber(q')$\;
			\If{$y = 1$}{
				send message $(i)$\;
			}
		}
	\end{algorithm}

	That is, in Algorithm~\ref{algo:randomizer-no-noise} we remove the noise messages sampled from the distribution $2 \cdot \NB(r/n,p)$. Also, we do not send the same message more than once (the loop over $i$ skips the element $x$).
	
	When viewing it as a local protocol, we can assume that each user first collects all the messages it would send in Algorithm~\ref{algo:randomizer-no-noise}, and then simply outputs the histogram (so our new local randomizer only sends a single message). The analyzer in the local protocol can then aggregate these histograms, and apply the analyzer in Algorithm~\ref{algo:analyzer}. By the same accuracy proof as in Theorem~\ref{theo:upper-bound-DE-base}, it follows that the protocol achieves error $O(\sqrt{n})$ with probability at least $0.99$. So it only remains to prove that the protocol is $\ln(O(n))$-$\localDP$.
	
	We let $R$ be the randomizer in Algorithm~\ref{algo:randomizer-no-noise}, and we use $\Hist(R(x))$ to denote the distribution of the histogram of the messages output by $R$, which is exactly the output distribution of our new local randomizer.
	
	Without loss of generality, it suffices to show that for all possible histograms $z \in \{0,1\}^{D}$ (note that Algorithm~\ref{algo:randomizer-no-noise} does not send a message more than once), it holds that
	\[
	\frac{\Hist(R(1))_z}{\Hist(R(2))_z}\le O(n).
	\]
	Note that $\frac{\Hist(R(1))_z}{\Hist(R(2))_z}$ only depends on the first two bits of $z$. By enumerating all possible combination of two bits, we can bound it by
	\[
	\frac{\Hist(R(1))_z}{\Hist(R(2))_z} \le \frac{1/2}{q'} \cdot \frac{1-q'}{1/2} = \frac{1-q'}{q'} \le O(n).
	\]
	The last inequality follows from the fact that $q' = \Theta(1 / n)$.
\end{proofsketch}

\subsection{Public-Coin  Protocol}\label{sec:UB-DE-main-protocol}

We are now ready to prove Theorem~\ref{theo:upper-bound-DE} (restated below).

\begin{reminder}{Theorem~\ref{theo:upper-bound-DE}.}
	For all $\eps  \le O(1)$ and $\delta \le 1/n$, there is a public-coin $(\eps,\delta)$-$\shuffledDP$ protocol computing $\DE_{n}$ with error $O\left(\sqrt{n}\cdot \log(\delta^{-1}) \cdot \eps^{-1.5} \cdot \sqrt{\ln(2/\eps)}\right)$ and probability at least $0.99$. Moreover, the expected number of messages sent by each user is at most $1$. 
\end{reminder}
\begin{proof}
	
	Let $c_1$ be the constant in Theorem~\ref{theo:upper-bound-DE-base} such that the expected number of messages is bounded by $ c_1 \cdot \frac{\log(1/\delta)^2 \ln(2/\eps)}{\eps} \cdot \frac{D}{n} + 1/2$.
	\paragraph*{The Protocol.} We set $D = \left\lfloor n \left/ \left(2 \cdot c_1 \cdot \frac{\log(1/\delta)^2\ln(2/\eps)}{\eps}\right)\right. \right\rfloor$ so that the foregoing expected number of messages is bounded by $1$. 
	
	Note that we can assume $\eps^{-1} \cdot \ln(2/\eps) \cdot \log(1/\delta)^2 = o(n)$ as otherwise we are only required to solve $\DE_n$ with the trivial error bound $O(n)$. Therefore, we have $D \ge 1$ and $n/D = O(\eps^{-1} \cdot \ln(2/\eps) \cdot \log(1/\delta)^2 )$.
	
	We are going to apply a reduction to the private-coin protocol for $\DE_{n,D}$ in Theorem~\ref{theo:upper-bound-DE-base}. The full protocol is as follows:
	
	\begin{itemize}
		\item Using the public randomness, the users jointly sample a uniformly random mapping $f\colon \calX \to [n]$ and a uniformly random permutation $\pi$ on $[n]$.
		
		\item For each user holding an input $x \in \calX$, it computes $z = \pi(f(x))$, and sets its new input to be $z$ if $z \le D$, and $0$ otherwise. Then it runs the private-coin protocol in Theorem~\ref{theo:upper-bound-DE-base}.
		
		\item Let $f_{n}(m) := n \cdot \left( 1 - \left( 1 - \frac{1}{n} \right)^m \right) $. The analyzer first runs the analyzer in Theorem~\ref{theo:upper-bound-DE-base} to obtain an estimate $\bar{z}$. Then it computes $\hat{z} = \bar{z} \cdot \frac{n}{D}$, and outputs
		\[
		z = \argmin_{m \in \{0,1,\dotsc,n\}} |f_{n}(m) - \hat{z}|.
		\]
	\end{itemize}
	
	\paragraph*{Analysis of the Protocol.} The privacy of the protocol above follows directly from the privacy property of the protocol from Theorem~\ref{theo:upper-bound-DE-base}. Moreover, the bound on the expected number of messages per user simply follows from our choice of the parameter $D$.
	
	It thus suffices to establish the accuracy guarantee of the protocol. Let $S = \{x_i\}_{i \in [n]}$ be the set of all inputs, and the goal is to estimate $|S|$. We also let $\hat{S} = \{f(x_i)\}_{i \in [n]}$ and $\bar{S} = \hat{S} \cap \pi^{-1}([D])$.
	
	By the accuracy guarantee of Theorem~\ref{theo:upper-bound-DE-base}, it follows that with probability at least $0.99$, we have 
	\[
	|\bar{z} - |\bar{S}|| \le O(\sqrt{D} \cdot \eps^{-1}).
	\]
	In the following, we will condition on this event.
	
	\paragraph{Proving that $\hat{z}$ is a good estimate of $|\hat{S}|$.} Next, we show that $\hat{z}$ is close to $|\hat{S}|$. We will rely on the following lemma.
	
	\begin{lemma}\label{lm:sampling-is-good}
		For a uniformly random permutation $\pi \colon [n] \to [n]$ and a fixed set $E$, for every $B \in [1,n]$ such that $n/B$ is an integer, let $E_{\pi,n/B} = E \cap \pi^{-1}([n/B])$, we have
		\[
		\Pr_{\pi}\left[ \Big| |E_{\pi,n/B}| \cdot B - |E| \Big| \ge 10 \cdot \sqrt{B \cdot |E|}\right] \le 0.01.
		\]
	\end{lemma}
	\begin{proof}
		For each $i \in [n]$, let $X_i = X_i(\pi)$ be the indicator that $i \in E_{\pi,n/B}$. Note that these $X_i$'s are not independent, but they are negatively correlated~\cite[Proposition~7 and 11]{DubhashiR98}, hence a  Chernoff bound still applies.
		
		Note that $\Ex[X_i] = \frac{1}{B} \cdot \frac{|E|}{n}$. 
		By a Chernoff bound, it thus follows that
		\[
		\Pr_{\pi}\left[ \Big| \sum_{i=1}^{n} X_i - n \cdot \Ex[X_1] \Big| \ge 10 \cdot \sqrt{n \cdot \Ex[X_i]} \right] \le 0.01,
		\]
		and hence
		\[
		\Pr_{\pi}\left[ \Big| |E_{\pi,n/B}| - |E|/B \Big| \ge 10 \cdot \sqrt{|E|/B}\right] \le 0.01.
		\]
		Scaling both sides of the above inequality by $B$ concludes the proof.
	\end{proof}
	
	We now set $B = \frac{n}{D}$ and recall that $\hat{z} = B \cdot \bar{z}$. By Lemma~\ref{lm:sampling-is-good}, with probability at least $0.98$, it holds that
	\begin{align*}
	|\hat{z} - |\hat{S}| | &\le \left|\hat{z} - |\bar{S}| \cdot B \right| + \left||\hat{S}| - |\bar{S}| \cdot B \right|\\ 
	&\le B \cdot |\bar{z} - |\bar{S}|| + O(\sqrt{B \cdot n})\\ 
	&= O(B \sqrt{D} \cdot \eps^{-1}) + O(\sqrt{B \cdot n}) = O(\sqrt{B n} \cdot \eps^{-1}).
	\end{align*}
	
	\paragraph{Proving that $z$ is a good estimate of $|S|$.} Finally, we show that our output $z$ is a good estimate of $|S|$. To do so, we need the following lemma.
	
	\begin{lemma}\label{lm:hashing-is-good}
		Let $f_{n}(m) := n \cdot \left( 1 - \left( 1 - \frac{1}{n} \right)^m \right)$. For a uniformly random mapping $f\colon \calX \to [n]$ and a fixed set $E \subset \calX$ such that $|E| = m \le n$, we have that
		\[
		\Pr[ \left||\{ f(x) \}_{x \in E}| - f_{n}(m) \right| \ge 10 \sqrt{n} ] \le 0.01.
		\]
	\end{lemma}
	\begin{proof}
		For each $i \in [n]$, let $X_i = X_i(f)$ be the indicator whether $i \in \{ f(x) \}_{x \in E}$.  As before, these $X_i$'s are not independent but are negatively correlated~\cite[Proposition~7 and 11]{DubhashiR98}, hence a Chernoff bound still applies.
		
		Note that $\Ex[X_i] = \left( 1 - \left( 1 - \frac{1}{n} \right)^m \right)$. By a Chernoff bound, it thus follows that
		\[
		\Pr_{\pi}\left[ \Big| \sum_{i=1}^{n} X_i - n \cdot \Ex[X_1] \Big| \ge 10 \cdot \sqrt{n \cdot \Ex[X_i]} \right] \le 0.01.
		\]
		Noting that $\sum_{i=1}^{n} X_i = |\{ f(x) \}_{x \in E}|$ and $n \cdot \Ex[X_i] = f_{n}(m)$ completes the proof.
	\end{proof}

	By Lemma~\ref{lm:hashing-is-good}, it follows that with probability at least $0.99$, we have 
	$ \left| |\hat{S}| - f_{n}(|S|) \right| \le 10 \sqrt{n}$.
	
	Putting everything together, with probability at least $0.97$, we get that
	\[
	\left| \hat{z} - f_{n}(|S|) \right| \le O(\sqrt{B n} \cdot \eps^{-1}).
	\]
	
	The final step is to show that $z$ accurately estimates $|S|$. Recall that $
	z = \argmin_{m \in \{0,1,\dotsc,n\}} |f_{n}(m) - \hat{z}|$, which in particular means that
	\[
	|f_{n}(z) - \hat{z}| \le \left| f_{n}(|S|) - \hat{z} \right| \le O(\sqrt{B n} \cdot \eps^{-1}).
	\]
	By a triangle inequality, it follows that
	\[
	|f_{n}(z) - f_{n}(|S|)| \le O(\sqrt{B n} \cdot \eps^{-1}).
	\]
	
	We need the following lemma to finish the analysis.
	\begin{lemma}\label{lm:inverse-is-good}
		There is a constant $c > 0$ such that for all $a,b \in \{0,1,\dotsc,n\}$, it holds that
		\[
		|f_{n}(a) - f_{n}(b)| \ge c \cdot |a-b|.
		\]
	\end{lemma}
	\begin{proof}
		Suppose $a < b$ without loss of generality. Let $t = b - a$. We have that
		\begin{align*}
		f_{n}(b) - f_{n}(a) &= n \cdot \left( \left( 1 - \frac{1}{n} \right)^{a} - \left( 1 - \frac{1}{n} \right)^{b} \right)\\
		&= n \cdot \left( 1 - \frac{1}{n} \right)^{a} \cdot \left(1 - \left( 1 - \frac{1}{n} \right)^{t}\right)\\
		&= \Omega\left( n \cdot \frac{t}{n} \right) = \Omega(t).\qedhere
		\end{align*}
	\end{proof}

	Finally, by Lemma~\ref{lm:inverse-is-good}, with probability at least $0.97 > 0.9$, we obtain that $|z - |S|| \le O(\sqrt{B n} \cdot \eps^{-1})$, which concludes the proof.
\end{proof}
	\section{Lower Bounds in Two-Party Differential Privacy}
\label{sec:two-party}

In this section, we depart from the local and shuffle models, and instead consider the \emph{two-party} differential privacy~\cite{mcgregor2010limits}, which can be defined as follows.

\begin{definition}[DP in the Two-Party Model~\cite{mcgregor2010limits}]
There are two parties $A$ and $B$; $A$ holds $X = (x_1, \dots, x_n) \in \calX^n$ and $B$ holds $Y = (y_1, \dots, y_n) \in \calX^n$. Let $P$ be any randomized protocol between $A$ and $B$. Let $\VIEW^A_P(X, Y)$ denote the tuple ($X$, the private randomness of $A$, the transcript of the protocol). Similarly, let $\VIEW^B_P(X, Y)$ denote the tuple ($Y$, the private randomness of $B$, the transcript of the protocol).

We say that $P$ is \emph{$(\eps, \delta)$-$\twopartyDP$} if, for any $X, Y \in \calX^n$, the algorithms $$(y_1, \dots, y_n) \mapsto \VIEW^A_P(X, (y_1, \dots, y_n))$$  $$(x_1, \dots, x_n) \mapsto \VIEW^B_P((x_1, \dots, x_n), Y)$$ are both $(\eps, \delta)$-DP.
\end{definition}

We say that a two-party protocol $P$ computes a function $f\colon \calX^{2n} \to \R$ with error $\beta$ if, at the end of the protocol, at least one of the parties can output a number that lies in $f(x_1, \dots, x_n, y_1, \dots, y_n) \pm \beta$ with probability at least $0.9$.

We quickly note that, unlike in the local and shuffle models, we need not consider the public-coin and private-coin cases separately: as noted in~\cite{mcgregor2010limits}, the two parties may share fresh private random bits without violating privacy, meaning that public randomness is unnecessary.

The goal of this section is to prove Theorem~\ref{thm:two-party-lb}. To do this, we first state the necessary lower bound from~\cite{mcgregor2010limits} in Section~\ref{subsec:two-party-inner-prod}. We then give our reduction and prove Theorem~\ref{thm:two-party-lb} in Section~\ref{subsec:two-party-lb}. Finally, in Section~\ref{subsec:two-party-symmetrization}, we extend the lower bound to the case where the function is symmetric.

\subsection{Inner Product Lower Bound from~\cite{mcgregor2010limits}}
\label{subsec:two-party-inner-prod}

McGregor et al.~\cite{mcgregor2010limits} show that the inner product function is hard in the two-party model. Roughly speaking, they show that, if we let $X, Y$ be uniformly random strings, then, for any not-too-large $m \in \N$, no $(O(1),o(1/n))$-$\twopartyDP$ protocol can distinguish between $\left<X, Y\right> \mod m$ and a uniformly random number from $\{0, \dots, m - 1\}$. McGregor et al. use this result when $m = \tilde{\Omega}_\eps(\sqrt{n})$, but we will use their result for $m = 2$.

To avoid confusion in the next subsection, we will use $D$ in place of $n$ in this subsection. The following theorem is implicit in the proof of Theorem A.5 of~\cite{McGregorMPRTV11} (it follows by replacing $m = 6\Delta/\delta$ with $m = 2$ there). Recall that $\calU_D$ is the uniform distribution over $\{0, 1\}^D$.

\begin{theorem}[\cite{McGregorMPRTV11}] \label{thm:two-party-extractor}
Let $P$ be any $(\eps, \delta)$-$\twopartyDP$ protocol. Suppose $(X, Y) \leftarrow \calU_D^{\otimes 2}$ and let $Z$ be a uniformly random bit. Then, we have
\begin{align*}
\|(\VIEW^B_P(X, Y), \langle X, Y \rangle \mod 2) - (\VIEW^B_P(X, Y), Z)\|_{TV} \leq O\left(D\delta\right) + e^{-\Omega_{\epsilon}(D)}.
\end{align*}
\end{theorem}

It will be more convenient to state the above lower bound in terms of hardness of distinguishing two distributions, as we have done in the rest of this paper. To state this, we will need the following few notation. First, we use $\calD^0$ to denote the distribution $\calU_D^{\otimes 2}$ conditioned on the inner product of the two strings being $0 \bmod 2$; similarly, we use $\calD^1$ to denote the distribution $\calU_D^{\otimes 2}$ conditioned on the inner product of the two strings being $1 \bmod 2$. Furthermore, for any distribution $\calD$ on $(\{0, 1\}^D)^2$, we write $\VIEW^A_P(\calD)$ (respectively, $\VIEW^B_P(\calD)$) to denote the distribution of $\VIEW^A_P(X, Y)$ (respectively, $\VIEW^B_P(X, Y)$) when $X, Y$ are drawn according to $\calD$.

We may now state the following corollary, which is an immediate consequence of Theorem~\ref{thm:two-party-extractor}.

\begin{cor} \label{cor:dot-prod-dist}
Let $P$ be any $(\eps, \delta)$-$\twopartyDP$ protocol. Then, we have that
\begin{align*}
\|\VIEW^B_P(\calD^0) - \VIEW^B_P(\calD^1)\|_{TV} \leq O\left(D\delta\right) + e^{-\Omega_{\epsilon}(D)}.
\end{align*}
\end{cor}

\subsection{From Parity to the $\tilde{\Omega}(n)$ Gap}
\label{subsec:two-party-lb}

We will now construct the hard distributions that eventually give the gap of $\tilde{\Omega}(n)$ between the sensitivity and the error achievable in two-party model. The hard distributions are simply concatenations of $\calD^0$ or $\calD^1$.
Specifically, tor $T \in \N$, we write $\calD^{0, T}$ (respectively, $\calD^{1, T}$) to denote the distribution of $((x_1, \dots, x_{DT}), (y_1, \dots, y_{DT}))$ where $((x_{(i - 1)D + 1}, \dots, x_{iD}), (y_{(i - 1)D + 1}, \dots, y_{iD}))$ is an i.i.d. sample from $\calD^0$ (respectively, $\calD^{1, T}$) for all $i \in [T]$. Similar to before, it is hard to distinguish the two distributions:

\begin{lemma} \label{lem:two-party-hybrid}
Let $P$ be any $(\eps, \delta)$-$\twopartyDP$ protocol. Then, we have
\begin{align*}
\|\VIEW^B_P(\calD^{0, T}) - \VIEW^B_P(\calD^{1, T})\|_{TV} \leq O\left(T D\delta\right) + T \cdot e^{-\Omega_{\epsilon}(D)}.
\end{align*}
\end{lemma}

\begin{proof}
We prove this via a simple hybrid argument. For $j \in [T + 1]$, let us denote by $\calD_j$ the distribution of $((x_1, \dots, x_{DT}), (y_1, \dots, y_{DT}))$ where, for all $i \in [T]$, $((x_{(i - 1)D + 1}, \dots, x_{iD}), (y_{(i - 1)D + 1}, \dots, y_{iD}))$ is independent from $\calD^1$ if $i < j$ and from $\calD^0$ otherwise. Notice that $\calD_1 = \calD^{0, T}$ and $\calD_{T + 1} = \calD^{1, T}$.

Our main claim is the following.
For every $j \in [T]$ and any $(\eps, \delta)$-$\twopartyDP$ protocol $P$,
\begin{align} \label{eq:two-party-hybrid-one-step}
\|\VIEW^B_P(\calD_j) - \VIEW^B_P(\calD_{j + 1})\|_{TV} \leq O\left(D\delta\right) +  e^{-\Omega_{\epsilon}(D)}.
\end{align}
Note that summing~\eqref{eq:two-party-hybrid-one-step} over all $j \in [T]$ immediately yields Lemma~\ref{lem:two-party-hybrid}. 

We will now prove~\eqref{eq:two-party-hybrid-one-step}.
Given an $(\eps, \delta)$-$\twopartyDP$ protocol $P$ (where each party's input has $DT$ bits), we construct a protocol $P'$ (where each party's input has $D$ bits) as follows:
\begin{itemize}
\item Suppose that the input of $A$ is $x'_1, \dots, x'_D$, and the input of $B$ is $y'_1, \dots, y'_D$.
\item For $i = 1, \dots, j - 1$, $A$ samples $((x_{(i - 1)D + 1}, \dots, x_{iD}), (y_{(i - 1)D + 1}, \dots, y_{iD}))$ from $\calD^1$ and sends $(y_{(i - 1)D + 1}, \dots, y_{iD})$ to $B$.
\item $i = j + 1, \dots, T$, $A$ samples $((x_{(i - 1)D + 1}, \dots, x_{iD}), (y_{(i - 1)D + 1}, \dots, y_{iD}))$ from $\calD^0$ and sends $(y_{(i - 1)D + 1}, \dots, y_{iD})$ to $B$.
\item $A$ sets $(x_{(j - 1)D + 1}, \dots, x_{jD}) = (x'_1,\dots,x'_D)$.
\item $B$ sets $(y_{(j - 1)D + 1}, \dots, y_{jD}) = (y'_1,\dots,y'_D)$.
\item $A$ and $B$ then run the protocol $P$ on $((x_1, \dots, x_{DT}), (y_1, \dots, y_{DT}))$.
\end{itemize}
It is clear that $P'$ is $(\eps, \delta)$-$\twopartyDP$ and that
\begin{align*}
\|\VIEW^B_{P'}(\calD^{0, T}) - \VIEW^B_{P'}(\calD^{1, T})\|_{TV} \geq \|\VIEW^B_P(\calD_j) - \VIEW^B_P(\calD_{j + 1})\|_{TV}.
\end{align*}
Inequality \eqref{eq:two-party-hybrid-one-step} then follows from Corollary~\ref{cor:dot-prod-dist}.
\end{proof}

We can now prove our main theorem of this section.

\begin{proof}[Proof of Theorem~\ref{thm:two-party-lb}]
Let $C > 0$ be a sufficiently large constant to be chosen later. Let $D = \lceil C \log n \rceil$ and $T = \lfloor n / D \rfloor$. We may define $f$ on only $2DT$ bits, as it can be trivially extended to $2n$ bits by ignoring the last $n - DT$ bits of $X, Y$. Let
\begin{align*}
f(x_1, \dots, x_{DT}, y_1, \dots, y_{DT}) = \sum_{i \in [T]} \left(\sum_{\ell \in [D]} x_{(i - 1)D + j}y_{(i - 1)D + j} \mod 2\right),
\end{align*}
where the outer summation is over $\Z$.

It is immediate that the sensitivity of $f$ is one. We will now argue that any $(\eps, \delta)$-$\twopartyDP$ protocol $P$ with $\delta = o(1/n)$ incurs error at least $\Omega(n / \log n)$. Since the function is symmetric with respect to the two parties, it suffices without loss of generality to show that the output of $B$ incurs error $\Omega(n / \log n)$ with probability 0.1. To do so, we start by observing that we have $f(X, Y) = 0$ for any $(X, Y) \in \supp(\calD^{0, T})$ whereas $f(X, Y) = T$ for any $(X, Y) \in \supp(\calD^{1, T})$. From Lemma~\ref{lem:two-party-hybrid}, we have that
\begin{align*}
\|\VIEW^B_P(\calD^{0, T}) - \VIEW^B_P(\calD^{1, T})\|_{TV} \leq O\left(T D\delta\right) + T \cdot e^{-\Omega_{\epsilon}(D)}.
\end{align*}
As a result, if we sample $(X, Y)$ from $\calD^{0, T}$ with probability 1/2 and $\calD^{1, T}$ with probability 1/2, then the probability that $B$'s output incurs error at least $T / 2$ is at least
\begin{align*}
\frac{1}{2} - O\left(T D\delta\right) - T \cdot e^{-\Omega_{\epsilon}(D)} \geq \frac{1}{2} - o(1) - n \cdot e^{-\Omega_{\epsilon}(C \log n)}.
\end{align*}
When $C$ is sufficiently large, we also have that $n \cdot e^{-\Omega_{\epsilon}(C \log n)} = o(1)$. As a result, with probability $1/2 - o(1)$ (which is at least $0.1$ for any sufficiently large $n$), the protocol $P$ must incur an error of at least $T / 2 = \Omega(n / \log n)$.
\end{proof}

\subsection{Symmetrization}
\label{subsec:two-party-symmetrization}

Notice that the function in Theorem~\ref{thm:two-party-lb} is asymmetric. It is a natural question to ask whether we can get a similar lower bound for a symmetric function. In this subsection, we give a simple reduction that positively answers this question, ultimately yielding the following:

\begin{theorem} \label{thm:two-party-lb-symmetric}
For any $\eps = O(1)$ and any sufficiently large $n \in \N$, there is a symmetric function $f: [2n]^{2n} \to \R$ whose sensitivity is one and such that any $(\eps, o(1/n))$-$\twopartyDP$ protocol cannot compute $f$ to within an error of $o(n/\log n)$.
\end{theorem}

We remark that the input to each user comes from a set $\calX$ of size $\Omega(n)$, instead of $\calX = \{0, 1\}$ as in Theorem~\ref{thm:two-party-lb}. This larger value of $|\calX|$ turns out to be necessary for symmetric functions: when $f$ is symmetric, we may use the Laplace mechanism from both sides to estimate the histogram of the input, which we can then use to compute $f$. If the sensitivity of $f$ is $O(1)$, this algorithm incurs an error of $O_{\eps}(|\calX|)$. Hence, to achieve a lower bound of $\tilde{\Omega}(n)$, we need $|\calX|$ to be at least $\tilde{\Omega}_{\eps}(n)$.

The properties of our reduction are summarized in the following lemma, which combined with Theorem~\ref{thm:two-party-lb}, immediately implies Theorem~\ref{thm:two-party-lb-symmetric}.

\begin{lemma} \label{lem:symmetrization}
For any function $g\colon \calX^{2n} \to \R$, there is another function $f: (\calX \times [n])^{2n} \to \R$ such that the following holds:
\begin{itemize}
\item The sensitivity of $f$ is no more than that of $g$.
\item If there exists an $(\eps, \delta)$-$\twopartyDP$ protocol that solves $f$ with error $\beta$, then there exists an $(\eps, \delta)$-$\twopartyDP$ protocol that solves $f$ with error $2\beta$.
\end{itemize}
\end{lemma}

The idea behind the proof of Lemma~\ref{lem:symmetrization} is simple. Roughly speaking, we view each input $(x, i) \in \calX \times [n]$ of $f$ as ``setting the $i$th position to $x$'' for the input to $g$. This is formalized below.

\begin{proof}[Proof of Lemma~\ref{lem:symmetrization}]
We start by defining $f$. Let $u^*$ be an arbitrary element of $\calX$. For every $i \in [n]$, we define $h_i: (\calX \times [n])^{n}$ where
\begin{align*}
h_i((w_1, \dots, w_n)) = 
\begin{cases}
\text{the unique } x \text{ such that } \exists j \in [n],  w_j = (x, i) & \text{ if } |\{x \in \calX \mid \exists j \in [n], w_j = (x, i)\}| = 1, \\
u^* & \text{ otherwise.}
\end{cases}
\end{align*}
We now define $f$ by
\begin{align*}
f(W, V) = \frac{1}{2} \cdot g(h_1(W), \dots, h_n(W), h_1(V), \cdots, h_n(V)).
\end{align*}

We will next verify that the two properties hold.
\begin{itemize}
\item Notice that changing any user's input in $f$ results in at most two changes in the user's input of $g$. As a result, the sensitivity of $f$ is no more than that of $g$.
\item Suppose that there exists an $(\eps, \delta)$-$\twopartyDP$ protocol $P$ that solves $f$ with error $\beta$. Let $P'$ be the protocol for $g$ where $A,B$ transform their inputs $(x_1, \dots, x_n)$, $(y_1, \dots, y_n)$ to $((x_1, 1), \dots, (x_n, n))$, $((y_1, 1), \dots, (y_n, n))$ respectively, then run $P$, and finally return the output of $P$ multiplied by two. It is obvious that $P'$ is $(\eps, \delta)$-$\twopartyDP$; furthermore, since the protocol $P$ incurs error $\beta$, the protocol $P'$ incurs error $2\beta$ as desired. \qedhere
\end{itemize}
\end{proof}
	
	\section*{Acknowledgments}
	We would like to thank Noah Golowich for numerous enlightening discussions about lower bounds in the multi-message $\shuffledDP$ model, and for helpful feedback.
	
	\bibliographystyle{alpha}
	\bibliography{main}
	
	\appendix

\section{Total Variance Bound between Mixtures of Multi-dimensional Poisson Distributions}\label{app:TV-Poi}
In this section we prove Lemma~\ref{lm:key-bound-2} (restated below).

\begin{reminder}{Lemma~\ref{lm:key-bound-2}.}
	Let $U,V$ be two random variables supported on $[0,\Lambda]$ such that $\Ex[U^j] = \Ex[V^j]$ for all $j \in \{1,2,\dotsc,L\}$, where $L \ge 1$. Let $D \in \mathbb{N}$ and $\vtheta,\vlambda \in (\R^{\ge 0})^{D}$ such that $\|\vtheta\|_1 = 1$. Let $\calD_{\vtheta}$ be the distribution over $[D]$ corresponding to $\vtheta$. Suppose that
	\[
	\Pr_{i \leftarrow \calD_{\vtheta}} [\vlambda_i \ge 2 \Lambda^2 \cdot \vtheta_i] \ge 1 -  \frac{1}{2\Lambda}.
	\]
	Then,
	\[
	\|\Ex[\vPoi(U \vtheta + \vlambda)] - \Ex[\vPoi(V \vtheta + \vlambda)] \|_{TV}^2 \le \frac{1}{L!}.
	\]
\end{reminder}

To prove Lemma~\ref{lm:key-bound-2}, we begin with some notation. Let $D \in \N$. For vectors $\vm \in \mathbb{N}^{D}$ and $\vlambda \in \R^{D}$, we let 
\[
\vm! := \prod_{i=1}^{D} m_i!~~\text{ and }~~\vlambda^{\vm} := \prod_{i=1}^{D} (\vlambda_i)^{\vm_i}.
\]

We are going to apply the moment-matching technique~\cite{WY16polyapprox,JiaoHW18,WY2019chebyshev} for bounding total variance between mixtures of (single-dimensional) Poisson distributions. The following lemma is a direct generalization of the Theorem~4 of~\cite{HanLec19} to mixtures of multi-dimensional Poisson distributions. We will use the convention that $0^0 = 1$.

\begin{lemma}\label{lm:multi-moment-matching}
	For $\vlambda \in \R^{D}$ and two distributions $\vU$ and $\vV$ supported on $\prod_{i=1}^{D} [-\lambda_i, \infty]$, we have that\[
	\| \Ex[\vPoi(\vU + \vlambda)] - \Ex[\vPoi(\vV + \vlambda)] \|_{TV}^2 \le \sum_{\vm \in (\mathbb{Z}^{\ge 0})^{D}} \frac{ \left( \Ex[\vU^{\vm}] - \Ex[\vV^{\vm}] \right)^2 }{\vm! \cdot \vlambda^{\vm}}.
	\]
\end{lemma}

In order to prove the above lemma, we need to use the Charlier polynomial $c_m(x;\lambda)$. The explicit definition of $c_m(x;\lambda)$ is not important here; we simply list two important properties of this polynomial family~\cite{peccati2011some}:

\begin{prop}\label{prop:Charlier-facts}
	Let $\lambda\in \R$ and $u \in [-\lambda,\infty]$, the following hold:
	
	\begin{enumerate}
		\item We have that
		\[
		\Ex_{X \leftarrow \Poi(\lambda)}[c_m(X;\lambda) c_n(X;\lambda)] = \frac{n!}{\lambda^n} \cdot \mathbb{1}[m = n].
		\]
		
		\item For all $z \in \mathbb{Z}^{\ge 0}$,
		\[
		\frac{\Poi(\lambda + u)_z}{\Poi(\lambda)_z} = e^{-u} \cdot \left( 1 + \frac{u}{\lambda} \right)^{z} = \sum_{m=0}^{\infty} c_m(z;\lambda) \cdot \frac{u^m}{m!}.
		\]
	\end{enumerate}
\end{prop}

We now prove Lemma~\ref{lm:multi-moment-matching}. Our proof closely follows the proof of Theorem~4 of~\cite{HanLec19}.

\begin{proofof}{Lemma~\ref{lm:multi-moment-matching}}
	Let $\Delta := \| \Ex[\vPoi(\vU + \vlambda)] - \Ex[\vPoi(\vV + \vlambda)] \|_{TV}$. We have that
	\begin{align*}
	\Delta 
	&= \frac{1}{2} \cdot \sum_{\vz \in (\mathbb{Z}^{\ge 0})^{D}} \left| \Ex_{\vu \leftarrow \vU} \vPoi(\vu + \vlambda)_{\vz} - \Ex_{\vu \leftarrow \vV} \vPoi(\vu + \vlambda)_{\vz} \right|\\
	&\le\Ex_{\vz \leftarrow \vPoi(\vlambda)} \left| \Ex_{\vu \leftarrow \vU} \frac{\vPoi(\vu + \vlambda)_{\vz}}{\vPoi(\vlambda)_{\vz}} - \Ex_{\vu \leftarrow \vV} \frac{\vPoi(\vu + \vlambda)_{\vz}}{\vPoi(\vlambda)_{\vz}} \right|\\
	%\overset{\text{Item~(2) of Proposition~\ref{prop:Charlier-facts}}}&{=} 
	&= \Ex_{\vz \leftarrow \vPoi(\vlambda)} \left| \sum_{\vm \in (\mathbb{Z}^{\ge 0})^{D}} \prod_{i=1}^{D} c_{\vm_i}(\vz_i;\vlambda_i) \cdot \frac{\Ex[\vU^{\vm}] - \Ex[\vV^{\vm}]}{\vm!} \right|,
	\end{align*}
	where the last equality follows from Item~(2) of Proposition~\ref{prop:Charlier-facts}.
	
	Applying the Cauchy-Schwarz inequality, we get that
	\begin{align*}
	\Delta^2 &\le \Ex_{\vz \leftarrow \vPoi(\vlambda)} \left| \sum_{\vm \in (\mathbb{Z}^{\ge 0})^{D}} \prod_{i=1}^{D} c_{\vm_i}(\vz_i;\vlambda_i) \cdot \frac{\Ex[\vU^{\vm}] - \Ex[\vV^{\vm}]}{\vm!} \right|^2\\
	&= \Ex_{\vz \leftarrow \vPoi(\vlambda)} \sum_{\vm,\vm' \in (\mathbb{Z}^{\ge 0})^{D}} \prod_{i=1}^{D} c_{\vm_i}(\vz_i;\vlambda_i) c_{\vm'_i}(\vz_i;\vlambda_i) \cdot \frac{\Ex[\vU^{\vm}] - \Ex[\vV^{\vm}]}{\vm!} \cdot \frac{\Ex[\vU^{\vm'}] - \Ex[\vV^{\vm'}]}{\vm'!} \\
	&= \sum_{\vm,\vm' \in (\mathbb{Z}^{\ge 0})^{D}} \frac{\Ex[\vU^{\vm}] - \Ex[\vV^{\vm}]}{\vm!} \cdot \frac{\Ex[\vU^{\vm'}] - \Ex[\vV^{\vm'}]}{\vm'!} \Ex_{\vz \leftarrow \vPoi(\vlambda)} \prod_{i=1}^{D} c_{\vm_i}(\vz_i;\vlambda_i) c_{\vm'_i}(\vz_i;\vlambda_i)\\
	&= \sum_{\vm,\vm' \in (\mathbb{Z}^{\ge 0})^{D}} \frac{\Ex[\vU^{\vm}] - \Ex[\vV^{\vm}]}{\vm!} \cdot \frac{\Ex[\vU^{\vm'}] - \Ex[\vV^{\vm'}]}{\vm'!} \prod_{i=1}^{D} \Ex_{\vz_i \leftarrow \Poi(\vlambda_i)} c_{\vm_i}(\vz_i;\vlambda_i) c_{\vm'_i}(\vz_i;\vlambda_i)\\
	%\overset{\text{Item~(1) of Proposition~\ref{prop:Charlier-facts}}}&{=} 
	&= \sum_{\vm \in (\mathbb{Z}^{\ge 0})^{D}} \left(\frac{\Ex[\vU^{\vm}] - \Ex[\vV^{\vm}]}{\vm!}\right)^2 \cdot \prod_{i=1}^{D} \frac{(\vm_i)!}{\vlambda_i^{m_i}}\\ 
	&= \sum_{\vm \in (\mathbb{Z}^{\ge 0})^{D}} \frac{ \left( \Ex[\vU^{\vm}] - \Ex[\vV^{\vm}] \right)^2 }{\vm! \cdot \vlambda^{\vm}},
	\end{align*}
	where the penultimate equality follows from Item~(1) of Proposition~\ref{prop:Charlier-facts}.
\end{proofof}

Applying Lemma~\ref{lm:multi-moment-matching}, the next lemma follows from a straightforward calculation.

\begin{lemma}\label{lm:key-bound-1}
	Let $U,V$ be two random variables supported on $[0,\Lambda]$ such that $\Ex[U^j] = \Ex[V^j]$ for all $j \in \{1,2,\dotsc,L\}$, where $L \ge 1$. For $\vtheta,\vlambda \in (\R^{\ge 0})^{D}$, let $\valpha = \frac{\vtheta^{2}}{\Lambda \vtheta + \vlambda}$ (division here is coordinate-wise, and $\vtheta^2$ denotes taking coordinate-wise square of $\vtheta$). The following holds:
	\[
	\| \Ex[\vPoi(U \vtheta + \vlambda)] - \Ex[\vPoi(V \vtheta + \vlambda)] \|_{TV}^2 \le \sum_{z = L+1}^{\infty} \frac{(\Lambda^2 \cdot \|\valpha\|_1)^z}{z!}.
	\]
\end{lemma}
\begin{proof}

Let $\Delta := \| \Ex[\vPoi(U \vtheta + \vlambda)] - \Ex[\vPoi(V \vtheta + \vlambda)] \|_{TV}$. We also set $\vU = (U-\Lambda) \vtheta$, $\vlambda' =(\Lambda \vtheta + \vlambda)$ and $\vV = (V-\Lambda) \vtheta$. Note that for every $i \in [D]$, we have $\vU_i \ge (-\Lambda \vtheta)_i \ge -(\Lambda \vtheta + \vlambda)_i = -\vlambda'_i$, and the same holds for each $\vV_i$ as well. Hence, we can apply Lemma~\ref{lm:multi-moment-matching} to bound $\Delta^2$ as follows:

\begin{align*}
\Delta^2 &= \| \Ex[\vPoi( \vU +\vlambda')] - \Ex[\vPoi( \vV + \vlambda' )] \|_{TV}^2 \\
%\overset{\text{Lemma~\ref{lm:multi-moment-matching}}}&{\le}
&\le \sum_{\vm \in (\mathbb{Z}^{\ge 0})^{D}} \frac{ \left( \Ex[\vU^{\vm}] - \Ex[\vV^{\vm}] \right)^2 }{\vm! \cdot (\vlambda')^{\vm} }\\
&\le
\sum_{\vm \in (\mathbb{Z}^{\ge 0})^{D}} \frac{ \left( \Ex[((U-\Lambda)\vtheta)^{\vm}] - \Ex[((U-\Lambda)\vtheta)^{\vm}] \right)^2 }{\vm! \cdot (\Lambda \vtheta + \vlambda)^{\vm} }\\
&\le 
\sum_{\vm \in (\mathbb{Z}^{\ge 0})^{D}} \frac{ \left( \vtheta^{\vm} \cdot \Ex[((U-\Lambda))^{\|\vm\|_1}] - \vtheta^{\vm} \cdot \Ex[(V-\Lambda)^{|\vm\|_1}] \right)^2 }{\vm! \cdot (\Lambda \vtheta + \vlambda)^{\vm} }\\
&\le
\sum_{z = L+1}^{\infty} \sum_{\vm \in (\mathbb{Z}^{\ge 0})^{D} \text{s.t.} \|\vm\|_1 = z} \frac{\left( \vtheta^{\vm} \cdot (\Ex[(U-\Lambda)^{z}] - \Ex[(V-\Lambda)^{z}]) \right)^2}{\vm! \cdot (\Lambda \vtheta + \vlambda)^{\vm}}\\
&\le
\sum_{z = L+1}^{\infty} \Lambda^{2z} \sum_{\vm \in (\mathbb{Z}^{\ge 0})^{D} \text{s.t.} \|\vm\|_1 = z} \frac{\vtheta^{2\vm}}{\vm! \cdot (\Lambda \vtheta + \vlambda)^{\vm}},
\end{align*}
where the first inequality follows from Lemma~\ref{lm:multi-moment-matching}.

Now, recall that $\valpha = \frac{\vtheta^{2}}{\Lambda \vtheta + \vlambda}$. We claim that
\[
\sum_{\vm \in (\mathbb{Z}^{\ge 0})^{D} \text{s.t.} \|\vm\|_1 = z} \frac{\valpha^{\vm}}{\vm!} = \frac{\|\valpha\|_1^{z}}{z!}.
\]

To prove this equality, consider the random process of drawing $z$ samples from $[D]$ using the distribution corresponding to $\valpha / \|\valpha\|_1$ (that is, we get $i \in [D]$ with probability $\frac{\valpha_i}{\|\valpha\|_1}$. It is a well-defined distribution since $\valpha \in (\R^{\ge 0})^{D}$). Let $\vec{M}$ be the random variable corresponding to the histogram of the $z$ samples (that is, $\vec{M}_i$ denotes the number of occurrences of the element $i$). For $\vm \in (\mathbb{Z}^{\ge 0})^{D} \text{s.t.} \|\vm\|_1 = z$, we have that
\[
\Pr[\vec{M} = \vm] = \left(\frac{\valpha}{\|\valpha\|_1} \right)^{\vm} \cdot \frac{z!}{\vm!}.
\]
Hence, we get
\[
\sum_{\vm \in (\mathbb{Z}^{\ge 0})^{D} \text{s.t.} \|\vm\|_1 = z} \Pr[\vec{M} = \vm] = 1,
\]
and
\[
\sum_{\vm \in (\mathbb{Z}^{\ge 0})^{D} \text{s.t.} \|\vm\|_1 = z} \frac{\valpha^{\vm}}{\vm!} = \frac{\|\valpha\|_1^{z}}{z!}.
\]
Plugging in, we obtain
\[
\Delta^2 \le \sum_{z = L+1}^{\infty} \Lambda^{2z} \cdot \frac{\|\valpha\|_1^{z}}{z!} = \sum_{z = L+1}^{\infty} \frac{(\Lambda^2 \cdot \|\valpha\|_1)^z}{z!}.\qedhere
\]
\end{proof}

Applying Lemma~\ref{lm:key-bound-1}, we are now ready to prove Lemma~\ref{lm:key-bound-2}.

\begin{proofof}{Lemma~\ref{lm:key-bound-2}}
	Let $\valpha = \frac{\vtheta^{2}}{\Lambda \vtheta + \vlambda}$. We have that
	\begin{align*}
	\|\valpha\|_1 &= \sum_{i \in [D]} \vtheta_i \cdot \frac{\vtheta_i}{\Lambda \vtheta_i + \vlambda_i}\\
				  &= \Ex_{i \leftarrow \calD_{\vtheta}} \frac{\vtheta_i}{\Lambda \vtheta_i + \vlambda_i}\\
				  &\le \frac{1}{2 \Lambda^2} + \Pr_{i \leftarrow \calD_{\vtheta}} [\Lambda \vtheta_i + \vlambda_i < 2 \Lambda^2 \cdot \vtheta_i] \cdot \frac{1}{\Lambda} \\
				  &\le \frac{1}{\Lambda^2}.
	\end{align*}
	Applying Lemma~\ref{lm:key-bound-1}, we get
	\[
	\Ex[\vPoi(U \vtheta + \vlambda)] - \Ex[\vPoi(V \vtheta + \vlambda)] \|_{TV}^2 \le \sum_{z = L+1}^{\infty} \frac{1}{z!} \le \frac{1}{L!}. \qedhere
	\]
\end{proofof}
	\section{Lower Bounds on Hockey Stick Divergence}\label{app:missing-proof-HS}

In this section, we prove Lemma~\ref{lm:key-HS-lowb} (restated below).

\begin{reminder}{Lemma~\ref{lm:key-HS-lowb}.}
	There exists an absolute constant $c_0$ such that, for every integer $m \ge 1$, three reals $\alpha,\beta,\eps > 0$ such that $\alpha > e^{\eps}\beta$, letting $\Delta = \alpha - e^{\eps}\beta$ and supposing $4 \frac{e^{\eps}}{\Delta} \beta < 1/2$, it holds that
	\[
	d_{\eps}(\Ber(\alpha) + \Bin(m,\beta) || \Ber(\beta) + \Bin(m,\beta)) \ge \Delta \cdot \frac{1}{2\sqrt{2m}} \cdot \exp\left(-c_0 \cdot m \cdot \frac{e^{\eps}}{\Delta} \beta \cdot \left[\log(\Delta^{-1}) + 1\right]\right).
	\]
\end{reminder}

Before proving Lemma~\ref{lm:key-HS-lowb}, we need several technical lemmas. First, we show the hockey stick divergence between $\Ber(\alpha) + X$ and $\Ber(\beta) + X$ can be characterized by the hockey sticky divergence between $X + 1$ and $X$.

\begin{lemma}\label{lm:beralpha-and-berbeta}
	Let $\alpha,\beta,\eps > 0$ be three reals such that $\alpha > e^{\eps}\beta$, and $X$ be a random variable over $\mathbb{Z}^{\ge 0}$. The following holds:
	\[
	d_{\eps}(\Ber(\alpha) + X || \Ber(\beta) + X) = (\alpha - e^{\eps}\beta) \cdot d_{\ln \tau}(1+X || X),
	\]
	where $\tau = \frac{e^{\eps}-e^{\eps}\beta-1+\alpha}{\alpha - e^{\eps}\beta}$.
\end{lemma}
\begin{proof}
	We have that
	\begin{align*}
	d_{\eps}(\Ber(\alpha) + X || \Ber(\beta) +  X) 
	&= \sum_{k \in \mathbb{Z}^{\ge 0}} \left[ (1-\alpha) X_k + \alpha X_{k-1} - e^{\eps} (1-\beta) X_k - e^{\eps} \beta X_{k-1} \right]_{+}\\
	&= \sum_{k \in \mathbb{Z}^{\ge 0}} \left[ (\alpha - \eps^{\eps}\beta) \cdot X_{k-1} - (e^{\eps}-e^{\eps}\beta-1+\alpha) \cdot X_k \right]_{+}\\
	&= (\alpha - e^{\eps}\beta) \cdot \sum_{k \in \mathbb{Z}^{\ge 0}} \left[X_{k-1} - \frac{e^{\eps}-e^{\eps}\beta-1+\alpha}{\alpha - e^{\eps}\beta} \cdot X_k \right]_{+}\\
	&= (\alpha - e^{\eps}\beta) \cdot d_{\ln \tau}(1+X || X). \qedhere
	\end{align*}
\end{proof}

Next, we need a lemma giving a lower bound on $d_{\eps}(1+X || X)$ for a random variable $X$.

\begin{lemma}\label{lm:Xp1-and-X}
	Let $X$ be a random variable over $\mathbb{Z}^{\ge 0}$ and $\eps > 0$. The following holds:
	\[
	d_{\eps}(1+X || X) \ge \frac{1}{2} \cdot \Pr_{k \leftarrow X}\left[\frac{X_k}{X_{k+1}} \ge 2 e^{\eps} \right].
	\]
\end{lemma}
\begin{proof}
	We have that
	\begin{align*}
	d_{\eps}(1+X || X) =& \sum_{z = 0}^{\infty} [X_{z - 1} - e^{\eps} X_z]_{+}\\
	=& \sum_{z = 0}^{\infty} [X_{z} - e^{\eps} X_{z+1}]_{+}\\
	\ge& \sum_{z = 0}^{\infty} \frac{1}{2} \cdot X_z \cdot \mathbb{1}[X_{z} \ge 2 e^{\eps} X_{z+1}]\\
	=&\frac{1}{2} \Pr_{k \leftarrow X}\left[\frac{X_k}{X_{k+1}} \ge 2 e^{\eps} \right].\qedhere
	\end{align*}
\end{proof}

Applying Lemma~\ref{lm:Xp1-and-X}, we obtain the following lower bound on $d_{\eps}(1 + \Bin(n,p) || \Bin(n,p))$.

\begin{lemma}\label{lm:binp1-and-bin}
	For $n \in \N$ and $p \in (0,0.5)$, $\eps > 0$ such that $4 e^{\eps} p < 1/2$, 
	\[
	d_{\eps}(1 + \Bin(n,p) || \Bin(n,p)) \ge \frac{1}{2\sqrt{2n}} \exp(- n 4 e^{\eps} p \cdot \log(4 e^{\eps}) ).
	\]
\end{lemma}
\begin{proof}
	We have
	\[
	\frac{\Bin(n,p)_k}{\Bin(n,p)_{k+1}} = \frac{1-p}{p} \cdot \frac{k+1}{n-k} \ge \frac{1-p}{p} \cdot \frac{k}{n}.
	\]
	By Lemma~\ref{lm:Xp1-and-X},
	\begin{align*}
	2 \cdot d_{\eps}(1 + \Bin(n,p) || \Bin(n,p))
	&\ge \Pr_{k \leftarrow \Bin(n,p)}\left[\frac{\Bin(n,p)_k}{\Bin(n,p)_{k+1}} \ge 2 e^{\eps} \right]\\
	&\ge \Pr_{k \leftarrow \Bin(n,p)}\left[ \frac{1-p}{p} \cdot \frac{k}{n} \ge 2 e^{\eps} \right]\\
	&=\Pr\left[ \Bin(n,p)\ge 2 e^{\eps} \cdot n \cdot \frac{p}{1-p} \right].\\
	&\ge \Pr\left[ \Bin(n,p)\ge 4 e^{\eps} \cdot pn \right].
	\end{align*}
	
	Now, by anti-concentration of the binomial distribution~\cite[Page 115]{robert1990ash}, we have
	\[
	\Pr\left[ \Bin(n,p)\ge 4 e^{\eps} \cdot pn \right] \ge \frac{1}{\sqrt{2n}} \exp(- n \cdot KL(4 e^{\eps} p || p)).
	\]
	Letting $\lambda = 4 e^{\eps}$, we have
	\[
	KL(\lambda p || p) = \lambda p \cdot \log \frac{\lambda p}{p} + (1 - \lambda p) \cdot \log \frac{1 - \lambda p}{1- p} \le \lambda p \cdot \log \frac{\lambda p}{p} = \lambda p \cdot \log \lambda.
	\]
	Putting everything together, we get
	\[
	d_{\eps}(1 + \Bin(n,p) || \Bin(n,p)) \ge \frac{1}{2\sqrt{2n}} \exp(- n \lambda p \cdot \log \lambda) = 
	\frac{1}{2\sqrt{2n}} \exp(- n 4 e^{\eps} p \cdot \log(4 e^{\eps}) ).\qedhere
	\]
\end{proof}

We are now ready to prove Lemma~\ref{lm:key-HS-lowb}.

\begin{proofof}{Lemma~\ref{lm:key-HS-lowb}}
	Let $\tau = \frac{e^{\eps}-e^{\eps}\beta-1+\alpha}{\alpha - e^{\eps}\beta}$, we have $\tau \le \frac{e^{\eps}}{\Delta}$. Let $N=\Bin(m-1,\beta)$. By Lemma~\ref{lm:beralpha-and-berbeta}, we have that
	\[
	d_{\eps}(\Ber(\alpha) + N || \Ber(\beta) + N) \ge \Delta \cdot d_{\ln \tau}(1 + N || N).
	\]
	Applying Lemma~\ref{lm:binp1-and-bin} and note that $4 \tau \beta \le 4 \frac{e^{\eps}}{\Delta} \beta < 1/2$, it follows that
	\[
	d_{\ln \tau}(1 + N || N) \ge \frac{1}{2\sqrt{2m}} \cdot \exp(-m \cdot 4 \tau \beta \log(4 \tau)).
	\]
	We thus have that
	\[
	m \cdot 4 \tau \beta \log(4 \tau) \le O\left( m \cdot \frac{e^{\eps}}{\Delta} \beta \cdot \left[\log(\Delta^{-1}) + 1\right] \right).\qedhere
	\]
\end{proofof}
	\section{Simulation of Shuffle Protocols by SQ Algorithms}\label{sec:SQ_connection}

In this section, we show the connection between dominated protocols and SQ algorithms, which implies that $\shuffledDP^{k}$ protocols can be simulated by SQ algorithms. This is analogous to the result of Kasiviswanathan et al.~\cite{kasiviswanathan2011can} who proved such a connection between $\localDP$ protocols and SQ algorithms. In the following, we use the notation of~\cite{kasiviswanathan2011can}.

\subsection{SQ Model} 
\newcommand{\SQ}{\textsf{SQ}}

We first introduce the statistical query (SQ) model. 
In the SQ model, algorithms access a distribution through its statistical properties instead of individual samples.

\begin{definition}[SQ Oracle]
	Let $\calD$ be a distribution over $\calX$. An \emph{SQ oracle} $\SQ_{\calD}$ takes as input a function $g\colon D \rightarrow\{-1,1\}$ and a tolerance parameter $\tau \in (0,1)$; it outputs an estimate $v$ such that: $$|v-g(\calD)|\leq \tau.$$
\end{definition}

\begin{definition}[SQ algorithm]
	An {\em SQ algorithm} is an algorithm that accesses the distribution $\calD$ only through the SQ oracle $\SQ_\calD$. 
\end{definition}

\subsection{Simulation of Dominated Algorithms by SQ Algorithms}

We have the following simulation of dominated algorithms by SQ algorithms.

\begin{theorem}\label{theo:dominated-by-SQ}
	Suppose $R\colon \calX \to \calM$ is $(\eps,\delta)$-dominated. Then, for any distribution $\calU$ and error parameter $\beta \ge \delta$, one can take a sample from $R(\calU)$ with statistical error $O(\beta)$ using $O(e^{\eps})$ queries in expectation to $\SQ_{\calU}$ with tolerance $\tau = \beta/e^{\eps}$.
\end{theorem}

\begin{proof}
	Suppose $R$ is $(\eps,\delta)$-dominated by $\calD$. Let $\tau = \beta/e^{\eps}$. For every $x \in \calX$ and $z \in \calM$, we use $p_{x,z}$ (respectively, $p_{x,E}$) to denote $\Pr[R(x) = z]$ (respectively, $\Pr[R(x) \in E]$), 	
	and let 
	$$
	f_z(x) = \frac{p_{x,z}}{e^{\eps} \cdot \calD_z}~~\text{and}~~ g_z(x) = \min(1,f_z(x)).
	$$
	Our algorithm is a rejection sampling procedure adapted from~\cite{kasiviswanathan2011can}. It works as follows:
	
	\begin{enumerate}
		\item Take a sample $z \leftarrow \calD$.
		\item We make a query $g_z$ to $\SQ_{\calU}$ with tolerance level $\tau$, to obtain an estimate $\hat{g}_z$ such that $| \hat{g}_z - g_z(\calU) | \le \tau$.
		\item With probability $\max(\hat{g}_z,0)$, we output $z$ and stop. Otherwise we go back to Step 1.
	\end{enumerate}
	
	Let $\calT_x = \{ z : p_{x,z} > e^{\eps} \cdot \calD_z \}$. Note that $p_{z,\calT_x} \le 2\delta \le 2 \beta$ since $R$ is $(\eps,\delta)$-dominated by $\calD$. By the definition of $g_z$, it holds that $g_z(x) = f_z(x)$ for every $z \not\in \calT_x$. We will need the following claim.
	\begin{claim}\label{claim:bound-diff-fg}
		For every $x \in \calX$,
		\[
		\Ex_{z \leftarrow \calD} \left| f_z(x) - g_z(x) \right| \le 2\beta / e^{\eps}.
		\]
	\end{claim}
	\begin{proof}
		\begin{align*}
		\Ex_{z \leftarrow \calD} \left| f_z(x) - g_z(x) \right| &\le \sum_{z \in \calT_x } \calD_z \cdot f_z(x)\\
		&\le \sum_{z \in \calT_x} p_{x,z} \cdot e^{-\eps}\\
		&\le p_{x,\calT_x} /e^{\eps} \le 2\beta/e^{\eps}.\qedhere
		\end{align*}
	\end{proof}
	
	Now, in a single run, the above algorithm outputs $z$ with probability in the interval
	\[
	[\calD_z \cdot (g_z(\calU) - \tau), \calD_z \cdot (g_z(\calU) + \tau)].
	\]
	
	Note that $\Ex_{z \leftarrow \calD} f_z(\calU) = \Ex_{x \leftarrow \calU} \sum_{z} \frac{p_{x,z}}{e^{\eps}} = e^{-\eps}$. By Claim~\ref{claim:bound-diff-fg}, the algorithm terminates in a single run with probability at least
	\[
	\Ex_{z \leftarrow \calD} (g_z(\calU) - \tau) \ge \left(\Ex_{z \leftarrow \calD} f_z(\calU) \right) - \tau - 2\beta/e^{\eps} = e^{-\eps} \cdot (1 - 3\beta),
	\]
	and at most
	\[
	\Ex_{z \leftarrow \calD} (g_z(\calU) + \tau) \le \left(\Ex_{z \leftarrow \calD} f_z(\calU) \right) + \tau + 2\beta/e^{\eps} = e^{-\eps} \cdot (1 + 3\beta).
	\]
	The above implies that the algorithm makes at most $O(e^{\eps})$ queries to $\SQ_{\calU}$ in expectation.
	
	Putting everything together, our algorithm outputs $z$ with probability in the following interval:
	\[
	I_z := \left[\frac{\calD_z \cdot (g_z(\calU) - \tau) \cdot e^{\eps}}{(1 + 3\beta)},
	\frac{\calD_z \cdot (g_z(\calU) + \tau) \cdot e^{\eps}}{(1 - 3\beta)}\right].
	\]
	We have that
	\[
	\Pr[R(\calU) = z] = \Ex_{x \leftarrow \calU} p_{x,z} = f_z(\calU) \cdot\calD_z \cdot e^{\eps}.
	\]
	Note that
	\begin{align*}
	\max_{v \in I_z} \left|v - f_z(\calU) \cdot\calD_z \cdot e^{\eps} \right| 
	&\le \max_{v \in I_z} \left|v - g_z(\calU) \cdot\calD_z\cdot e^{\eps}\right| + \left|f_z(\calU) \cdot\calD_z\cdot e^{\eps} - g_z(\calU) \cdot\calD_z\cdot e^{\eps} \right|.
	\end{align*}
	Moreover, we have that
	\begin{align*}
	\max_{v \in I_z} \left|v - g_z(\calU) \cdot\calD_z\cdot e^{\eps}\right| &= 
	e^{\eps} \cdot
	\max\left\{ \left| \frac{\calD_z \cdot (g_z(\calU) - \tau)}{ (1 + 3\beta)} - g_z(\calU) \cdot\calD_z \right|, \left| \frac{\calD_z \cdot (g_z(\calU) + \tau)}{(1 - 3\beta)} - g_z(\calU) \cdot\calD_z \right|\right\}\\
	&\le e^{\eps} \cdot g_z(\calU) \cdot\calD_z \cdot O(\beta) + e^{\eps} \cdot \calD_z \cdot O(\tau).
	\end{align*}
	
	The final statistical error of our sampling algorithm is therefore bounded by
	\begin{align*}
	\sum_{z \in \calM} 
	\max_{v \in I_z} \left|v - f_z(\calU) \cdot\calD_z \cdot e^{\eps} \right| 
	&\le e^{\eps} \cdot \sum_{z \in \calM} \left( g_z(\calU) \cdot\calD_z \cdot O(\beta) + \calD_z \cdot O(\tau) + \left|f_z(\calU) \cdot\calD_z - g_z(\calU) \cdot\calD_z \right| \right)\\
	&\le e^{\eps} \cdot \Ex_{z \leftarrow \calD} \Big[ f_z(\calU) \cdot O(\beta) + O(\tau) + |f_z(\calU) - g_z(\calU)| \Big] \tag{$g_z(\calU) \le f_z(\calU)$}\\
	&\le O(\beta).  \tag{$\Ex_{z \leftarrow \calD} f_z(\calU) = e^{-\eps}$ and Claim~\ref{claim:bound-diff-fg}}
	\end{align*}
\end{proof}

\subsection{Applications}

We are now ready to apply Theorem~\ref{theo:dominated-by-SQ} to show that protocols in the $\shuffledDP^k$ model can be simulated by SQ algorithms when the database is drawn i.i.d. from a single distribution.

\begin{theorem}\label{theo:simulation-shuffle-by-SQ}
	Let $z$ be a database with $n$ entries drawn i.i.d. from a distribution $\calU$. Let $P = (R,S,A)$ be an $(\eps,o(1/n))$-$\shuffledDP^{k}$ protocol on $n$ users. Then, there is an algorithm making $O((en)^{k + 1} \cdot e^{\eps})$ queries in expectation to $\SQ_{\calU}$ with tolerance $\tau = \Theta\left(\frac{1}{(en)^{k+1} \cdot e^{\eps}}\right)$, such that its output distribution differs by at most $0.01$ in statistical distance from the output distribution of $P$ on the dataset $z$.
\end{theorem}
\begin{proof}
	Note that it suffices to draw $n$ i.i.d. samples from $R(\calU)$. By Lemma~\ref{lm:SDP-to-dominated-Toy}, we now that $R$ is $(\eps+ k(1 + \ln n),o(1/n))$-dominated. Let $\gamma = e^{\eps+ k(1 + \ln n)} = e^{\eps} \cdot (en)^{k}$. By Theorem~\ref{theo:dominated-by-SQ}, using $O(\gamma)$ queries in expectation to $\SQ_{\calU}$ with tolerance $\tau = \Theta(1/\gamma n)$, we can sample from $R(\calU)$ with statistical error $1/100n$. Taking $n$ such samples completes the proof.
\end{proof}

We remark that~\cite{blum1994weakly} proved that if an SQ algorithm solves \paritylearning with probability at least $0.99$, $T$ queries and tolerance $1/T$, then $T = \Omega(2^{D/3})$ (recall that $D$ is the dimension of the hidden vector in \paritylearning). Combing the foregoing lower bound with Theorem~\ref{theo:simulation-shuffle-by-SQ}, it translates to an $\Omega(2^{D/3(k+1)})$ lower bound on the sample complexity of $(O(1),o(1/n))$-$\shuffledDP^{k}$ protocols solving \paritylearning, which is slightly weaker than our Theorem~\ref{theo:lowb-learning-parity}.

	\section{Upper Bounds for \selection in $\shuffledDP^{k}$}\label{app:up-selection}

In this section, we give a proof sketch for the $\shuffledDP^k$ protocol for \selection with sample complexity $\Tilde{O}(D/\sqrt{k})$.

\begin{theorem}\label{theo:uppb-selection}
	For any $k \le D$, $\eps = O(1)$ and $\delta = 1 / \poly(n)$, there is an $(\eps, \delta)$-$\shuffledDP^{k}$ protocol solving \selection with probability at least $0.99$ and $n = \Tilde{O}(D / \sqrt{k})$.
\end{theorem}
\begin{proofsketch}
    Let $\eps = O(1)$ and $\delta = 1/\poly(n)$ be the privacy parameters. We also let $\eps_0 = \Theta(\eps / \sqrt{k})$, $\delta_0 = 1/\poly(n)$ and $n = \Tilde{\Theta}(D / \sqrt{k})$ to be specified later.
    
    We can assume that $k = (\log n)^{\omega(1)}$, as otherwise the protocol simply follows from the $\Tilde{O}(D)$ sample upper bound by an $(\eps,\delta)$-$\shuffledDP^{1}$ protocol~\cite{GhaziGKPV19}.
    
    Let $m = k / \log^2 n$, and $N = n m / D$. Note that by our choice of $k$ and $D$, $N = (\log n)^{\omega(1)}$.
    
    Now for each $i \in [D]$, our protocol maintains an $(\eps_0,\delta_0)$-DP subprotocol aiming at estimating the fraction of users whose input $x$ satisfies $x_i = 1$. These subprotocols assume they will receive between $0.99N$ and $1.01N$ inputs. By~\cite{ghazi2019privateEurocrypt,balle_merged}, there is such a protocol which achieves error $O(\eps_0^{-1} \log n)$ with probability at least $1-1/n^2$ and using $O\left(\frac{\log (1/\delta)}{\log N}\right) \leq O(\log n)$ messages.
    
    In our protocol, each user selects $k/\log^2 n$ coordinates from $[D]$ uniformly at random, and participates in the corresponding subprotocols. Finally, the analyzer aggregates the outputs of all subprotocols and outputs the coordinate with the highest estimate.
    
    Note that by a union bound a Chernoff bound, it follows that with probability at least $1 - n^{-\omega(1)}$, the number of users of every subprotocol falls in the range $[0.99 N, 1.01 N]$, and their mean is $0.01$-close to the true mean of $i$-th coordinates of all users. 
    
    Setting $\eps_0 = \Theta(\eps / \sqrt{k})$ and $\delta_0 = 1/\poly(n)$ appropriately, the protocol is $(\eps,\delta)$-DP by the advanced composition theorem of DP~\cite{DworkRV10}. Moreover, with probability at least $1-1/n$, all subprotocols obtain estimates with error $O(\eps_0^{-1} \log n)$. 
    
    Setting $n$ so that $\eps_0^{-1} \log n = o\left(N\right)$, our protocol solves \selection with probability at least $0.99$.
    
\end{proofsketch}

\end{document}